\documentclass[11pt]{article}
\usepackage{amsmath,amssymb,amscd}
\usepackage{amsthm,bbm}
\usepackage{mathtools}
\usepackage{float}
\usepackage[caption = false]{subfig}
\usepackage[final]{graphicx}
\usepackage[mathscr]{eucal}
\usepackage{color}
\usepackage{booktabs}
\usepackage{float}
\usepackage{lipsum}
\newtheorem{lemma}{Lemma}
\newtheorem{Proposition}{Proposition}
\newtheorem{Theorem}{Theorem}

\large\normalsize

\title{
\color{black}Evolutionary game with stochastic payoffs in a finite island model \color{black}}
\author{Dhaker Kroumi$^1$\footnote{Author for
correspondence, and e-mail: dhaker.kroumi@kfupm.edu.sa} and Sabin Lessard$^2$
\\$^1$Department of Mathematics and Statistics\\King Fahd University of Petroleum and Minerals\\Dhahran 31261, Saudi Arabia\\
$^2$Department of Mathematics and Statistics\\University of Montreal \\
 Montreal H3C 3J7, Canada\\
 }
\date{}
\begin{document}
\maketitle

%%%%%%%%%%%%%%%%%%%%%%%%%%%%%%%%%%%%%%%%%%%%%%%%%%%%%%%%%%%%%%%%%%%%%%%%%%%%%%%%%%%%%%%%%%%%%%%
%%%%%%%%%%%%%%%%%%%%%%%%%%%%%%%%%%%%%%%%%%%%%%%%%%%%%%%%%%%%%%%%%%%%%%%%%%%%%%%%%%%%%%%%%%%%%%%

\section*{Abstract}
In this paper, we consider a two-player two-strategy game with random payoffs in a population subdivided into $d$ demes, each containing $N$ individuals at the beginning of any given generation and experiencing local extinction and recolonization with some fixed probability $m$ after reproduction and selection among offspring. Within each deme, offspring engage in random pairwise interactions, and the payoffs are assumed to have means and variances proportional to the inverse of the population size. By verifying the conditions given in Ethier and Nagylaki (1980) to approximate Markov chains with two time scales, we establish that the discrete-time evolutionary dynamics with $Nd$ generations as unit of time converges to a continuous-time diffusion as $d\rightarrow\infty$. The infinitesimal mean and variance of this diffusion are expressed in terms of the population-scaled means and variances of the payoffs besides identity-by-descent measures between offspring in the same deme in a neutral population. We show that the probability for a strategy to fix in the population starting from an initial frequency $(Nd)^{-1}$ generally increases as the payoffs to that strategy exhibit less variability or the payoffs to the other strategy more variability. As a result, differences in variability can make this fixation probability for cooperation larger than the corresponding one for defection. As the deme-scaled extinction rate $\nu=mN$ decreases for $N$ large enough and $m$ small enough,  creating a higher level of identity among offspring within demes, the differences between the population-scaled variances of the payoffs for interacting offspring of different types increases this effect to a greater extent than the differences for interacting offspring of the same type. Furthermore, assuming no variability in the payoffs to cooperation, we explore in more details the relationships between the levels of variability in the payoffs to defection as functions of $\nu$ that are necessary for the fixation probability of cooperation to be above $(Nd)^{-1}$, the corresponding fixation probability of defection  to be below $(Nd)^{-1}$, and the former fixation probability to exceed the latter.

%%%%%%%%%%%%%%%%%%%%%%%%%%%%%%%%%%%%%%%%%%%%%%%%%%%%%%%%%%%%%%%%%%%%%%%%%%%%%%%%%%%%%%%%%%%%%%%%%%%%%%%%%%%%%%%%%%%%%%%%%%%%%%%%%%%%%%%%%%%%%%%%%%%%%%%%%%%%%%%%%%%%%%%%%%%%%%%%%%%%%%%%%%

\noindent \textbf{Keywords and phrases}:  Prisoner's Dilemma; Random payoffs; Fixation probability; Identity by descent; Island Model; Diffusion approximation.  

\noindent \textbf{Mathematics Subject Classification (2010)}: Primary 92D25; Secondary 60J70

%%%%%%%%%%%%%%%%%%%%%%%%%%%%%%%%%%%%%%%%%%%%%%%%%%%%%%%%%%%%%%%%%%%%%%%%%%%%%%%%%%%%%%%%%%%%%%%%%%%%%%%%%%%%%%%%%%%%%%%%%%%%%%%%%%%%%%%%%%%%%%%%%%%%%%%%%%%%%%%%%%%%%%%%%%%%%%%%%%%%%%%%%

\section{Introduction}

Evolutionary game dynamics is a process of frequency-dependent selection, in which the reproductive success of an individual, also called its fitness, is determined by the payoff that the individual receives  according to the strategy used by the individual in interactions with others. Individuals with higher fitnesses are more likely to reproduce and pass on their strategies (Maynard Smith \cite{M1982}, Maynard Smith and Price \cite{MP1973}, Hofbauer and Sigmund \cite{HS1988}, Nowak \cite{N2006}). The reproductive success of an individual is not determined by its absolute fitness but rather by its relative fitness compared to the other individuals in the population. This framework can be used to shed light on the emergence of cooperation or competition in natural populations (Axelrod and Hamilton \cite{AH1981}, Sugden \cite{S1986}).

Early studies of evolutionary games have primarily focused on an infinite well-mixed population, using a straightforward differential equation known as the replicator equation which assumes that the growth rate of a strategy is given its average payoff (Taylor and Jonker \cite{TJ1978}, Zeeman \cite{Z1980}). However, real biological species are finite in number, and their dynamics cannot be adequately captured by a deterministic differential equation. In order to address this fact, investigations have turned towards finite well-mixed populations as modeled by finite Markov chains (see, e.g., Nowak \textit{et al.} \cite{NSTF2004}, Taylor and Wild \cite{WT2004}, Imhof and Nowak \cite{IN2006}, Antal \textit{et al.} \cite{ANT2009}). These stochastic processes play a crucial role in comprehending the emergence and maintenance of cooperative behaviors, the influence of random drift, and the impact of population size on evolutionary dynamics.
The Moran model  (Moran \cite{M1958}) assumes that one individual at a time is replaced by another following a birth-death event, while the Wright-Fisher model (Fisher \cite{F1930}, Wright \cite{W1931}) assumes non-overlapping generations with a binomial scheme from one generation to the next. More general exchangeable selection models from one time step to the next extending the neutral Cannings model (Cannings \cite{Cannings1974}) have also been considered (Lessard and Ladret \cite{LL2007}).

Well-mixed populations oversimplify the reality of biological species, which are more complex. Evidence of geographic structure in genetic data (Slatkin \cite{S1985}) highlights the importance of studying populations subdivided into different groups subject to specific environmental and ecological effects.  Evolutionary game dynamics in structured populations explore the dynamics of strategy frequencies among interacting individuals in populations that exhibit spatial or social structure. In such populations, individuals are not randomly mixed but instead interact primarily with their neighbors or with individuals in close geographic proximity. These structural contraints on interactions can significantly influence the evolutionary outcomes such as the emergence of cooperation.

The island model, which serves as a basic representation of a structured population, is commonly used to investigate the impact of limited migration, local adaptation, and population subdivision on the process of evolution. Ethier and Nagylaki \cite{EN1980} have established conditions for a discrete-time Markov chain with two time scales to converge weakly to a diffusion process in the limit of a large population. The strong migration limit as the deme size tends to infinity and the number of demes is kept fixed was studied in
Nagylaki \cite{N1980, N1982, N2000}, while Slatkin \cite{S1981} and Takahata \cite{T1991} considered the island model under conditions of low migration. 
%In Lessard \cite{L2009}, it is the deme size that was kept constant while the number of demes was let to go to infinity, which allows to study the effects of kin selection and group selection.

Numerical simulations have been conducted by Cherry \cite{C2003a,C2003b} and Cherry and Wakeley \cite{CW2003} to evaluate the effectiveness of diffusion methods in scenarios involving a large number of small demes. A similar approach using diffusion methods was employed by Whitlock \cite{W2003} to investigate the stepping-stone model of population structure. This model assumes a population subdivided into demes of fixed size connected by migration. Roze and Rousset \cite{RR2003} have proposed a general method for constructing diffusion approximations in structured population genetics models. They achieved this by using approximations for the expectation and variance in allele frequency change over one generation, based on partial derivatives of fitness functions with respect to the trait value and probabilities of genetic identity under neutrality. Wakeley \cite{W2003}, Wakeley and Takahashi \cite{WT2004}  and Lessard \cite{L2009} have studied an island model consisting of a finite number of demes, all of the same finite size in the limit of a large number of demes, using a diffusion approximation based on the conditions in Ethier and Nagylaki \cite{EN1980}. This is an ideal framework to study the effects of kin selection and group selection.

All the aforementioned studies assume a constant environment over time. However, this assumption is unrealistic as environmental conditions can randomly fluctuate, impacting competition abilities and birth-death rates (Kaplan \textit{et al.} \cite{KHH1990}, Lande \textit{et al.} \cite{LES2003}, May \cite{M1973}). These fluctuations can influence population size and composition, leading researchers to investigate their effects. For instance, Lambert \cite{L2006}, Parsons and Quince \cite{PQ2007a,PQ2007b}, and Otto and Whitlock \cite{OW1997} have studied the fixation probability of a mutant type in an unstructured population with dynamic size fluctuations due to demographic scenarios of growth or decline. Conversely, Uecker and Hermisson \cite{UH2011} have focused on a single beneficial allele in a population that experiences temporal variation in both size and selection pressure.

A random payoff matrix in pairwise interactions provides an alternative approach for examining fluctuations in the environment, reflecting the stochastic nature of real-world interactions. It encompasses variations in the environment, the presence of noise or decision-making errors, and the inherent unpredictability of outcomes. 

Recently, the effect of stochastic changes in payoffs in discrete time were studied with particular attention to stochastic local stability of fixation states and constant polymorphic equilibria in an infinite population (Zheng \emph{et al.} \cite{ZLLT2017, ZLLT2018}), the fixation probability in a large population that reproduces according to a Wright-Fisher model (Li and Lessard \cite{LL2020}) or in a finite well-mixed population that reproduces according to a Moran model (Kroumi et al. \cite{KMLL2021, KML2022}). Moreover,
Kroumi and Lessard \cite{KL2021, KL2022} explored the effects of random payoffs on the average abundance of strategies such as cooperation and defection under recurrent mutation in finite well-mixed populations using the Moran model for reproduction. The general findings reveal that the inclusion of uncertainty in the payoffs received by defectors or the reduction of uncertainty in the payoffs received by cooperators can facilitate the evolution of cooperation.

The objective of this paper is to examine the influence of stochastic variability in a population structured into a large number of isolated demes of the same finite size. Our analysis will focus on exploring how the interplay between the second moments in payoffs and the migration rate affects the evolutionary dynamics of cooperation, specifically its probability of becoming fixed under different scenarios.

The structure of this paper is as follows: Section 2 presents the island model with random pairwise interactions within demes and random payoffs for a two-player game, while Section 3 establishes the diffusion approximation in the limit of a large number of demes based on the two time scales for the changes in strategy frequencies and deme type frequencies. In Sections 4 and 5, we develop the expressions of the infinitesimal mean and variance in the diffusion approximation and we examine conditions for the evolution of a strategy to be favoured by selection based the fixation probability following its introduction as a single mutant. Sections 6-9 investigate these conditions under various scenarios on the variances of the payoffs. Finally, in Section 10, we discuss our results and their relationship with the existing literature. Technical details are relegated to Appendices A-J.

%%%%%%%%%%%%%%%%%%%%%%%%%%%%%%%%%%%%%%%%%%%%%%%%%%%%%%%%%%%%%%%%%%%%%%%%%%%%%%%%%%%%
%%%%%%%%%%%%%%%%%%%%%%%%%%%%%%%%%%%%%%%%%%%%%%%%%%%%%%%%%%%%%%%%%%%%%%%%%%%%%%%%%%%%
\section{Model}

Consider a haploid population that is structured into $d$ isolated demes. Generations are discrete, non-overlapping, and every generation starts with $N$ individuals in each deme. These produce a large number of offspring which is assumed to be the same for  each individual in the population.  This is followed by
random pairwise interactions between the offspring produced in the same deme. 

Each individual in the population can be of type $S_1$ or type $S_2$, and an offspring inherits the type of its parent. In the case where an offspring of type $S_1$ interacts with another offspring of type $S_1$, both offspring  receive a payoff  $a_{1,1}$. If an offspring  of type $S_1$ interacts with an offspring  of type $S_2$, the former receives a payoff $a_{1,2}$ while the latter receives $a_{2,1}$. Finally, when an offspring of type $S_2$ interacts with another offspring  of type $S_2$, both offspring  receive a payoff $a_{2,2}$.

The payoffs are given by the entries of the game matrix
\begin{equation}\label{sec1-eq1}
\bordermatrix {
& S_1 & S_2 \cr
S_1& a_{1,1} & a_{1,2} \cr
S_2 & a_{2,1} &a_{2,2} \cr}=(a_{k,l}).
\end{equation}
 In this paper, we suppose that the payoffs randomly fluctuate over successive generations in such a way that $|a_{k,l}|<1$ and
\begin{subequations}\label{sec1-eq2}
\begin{align}
&E\left[a_{k,l}\right]=\frac{\mu_{k,l}}{Nd}+o\left(\frac{1}{d}\right),\\
&E\left[a_{k,l}^2\right]=\frac{\sigma^2_{k,l}}{Nd}+o\left(\frac{1}{d}\right),\\
&E\left[|a_{k,l}|^n\right]=o\left(\frac{1}{d}\right),
\end{align}
\end{subequations}
for $k,l=1,2$ and $n\geq 3$. The parameters $\mu_{k,l}$ and $\sigma^2_{k,l}$ are the population-scaled mean and variance of the payoffs, respectively, for $k,l=1,2$. In addition, we suppose that $a_{k_1,l_1}$ and $a_{k_2,l_2}$ are independent random variables as long as $(k_1,l_1)\not=(k_2,l_2)$.
%In addition, we suppose that 
%\begin{equation}\label{sec1-eq3}
%E\left[\prod_{i,j=1}^2 |a_{i,j}|^{r_{i,j}}\right]=o\left(\frac{1}{d}\right)
%\end{equation}
%for any integers $r_{i,j} \geq 0$ such that $\sum_{i,j=1}^2r_{i,j}\geq3$. 

Note that the payoff received by an offspring $I$ can be written as  
\begin{equation}\label{sec1-eq4}
a_{I}=\sum_{k,l=1}^{2}a_{k,l}q_{k,I}q_{l,I'},
\end{equation}
where $I'$ represents a randomly chosen offspring produced in the same deme as $I$, while $q_{k,l}$ and $q_{l,I'}$ denote the frequencies ($0$ or $1$) of $S_k$ and $S_l$, respectively, in $I$ and $I'$, respectively, for $k,l=1, 2$. This payoff is added to the reference value $1$ to obtain the viability of $I$, given by
\begin{equation}\label{sec1-eq5}
\omega_{I}=1+a_{I}.
\end{equation}
This quantity, always positive, is proportional to the survival probability of $I$ from conception to maturity. 

Note that the average viability of all offspring produced by an individual can be interpreted as the reproductive success, or fitness, of this individual. Moreover, the average viability of all offspring produced in a deme is proportional to the number of offspring in the deme after viability selection. It is assumed that a fixed proportion of these offspring go to a migrant pool, while the other proportion stay in the deme where they were produced. 
%Then, each individual will produce a number of offspring proportionally to its fertility rate. 

%This model is equivalent to the following model, where in each generation, any individual produces the same large number of offspring, and the survival of these offspring is determined according to the payoffs received by their parents. Each offspring
%will be of the same strategy as its parent and its viability coefficient is $\omega=1+a$, where $a$ is the payoff received by the parent. 
%

On the other hand, each deme goes to extinction with a fixed probability $m\in (0, 1)$ and, when this occurs, the deme is repopulated by a large number of offspring coming from the migrant pool. Note that the contributions of the demes to the migrant pool being proportional to the average viabilities within demes, this is a case of hard selection (Christiansen \cite{C1975}). Moreover, the limiting case $m=1$ corresponds to a well-mixed population as in Hilbe \cite{H2011}.

Finally, the individuals to start the next generation are obtained by sampling at random $N$ offspring within each deme.
Note that the number of individuals of type $S_1$ to start a deme can vary from $0$ to $N$, leading to a total of $N+1$ possible deme states. A deme is considered to be of type $i$ if it starts with $i$ individuals of type $S_1$ and $N-i$ individuals of type $S_2$, for $i=0,1,\ldots,N$. Note that the initial frequency of $S_1$ in a deme of type $i$ is given by
\begin{align}
x_i=\frac{i}{N},
\end{align}
and this is also the frequency of $S_1$ among all offspring produced in the deme, for $i=0,1,\ldots,N$.

Let $Z_i(t)$ be the fraction of demes of type $i$ at the beginning of generation $t\geq0$. Then, the population state is described by the frequency vector $\mathbf{Z}(t)=(Z_0(t),Z_1(t),\ldots,Z_N(t))$, where $\sum_{i=0}^{N}Z_i(t)=1$.
Note that
\begin{equation}\label{sec1-eq6}
X(t)=\sum_{i=0}^{N}Z_i(t)x_i
\end{equation}
is the frequency of type $S_1$ in the whole population at the beginning of generation $t\geq0$. 

Suppose that the initial population state is $\mathbf{Z}(0)=\mathbf{z}=(z_0,z_1,\ldots,z_N)$ with corresponding frequency of type $S_1$ in the population given by
\begin{equation}\label{sec1-eq7}
X(0)=x=\sum_{i=0}^{N}z_ix_i.
\end{equation} 
%Note that 
%\begin{equation}\label{sec1-eq8}
%x=E_{\mathbf{z}}\left( q_{1,I}\right),
%\end{equation} 
%where $I$ is an individual chosen at random from the population at the beginning of generation $0$.
Following reproduction and random pairwise interactions, the  average payoffs received by an offspring of type $S_k$ and by offspring of all types in a deme of type $i$ are 
 \begin{subequations}\label{sec1-eq9}
 \begin{align}
\tilde{a}_{k,i}&=x_ia_{k,1}+\left(1-x_i\right)a_{k,2},\\
\tilde{a}_{i}&=x_i\tilde{a}_{1,i}+\left(1-x_i\right)\tilde{a}_{2,i},
\end{align}
\end{subequations}
respectively, and the corresponding average viabilities are
\begin{subequations}\label{sec1-eq10}
 \begin{align}
\tilde{\omega}_{k,i}&=1+\tilde{a}_{k,i},\\
\tilde{\omega}_{i}&=x_i\tilde{\omega}_{1,i}+\left(1-x_i\right)\tilde{\omega}_{2,i}=1+\tilde{a}_{i},
\end{align}
\end{subequations}
respectively, for $k=1, 2$ and $i=0,1, \ldots, N$. Therefore, the frequency
 of type $S_1$ among offspring after viability selection in a deme of type $i$ can be expressed as
\begin{equation}\label{sec1-eq11}
\tilde{x}_i=\frac{x_i\tilde{\omega}_{1,i}}{\tilde{\omega}_{i}},
\end{equation}
for $i=0,1, \ldots, N$. 

Note that the weighted frequency of $S_1$ among all offspring after viability selection is
\begin{equation}\label{sec1-eq12}
\tilde{x}=\sum_{i=0}^{N}z_i\tilde{x}_i=\sum_{i=0}^{N}z_ix_i\frac{\tilde{\omega}_{1,i}}{\tilde{\omega}_{i}}
\end{equation}
if all demes have the same weight, but 
\begin{equation}\label{sec1-eq15}
\tilde{\tilde{x}}=\sum_{i=0}^{N}z_i\tilde{x}_{i}\frac{\tilde{\omega}_{i}}{\tilde{\omega}}=\sum_{i=0}^{N}z_ix_i\frac{\tilde{\omega}_{1,i}}{\tilde{\omega}_{}}
\end{equation}
where
\begin{equation}\label{sec1-eq13}
\tilde{\omega}=\sum_{i=0}^{N}z_i\tilde{\omega}_{i}=1+\sum_{i=0}^{N}z_i\tilde{a}_{i}=1+\tilde{a},
\end{equation}
if the weights of the demes are given by the average viabilities of the offspring within demes. The latter is the frequency of $S_1$ in the migrant pool. Therefore, this is the frequency of $S_1$ among offspring after recolonization in a deme of type $i$ that goes extinct, which occurs with probability $m$. In such a case, the number of offspring of type $S_1$ drawn at random in that deme to start the next generation follows a binomial distribution with parameters $N$ and $\tilde{\tilde{x}}$. Otherwise, with probability $1-m$, the parameters are $N$ and $\tilde{x}_i$, for $i=0, 1, \ldots, N$.

%In case of extinction, which occurs with probability $m$, a deme will be recolonized by offspring coming from the migration pool. 
%Demes of type $i$ will contribute in average $z_i\tilde{\omega}_{i}/\tilde{\omega}$ in the vacant deme, with
%\begin{equation}\label{sec1-eq13}
%\tilde{\omega}=\sum_{i=0}^{N}z_i\tilde{\omega}_{i}=1+\tilde{a}
%\end{equation}
%is the mean viability in the whole population and
%\begin{equation}\label{sec1-eq14}
%\tilde{a}=\sum_{i=0}^{N}z_i\tilde{a}_{i}
%\end{equation} is the mean payoff in the whole population.
%Then, the average frequency of type $S_1$ in the vacant deme after recolonization is 
%\begin{equation}\label{sec1-eq15}
%\tilde{\tilde{x}}=\sum_{i=0}^{N}z_i\tilde{x}_{i}\frac{\tilde{\omega}_{i}}{\tilde{\omega}}.
%\end{equation}

%The new frequency of type $S_1$ in any deme of type $i$ at generation $0$, taking into account the possible scenario of local extinction and uniform recolonization, is transformed according to 
%\begin{equation}\label{sec1-eq16}
%\tilde{\tilde{x}}_i=(1-m)\tilde{x}_i+m\tilde{\tilde{x}},
%\end{equation}
%for $i=0,1,\ldots,N$. 
%From Eq. (\ref{sec1-eq16}), it is clear that the new frequency is the result of a competition within each deme measured by $\tilde{x}_i$ and weighted by $1-m$, and a competition between demes measured by $\tilde{\tilde{x}}$ and weighted by $m$.
%The next generation, denoted by $t=1$, is produced by randomly selecting $N$ individuals within each deme. 
We conclude that a deme that was of type $i$ at the beginning of generation $0$ has a probability 
\begin{equation}\label{sec1-eq17}
P_{ij}(\mathbf{z})=(1-m)\binom{N}{j}\left(\tilde{x}_i\right)^j\left(1-\tilde{x}_i\right)^{N-j}+m\binom{N}{j}\tilde{\tilde{x}}^j\left(1-\tilde{\tilde{x}}\right)^{N-j}
\end{equation}
to become a deme of type $j$ at the beginning of generation $1$, for $i, j= 0, 1, \ldots, N$.

%%%%%%%%%%%%%%%%%%%%%%%%%%%%%%%%%%%%%%%%%%%%%%%%%%%%%%%%%%%%%%%%%%%%%%%%%%%%%%%%%%%%%%%%%%%%%%%%
%%%%%%%%%%%%%%%%%%%%%%%%%%%%%%%%%%%%%%%%%%%%%%%%%%%%%%%%%%%%%%%%%%%%%%%%%%%%%%%%%%%%%%%%%%%%%%%%

\section{Diffusion approximation}
In this section, we will derive a diffusion approximation that relies on the existence of two timescales in the stochastic process. The first timescale is long and pertains to the changes in the individual type frequencies. The second timescale, that will be described later on, is short and relates to the changes in the deviations of the deme type frequencies from their equilibrium values in a neutral population subdivided into an infinite number of demes.

Let $Z_{ij}(1)$ be the proportion of demes of type $j$ at the beginning of generation $1$ that were of type $i$ at the beginning of generation $0$. Note that  the variables $Z_{i_1j_1}(1)$ and $Z_{i_2j_2}(1)$ are independent for all $i_2,j_2=0,1,\ldots,N$ as long as $i_1\neq i_2$. Moreover, given $\mathbf{Z}(0)=\mathbf{z}$ and the payoff matrix $(a_{k, l})$, the vector $dz_i\Big(Z_{i0}(1),Z_{i1}(1),\ldots,Z_{iN}(1)\Big)$ follows a conditional multinomial distribution with parameters $dz_i$ and $\Big(P_{i0}(\mathbf{z}),P_{i1}(\mathbf{z}),\ldots,P_{iN}(\mathbf{z})\Big)$, for $i=0, 1, \ldots, N$.
Therefore, we have
\begin{subequations}\label{sec2-eq1}
\begin{align}
&E_{\mathbf{z}}\left[ Z_{ij}(1))\Big|(a_{k,l})\right]=P_{ij}(\mathbf{z}),\\
&Var_{\mathbf{z}}\left[ Z_{ij}(1))\Big|(a_{k,l})\right]=\frac{P_{ij}(\mathbf{z})\left(1-P_{ij}(\mathbf{z})\right)}{dz_i},\\
&Cov_{\mathbf{z}}\left(Z_{ij_1}(1)),Z_{ij_2}(\mathbf{z}))\Big|(a_{k,l})\right)=-\frac{P_{ij_1}(\mathbf{z})P_{ij_2}(\mathbf{z})}{dz_i},
\end{align}
\end{subequations}
for $j_1\not= j_2$. Here, we use $E_{\mathbf{z}}$, $Var_{\mathbf{z}}$, and $Cov_{\mathbf{z}}$ to represent the conditional expectation, variance, and covariance, respectively, given that $\mathbf{Z}(0)=\mathbf{z}$.

Note that, for the fraction of demes of type $j$ at the beginning of generation $1$ given by
\begin{equation}\label{sec2-eq2}
Z_j(1)=\sum_{i=0}^{N}z_iZ_{ij}(1),
\end{equation}
we have
\begin{subequations}\label{sec2-eq3}
\begin{align}
&E_{\mathbf{z}}\left[Z_j(1)\Big|(a_{k,l})\right]=\sum_{i=0}^{N}z_iP_{ij}(\mathbf{z}),\\
%%%%
%%%%
&Var_{\mathbf{z}}\left(Z_j(1)\Big|(a_{k,l})\right)
=\frac{1}{d}\sum_{i=0}^{N}z_iP_{ij}(\mathbf{z})\left(1-P_{ij}(\mathbf{z})\right),\\
%%%%
%%%%
&Cov_{\mathbf{z}}\left(Z_{j_1}(1),Z_{j_2}(1)\Big|(a_{k,l})\right)
=-\frac{1}{d}\sum_{i=0}^{N}z_iP_{ij_1}(\mathbf{z})P_{ij_2}(\mathbf{z}),
\end{align}
\end{subequations}
for $j, j_1, j_2=0, 1, \ldots, N$ with $j_1\ne j_2$.

The first three propositions below address changes in the individual type frequencies. Throughout the paper, $I$ and $J$ will be used to designate randomly chosen offspring produced in demes chosen at random and independently before viability selection and recolonization. For offspring produced in the same deme chosen at random, the notation $I_1$, $I_2$, etc. will be used, and similarly $J_1$, $J_2$, etc. for offspring produced in the same deme chosen at random and independently from the former deme. 
Moreover, we introduce the population-scaled parameters 
\begin{subequations}\label{sec2-eq4}
\begin{align}
\mu_{I_1}&=\sum_{k,l=1}^{2}q_{k,I_1}q_{l,I_2}\mu_{k,l},\\
\sigma_{I_1,I_2}&=\sum_{k,l=1}^{2}q_{k,I_1}q_{k,I_2}q_{l,I_3}q_{l,I_4}\sigma^2_{k,l},\\
\sigma_{I_1,J_1}&=\sum_{k,l=1}^{2}q_{k,I_1}q_{k,J_1}q_{l,I_2}q_{l,J_2}\sigma^2_{k,l}.
\end{align}
\end{subequations}
On the other hand, note that
\begin{subequations}\label{sec1-eq8}
\begin{align}
&E_{\mathbf{z}}\left( q_{1,I_1}\right)=x,\\
&Var_{\mathbf{z}}(q_{1,I_1})=x-x^2,\\
&Cov_{\mathbf{z}}\left(q_{1,I_1},q_{1,I_2}\right)=\sum_{i=0}^{N}z_ix^2_i-x^2,
\end{align} 
\end{subequations}
where $I_1$ and $I_2$ are two randomly chosen offspring produced in the same deme at the beginning of generation $0$.
%where $I$ is an individual chosen at random from the population at the beginning of generation $0$.
%where $I'$ and $J'$ denote two offspring randomly selected in the same demes that include $I$ and and $J$, respectively. 

Our first result concern the expected change in the frequency of of $S_1$.
%%%%%%%%
%%%%%%%%

\begin{Proposition}\label{Proposition1}
The first conditional moment of the change in the frequency of $S_1$ in the whole population from the beginning of generation $0$ to the beginning of generation $1$ is 
\begin{equation}\label{sec2-eq5}
E_{\mathbf{z}}\left[X(1)-X(0)\right]=\frac{1}{Nd}M(\mathbf{z})+o\left(d^{-1}\right),
\end{equation}
where
\begin{align}\label{sec2-eq6}
M(\mathbf{z})&=Cov_{\mathbf{z}}(q_{1,I_1},\mu_{I_1})+(1-m)\Big[Cov_{\mathbf{z}}(q_{1,I_1},\sigma_{I_2,I_3})-Cov_{\mathbf{z}}(q_{1,I_1},\mu_{I_2})-Cov_{\mathbf{z}}(q_{1,I_1},\sigma_{I_1,I_2})\Big]\nonumber\\
&\quad-mCov_{\mathbf{z}}(q_{1,I_1},\sigma_{I_1,J_1}).
\end{align}
\end{Proposition}

%%%%%%%%%
%%%%%%%%%
In order to prove this proposition, we will use two  technical lemmas whose proofs are relegated to Appendices A and B.

\begin{lemma}\label{lemma1}
%The first order approximation of the average frequency of type $S_1$ in the population after reproduction but before extinction and recolonization is
For the uniformly weighted frequency of $S_1$ after viability selection within demes given in Eq. (\ref{sec1-eq12}), we have
\begin{align}\label{sec2-eq7}
E_{\mathbf{z}}\left[\tilde{x}\right]&=x+\frac{1}{Nd}\Big[Cov_{\mathbf{z}}(q_{1,I_1},\mu_{I_1})+Cov_{\mathbf{z}}(q_{1,I_1},\sigma_{I_2,I_3})-Cov_{\mathbf{z}}(q_{1,I_1},\mu_{I_2})\nonumber\\
&\quad\quad\quad\quad\quad-Cov_{\mathbf{z}}(q_{1,I_1},\sigma_{I_1,I_2})\Big]+o(d^{-1}).
\end{align}
%where $I_1$, $I_2$, and $I_3$ are three offspring randomly selected in the same deme before extinction and recolonization.
\end{lemma}

%%%%%%

%%%%%%%%%
%%%%%%%%%
\begin{lemma}\label{lemma2}
%The average frequency of type $S_1$ in a deme that has been extinguished and uniformly recolonized can be written as
For the frequency of $S_1$ in the migrant pool given in Eq. (\ref{sec1-eq15}), we have
\begin{equation}\label{sec2-eq12}
E_{\mathbf{z}}\left[\tilde{\tilde{x}}\right]=x+\frac{1}{Nd}\Big[Cov_{\mathbf{z}}(q_{1,I_1},\mu_{I_1})-Cov_{\mathbf{z}}(q_{1,I_1},\sigma_{I_1,J_1})\Big]+o\left(d^{-1}\right).
\end{equation}
%where $I_1$ and $J_1$ are two individuals randomly chosen in the population before reproduction.
\end{lemma}

%%%%%%%%%
%%%%%%%%%

%%%%%%%%%
\begin{proof}[Proof of Proposition \ref{Proposition1}] 
Note first that the probability distribution given by $(P_{ij}(\mathbf{z}))_{j=0}^{N}$ is a mixture of a binomial distribution with parameters $N$ and $\tilde{x}_i$ and a binomial distribution with parameters $N$ and $\tilde{\tilde{x}}$ whose weights are $1-m$ and $m$, respectively. Therefore, we have
\begin{equation}\label{sec2-eq15}
\sum_{j=0}^{N}jP_{ij}(\mathbf{z})=(1-m)N\tilde{x}_i+mN\tilde{\tilde{x}}.
\end{equation}
Using this identity and Eq. (\ref{sec2-eq3}), we find that
\begin{align}\label{sec2-eq16}
E_{\mathbf{z}}\left[X(1)-X(0)\Big|(a_{k,l})\right]&=\sum_{j=0}^{N}x_jE_{\mathbf{z}}\left[Z_{j}(1)\Big|(a_{k,l})\right]-x\nonumber\\
&=\sum_{i=0}^{N}z_i\sum_{j=0}^{N}x_jP_{ij}(\mathbf{z})-x\nonumber\\
&=(1-m)\tilde{x}+m\tilde{\tilde{x}}-x.
\end{align}
Taking the expected value and using Lemmas 1 and 2 lead  to Proposition \ref{Proposition1}.
\end{proof}

%%%%%%%%%
Our second result concerns the variance of the change in the frequency of $S_1$.
%%%%%%%%%

\begin{Proposition}\label{Proposition2}
The conditional variance of the change in the frequency of $S_1$ in the whole population from the beginning of generation $0$ to the beginning of generation $1$ is  
\begin{equation}\label{sec2-eq17}
Var_{\mathbf{z}}\left[X(1)-X(0)\right]=\frac{1}{Nd}Q(\mathbf{z})+o\left(d^{-1}\right),
\end{equation}
where
\begin{align}\label{sec2-eq18}
Q(\mathbf{z})&=Var_{\mathbf{z}}(q_{1,I_1})+(1-m)(Nm-1)Cov_{\mathbf{z}}(q_{1,I_1},q_{1,I_2})+m^2Cov_{\mathbf{z}}(q_{1,I_1}q_{1,J_1},\sigma_{I_1,J_1})\nonumber\\
&\quad-2mx\,Cov_{\mathbf{z}}(q_{1,I_1},\sigma_{I_1,J_1})
+(1-m^2)Cov_{\mathbf{z}}(q_{1,I_1},q_{1,J_1}\sigma_{I_1,J_1})\nonumber\\
&\quad-(1-m^2)Cov_{\mathbf{z}}(q_{1,I_1},q_{1,J_1}\sigma_{I_2,J_1})
+2m(1-m)x\,Cov_{\mathbf{z}}(q_{1,I_1},\sigma_{I_2,J_1})\nonumber\\
&\quad+(1-m)^2Cov_{\mathbf{z}}(q_{1,I_1},q_{1,J_1}\sigma_{I_2,J_2})-(1-m)^2Cov_{\mathbf{z}}(q_{1,I_1},q_{1,J_1}\sigma_{I_1,J_2}).
\end{align}
%Here $I_1$ and $I_2$ are two individuals chosen randomly in the same deme, while $J_1$ and $J_2$ designate two individuals selected randomly in an other deme, all before reproduction.
\end{Proposition}
%%%%%%%%%
The proof of this proposition relies on three technical lemmas proved in Appendices C, D and E.

%%%%%%%%%
%%%%%%%%%
\begin{lemma}\label{lemma3}
The conditional variance of the uniformly weighted frequency of $S_1$ after viability selection within demes is given by
\begin{align}\label{sec2-eq19}
Var_{\mathbf{z}}\left[\tilde{x}\right]&=
\frac{1}{Nd}\Big[
Cov_{\mathbf{z}}(q_{1,I_1},q_{1,J_1}\sigma_{I_1,J_1})
+Cov_{\mathbf{z}}(q_{1,I_1},q_{1,J_1}\sigma_{I_2,J_2})
-Cov_{\mathbf{z}}(q_{1,I_1},q_{1,J_1}\sigma_{I_2,J_1})\nonumber\\
&\quad-Cov_{\mathbf{z}}(q_{1,I_1},q_{1,J_1}\sigma_{I_1,J_2})
\Big]+o(d^{-1}).
\end{align}
%where, $I_1$ and $I_2$ represent two individuals chosen at random from the same deme, and, $J_1$ and $J_2$ represent two individuals chosen at random from another deme, all before reproduction.
\end{lemma}

%%%%%%%%%

%%%%%%%%%%
%%%%%%%%%%

\begin{lemma}\label{lemma4}
The conditional variance of the frequency of $S_1$ in the migrant pool is given by
\begin{equation}\label{sec2-eq23}
Var_{\mathbf{z}}(\tilde{\tilde{x}})=\frac{1}{Nd}\Big[
Cov_{\mathbf{z}}(q_{1,I_1}q_{1,J_1},\sigma_{I_1,J_1})-2x\,Cov_{\mathbf{z}}(q_{1,I_1},\sigma_{I_1,J_1})
\Big]+o(d^{-1}).
\end{equation}
%where $I_1$ and $J_1$ designate two individuals selected at random in the population before reproduction.
\end{lemma}

%%%%%%%%%

%%%%%%%%%
%%%%%%%%%
\begin{lemma}\label{lemma5}
The conditional covariance between the uniformly weighted frequency of $S_1$ and the frequency of $S_1$ in the migrant pool is given by
\begin{align}\label{sec2-eq25}
Cov_{\mathbf{z}}\left[\tilde{x},\tilde{\tilde{x}}\right]&=\frac{1}{Nd}\Big[
Cov_{\mathbf{z}}(q_{1,I_1},q_{1,J_1}\sigma_{I_1,J_1})
+x\,Cov_{\mathbf{z}}(q_{1,I_1},\sigma_{I_2,J_1})
-x\,Cov_{\mathbf{z}}(q_{1,I_1},\sigma_{I_1,J_1})\nonumber\\
&\quad
-Cov_{\mathbf{z}}(q_{1,I_1},q_{1,J_1}\sigma_{I_2,J_1})
\Big]+o(d^{-1}).
\end{align}
%where, $I_1$ and $I_2$ denote two individuals randomly chosen from the same deme, and $J_1$ denotes an individual chosen at random from another deme, all before reproduction takes place.
\end{lemma}

%%%%%%%%%

%%%%%%%%%
%%%%%%%%%

\begin{proof}[Proof of Proposition \ref{Proposition2}] 
First, we note that 
\begin{align}\label{sec2-eq27}
\sum_{j=0}^{N}j^2P_{ij}(\mathbf{z})&= (1-m)\sum_{j=0}^{N}j^2\binom{N}{j}\left(\tilde{x}_i\right)^j\left(1-\tilde{x}_i\right)^{N-j}+m\sum_{j=0}^{N}j^2\binom{N}{j}\tilde{\tilde{x}}^j\left(1-\tilde{\tilde{x}}\right)^{N-j}\nonumber\\
&=(1-m)N\tilde{x}_i\left((N-1)\tilde{x}_i+1\right)
+mN\tilde{\tilde{x}}\left((N-1)\tilde{\tilde{x}}+1\right).
\end{align}
%This identity is a consequence of the fact that $(P_{ij}(\mathbf{z}))_{j=0}^{N}$ is the sum of two binomial distributions. The first binomial distribution has parameters $N$ and $\tilde{x}_i$, weighted by $(1-m)$. The second binomial distribution has parameters $N$ and $\tilde{\tilde{x}}$, weighted by $m$.
Combining the above identity with Eqs. (\ref{sec2-eq3}) and (\ref{sec2-eq15}), we get
\begin{align}\label{sec2-eq28}
Var_{\mathbf{z}}\left[X(1)-X(0)\Big|(a_{k,l})\right]
&=Var_{\mathbf{z}}\left[\sum_{j=0}^{N}\frac{j}{N}Z_{j}(1)\Big|(a_{k,l})\right]\nonumber\\
&=\sum_{j=0}^{N}\frac{j^2}{N^2}Var_{\mathbf{z}}\left[Z_{j}(1)\Big|(a_{k,l})\right]+\sum_{j_1\not= j_2}\frac{j_1j_2}{N^2}Cov_{\mathbf{z}}\left(Z_{j_1}(1),Z_{j_2}(1)\Big|(a_{k,l})\right)\nonumber\\
&=\sum_{j=0}^{N}\frac{j^2}{dN^2}\sum_{i=0}^{N}z_iP_{ij}(\mathbf{z})\left(1-P_{ij}(\mathbf{z})\right)-\sum_{ j_1\not= j_2}\frac{j_1j_2}{dN^2}\sum_{i=0}^{N}z_iP_{ij_1}(\mathbf{z})P_{ij_2}(\mathbf{z})\nonumber\\
&=\frac{1}{d}\sum_{i=0}^{N}z_i\left[\sum_{j=0}^{N}\frac{j^2}{N^2}P_{ij}(\mathbf{z})-\left(\sum_{j=0}^{N}\frac{j}{N}P_{ij}(\mathbf{z})\right)^2\right]\nonumber\\
&=\frac{1}{Nd}\sum_{i=0}^{N}z_i\big[
(1-m)(Nm-1)\tilde{x}^2_i+m(N(1-m)-1)\tilde{\tilde{x}}^2\nonumber\\
&\quad\quad\quad\quad\quad\quad+(1-m)\tilde{x}_i+m\tilde{\tilde{x}}-2Nm(1-m)\tilde{x}_i\tilde{\tilde{x}}
\big]\nonumber\\
&=\frac{1}{Nd}\Big[(1-m)(Nm-1)\sum_{i=0}^{N}z_i\tilde{x}^2_i+m(N(1-m)-1)\tilde{\tilde{x}}^2\nonumber\\
&\quad\quad\quad+(1-m)\tilde{x}+m\tilde{\tilde{x}}-2Nm(1-m)\tilde{x}\tilde{\tilde{x}}
\Big].
\end{align}
Using this equation with Eq. (\ref{sec2-eq16}) and the facts that
\begin{subequations}\label{equation*}
\begin{align}
&E_{\mathbf{z}}\left(\tilde{x}_i^n\right)=x_i^n+o\left(1\right),\\
&E_{\mathbf{z}}\left(\tilde{x}\tilde{\tilde{x}}\right)=x^2+o\left(1\right),
\end{align}
\end{subequations}
for $n\geq 1$, we obtain
\begin{align}\label{sec2-eq29}
Var_{\mathbf{z}}\left[X(1)-X(0)\right]%\nonumber\\
&=Var_{\mathbf{z}}\left(E_{\mathbf{z}}\left[X(1)-X(0)\Big|(a_{k,l})\right]\right)
+E_{\mathbf{z}}\left(Var_{\mathbf{z}}\left[X(1)-X(0)\Big|(a_{k,l})\right]\right)\nonumber\\
&=Var_{\mathbf{z}}\Big((1-m)\tilde{x}+m\tilde{\tilde{x}}-x\Big)
+\frac{1}{Nd}\Big[(1-m)(Nm-1)\sum_{i=0}^{N}z_ix^2_i
\nonumber\\
&\quad+m(N(1-m)-1)x^2+(1-m)x+mx
-2Nm(1-m)x^2
\Big]+o\left(d^{-1}\right)\nonumber\\
%%%%
&=(1-m)^2Var_{\mathbf{z}}\left[\tilde{x}\right]+m^2Var_{\mathbf{z}}\left[\tilde{\tilde{x}}\right]+2m(1-m)Cov_{\mathbf{z}}\left[\tilde{x},\tilde{\tilde{x}}\right]+\frac{1}{Nd}\\
&\quad\times\Big[(1-m)(Nm-1)\sum_{i=0}^{N}z_ix^2_i
-m(N(1-m)+1)x^2+x\Big]+o\left(d^{-1}\right).\nonumber
\end{align}
%Here, we have used  the fact that
%\begin{subequations}\label{equation*}
%\begin{align}
%&E_{\mathbf{z}}\left(\tilde{x}_i^n\right)=x_i^n+o\left(1\right),\\
%&E_{\mathbf{z}}\left(\tilde{x}\tilde{\tilde{x}}\right)=x^2+o\left(1\right),
%\end{align}
%\end{subequations}
%for $n\geq 1$. 
%Additionally, the following identities are noted
%\begin{equation}
%\begin{split}
%&\sum_{i=0}^{N}z_ix^2_i=Cov_{\mathbf{z}}\left(q_{1,I_1},q_{1,I_2}\right)+x^2,\\
%&x=Var_{\mathbf{z}}(q_{1,I_1})+x^2,
%\end{split}
%\end{equation}
%where $I_1$ and $I_2$ are two individuals chosen at random in the same deme.
Using the identities in Eq. (\ref{sec1-eq8}) and Lemmas \ref{lemma3}-\ref{lemma5} leads to the statement of Proposition \ref{Proposition2}.
\end{proof}
%%%%%%%%%
%%%%%%%%%

The next result, proved in Appendix F, implies that the fourth conditional moment of the change in the frequency of $S_1$ is negligible compared to the conditional moments of order one and two when the number of demes is large.

%%%%%%%%%
%%%%%%%%%
\begin{Proposition}\label{Proposition3}
For the fourth conditional moment of the change in the frequency of $S_1$ in the whole population from the beginning of generation $0$ to the beginning of generation $1$, we have
\begin{equation}\label{sec2-eq30}
E_{\mathbf{z}}\left[\Big(X(1)-X(0)\Big)^4\right]=o(d^{-1}).
\end{equation}
\end{Proposition}
%%%%%%%%%
%%%%%%%%%

%%%%%%%%%
%%%%%%%%%

Let us now focus on the second timescale for the deviations of the deme type frequencies from their equilibrium values in a neutral population subdivided into an infinite number of demes.
We first note that, as a consequence of Eq. (\ref{sec1-eq17}), we have
\begin{equation}\label{sec2-eq37}
E_{\mathbf{z}}\Big[P_{ij}(\mathbf{z})\Big]=P^{*}_{ij}(x)+o(1),
\end{equation}
as $S_2$ approaches infinity, where
\begin{equation}\label{sec2-eq38}
P^{*}_{ij}(x)=(1-m)\binom{N}{j}x_i^j\left(1-x_i\right)^{N-j}+m\binom{N}{j}x^j\left(1-x\right)^{N-j},
\end{equation}
for $i, j =0, 1, \ldots,N$, with $x=\sum_{j=0}^N x_jz_j= \sum_{j=0}^N jz_j/N$. Here,
$P^{*}_{ij}(x)$ represents the transition probability for a deme of type $i$ at a the beginning of generation $0$ to become a deme of type $j$ at the beginning of generation $1$ in an infinite neutral population subdivided into an infinite number of demes each of size $N$, where the frequency of $S_1$ is constant and given by $x$.

Let  $\mathbf{v}(x)=(v_0(x), v_1(x), \ldots,v_N(x))$ be the frequency vector that is the solution to the linear system of equations
\begin{equation}\label{sec2-eq39}
v_j(x)=\sum_{i=0}^{N}v_i(x)P_{ij}^{*}(x)
\end{equation}
for $j=0,1,\ldots,N$.
In other words, $\mathbf{v}(x)$ is the stationary probability distribution of the process $\left\{\mathbf{Z}(t)\right\}_{t\geq 0}$ for the deme type frequencies in the absence of selection, when the population is subdivided into an infinite number of demes of the same finite size $N$, and the total frequency of type $S_1$ is constant and equal to $x$.

The focus of our analysis will now shift to the  process $\left\{\mathbf{Y}(t)=\left(Y_0(t),Y_1(t),\ldots,Y_N(t)\right)\right\}_{t\geq 0}$, where 
 \begin{equation}
 Y_j(t)=Z_j(t)-v_j(X(t))
 \end{equation}
 for $j=0,1,\ldots,N$ define the deviations of the deme type frequencies from their equilibrium values as the number of demes $S_2$ goes to infinity.

\begin{Proposition}\label{Proposition4}
The first conditional moment of the changes in the deviations of the deme type frequencies from their equilibrium values as $d\rightarrow \infty$ from the beginning of generation $0$ to the beginning of generation $1$ is given by
\begin{equation}\label{sec2-eq40}
E_{\mathbf{z}}\left[Y_i(1)-Y_i(0)\right]=c_i(x,\mathbf{y})+o\left(1\right),
\end{equation}
where
\begin{equation}\label{sec2-eq41}
c_i(x,\mathbf{y})=\sum_{j=0}^{N}y_jP_{ji}^*(x)-y_i
\end{equation}
with $x=\sum_{j=0}^N x_jz_j$ and $y_i= z_i - v_i(x)$, for $i=0,1,\ldots,N$.
\end{Proposition}
%%%%%%%%%
%%%%%%%%%
\begin{proof}
It has been shown in Lessard \cite{L2009} that
\begin{equation}\label{sec2-eq42}
v_i(X(1))=v_i(x)+\sum_{j=1}^{N}r_j\left(X(1)-X(0)\right)^j,
\end{equation} 
for $i=0, 1, \ldots, N$, where $r_j$ is a constant that depends only on $N$, $j$, $m$ and $x$, for $j=1, \ldots, N$.
Furthermore, we have
\begin{equation}\label{sec2-eq43}
\left|E_{\mathbf{z}}\left[\left(X(1)-X(0)\right)^k\right]\right|\leq E_{\mathbf{z}}\left[\left|X(1)-X(0)\right|^k\right]\leq E_{\mathbf{z}}\left[\left(X(1)-X(0)\right)^2\right]
\end{equation}
for any integer $k\geq 2$. Hence, owing to Propositions 1 and 2, we have
%\begin{equation}\label{sec2-eq44}
%E_{\mathbf{z}}\left[\left(X(1)-X(0)\right)^k\right]=o(1)
%\end{equation}
%for any integer $k\geq 2$. 
%Using the last equation and Proposition \ref{Proposition1}, we obtain
\begin{equation}\label{sec2-eq45}
E_{\mathbf{z}}\left[v_i(X(1))\right]=v_i(x)+\sum_{j=1}^{N}r_j E_{\mathbf{z}}\left[\left(X(1)-X(0)\right)^j\right]=v_i(x)+o(1).
\end{equation} 
 On the other hand, we have
\begin{align}\label{sec2-eq46}
E_{\mathbf{z}}\left[Z_i(1)\right]&=\sum_{j=0}^{N}z_jE_{\mathbf{z}}\left[P_{ji}(\mathbf{z})\right]\nonumber\\
&=\sum_{j=0}^{N}z_jP_{ji}^*(x)+o(1)\nonumber\\
&=\sum_{j=0}^{N}(y_j+v_j(x))P_{ji}^*(x)+o(1)\nonumber\\
&=\sum_{j=0}^{N}y_jP_{ji}^*(x)+v_i(x)+o(1),
\end{align}
where $y_j= z_j - v_j(x)$ for $j=0,1,\ldots,N$.
By combining (\ref{sec2-eq45}) and (\ref{sec2-eq46}), the first conditional moment of the change in $Y_i$ from the beginning of generation 0 to the beginning of generation 1 can be expressed as
\begin{equation}\label{sec2-eq47}
E_{\mathbf{z}}\left[Y_i(1)-Y_i(0)\right]=E_{\mathbf{z}}\left[Z_i(1)\right]-E_{\mathbf{z}}\left[v_i(X(1))\right]-y_i
=c_i(x,\mathbf{y})+o\left(1\right),
\end{equation}
with $c_i(x,\mathbf{y})$ as given in the statement of Proposition 4, for $i=0, 1, \ldots, N$.
\end{proof}

%%%%%%%%%
%%%%%%%%%
On the other hand, it can be shown (see Appendix G for a proof) that the conditional variance of the change in $Y_i$ is negligible compared to the first conditional moment when the number of demes is large.
\begin{Proposition}\label{Proposition5}
For  the conditional variance of the change in $Y_i$ from the beginning of generation 0 to the beginning of generation 1, we have
\begin{equation}\label{sec2-eq48}
Var_{\mathbf{z}}\left[Y_i(1)-Y_i(0)\right]=o\left(1\right)
\end{equation}
for $i=0, 1, \ldots, N$.
\end{Proposition}
%%%%%%%%%
%%%%%%%%%

%%%%%%%%%
%%%%%%%%%
Finally, we state and prove a convergence result that holds for an infinite neutral population. 

\begin{Proposition}\label{Proposition6}
For the function $\mathbf{c}(x, \mathbf{y})=(c_0(x, \mathbf{y}), c_1(x, \mathbf{y}), \ldots, c_N(x, \mathbf{y}))$ whose components are defined in Eq. (\ref{sec2-eq41}), we have $\mathbf{c}(x,\mathbf{0})=\mathbf{0}$ and the zero solution of the recurrence equation
\begin{equation}\label{sec2-eq51}
\mathbf{Y}(k+1,x,\mathbf{y})-\mathbf{Y}(k,x,\mathbf{y})=\mathbf{c}(x,\mathbf{Y}(k,x,\mathbf{y})),
\end{equation}
for $k\geq 0$ with the initial condition
$\mathbf{Y}(0,x,\mathbf{y})=\mathbf{y}$ is globally asymptotically stable.%, uniformly in $(x, \mathbf{y})$.
\end{Proposition}

%%%%%%%%%
%%%%%%%%%

\begin{proof}
Note that the recurrence equation Eq. (\ref{sec2-eq51}) can be rewritten as 
\begin{equation}\label{sec2-eq53}
\mathbf{Z}(k+1,x,\mathbf{y})=\mathbf{Z}(k,x,\mathbf{y})P^{*}(x),
\end{equation}
where $\mathbf{Z}(k,x,\mathbf{y})=\mathbf{Y}(k,x,\mathbf{y})+\mathbf{v}(x)$ and $P^{*}(x)=(P_{i,j}^{*}(x))_{}$ is as defined in Eq. (\ref{sec2-eq38}).
Moreover, owing to Eq. (\ref{sec2-eq39}), the frequency vector $\mathbf{v}(x)$ represents the stationary distribution associated with $P^{*}(x)$ which is the transition matrix of an irreducible Markov chain on a finite state space. Then, the ergodic theorem ensures that
\begin{equation}\label{sec2-eq54}
\lim_{k\rightarrow\infty}\mathbf{Z}(k,x,\mathbf{y})=\mathbf{v}(x)
\end{equation}
for any $\mathbf{y}$, which corresponds to the statement of Proposition 6.
\end{proof}

%%%%%%%%%%%%%%%%%%%%%%%%%%%%%%%%%%%%%%%%%%%%%%%%%%%%%%%%%%%%%%%%%%%%%%%%%%%%%%%%%%%%%%%%%%%%%%%%
%%%%%%%%%%%%%%%%%%%%%%%%%%%%%%%%%%%%%%%%%%%%%%%%%%%%%%%%%%%%%%%%%%%%%%%%%%%%%%%%%%%%%%%%%%
Owing to Propositions (\ref{Proposition1})-(\ref{Proposition6}), which are actually valid uniformly with respect to the population state, and assuming the existence and uniqueness of a diffusion process 
with infinitesimal mean and variance $M(\mathbf{v}(x))$ and $Q(\mathbf{v}(x))$, respectively, whose closure determines a strongly continuous semigroup (see Ethier \cite{E1976}), Theorem 3.3 in Ethier and Nagylaki \cite{EN1980} ensures the following convergence theorem.

\begin{Theorem}\label{Theorem}
Let $X(\lfloor Nd\tau\rfloor)$ be the frequency of type $S_1$ at the beginning of generation $\lfloor Nd\tau\rfloor$ in the model of Section 2. Here, $\lfloor Nd\tau\rfloor$ denotes the integer part of  $Nd\tau$ for any real $\tau\geq0$. Then, the process $\{X(\lfloor Nd\tau\rfloor)\}_{\tau\geq0}$ converges
in distribution to a diffusion process in $[0,1]$ whose infinitesimal generator is 
\begin{equation}
\mathcal{L}=\frac{1}{2}Q(\mathbf{v}(x))\frac{d^2}{dx^2}+M(\mathbf{v}(x))\frac{d}{dx},
\end{equation}
where the functions $M$ and $Q$ are given in Propositions 1 and 2, respectively, while the frequency vector $\mathbf{v}(x)$ is defined by Eqs. (\ref{sec2-eq38}) and (\ref{sec2-eq39}).
\end{Theorem}

\section{Infinitesimal mean and variance}

The infinitesimal mean and variance $M(\mathbf{v}(x))$ and $Q(\mathbf{v}(x))$ in the statement of Theorem \ref{Theorem} can be expressed in terms of identity-by-descent measures in a neutral population.
%In this section, we consider  $N\rightarrow\infty$ and $m\rightarrow0$ such that $Nm\rightarrow\nu$. Here, $\nu$ corresponds to a deme-scaled extinction rate.

Let $I_1^{(1)},\ldots,I_{n_1}^{(1)},\ldots,I_1^{(k)},\ldots,I_{n_k}^{(k)}$ be an ordered sample of $n=n_1+\cdots+n_k$ offspring chosen at random in the same deme at the beginning of a given generation in a neutral population of an infinite number of demes that is in the stationary state. Denote by 
\begin{align}\label{identitymeasures}
f_{I_1^{(1)}\ldots \,I_{n_1}^{(1)}| \, \ldots\, |I_1^{(k)}\ldots \,I_{n_k}^{(k)}}=f_{n_1n_2\ldots n_k}
\end{align}
the probability that $I_1^{(i)},\ldots,I_{n_i}^{(i)}$ are identical by descent for $i=1,\ldots,k$, but  $I_1^{(i)},\ldots,I_{n_i}^{(i)}$ and $I_1^{(j)},\ldots,I_{n_j}^{(j)}$ %$I_{l_1}^{(i_1)}$ and $I_{l_2}^{(i_2)}$ 
are not identical by descent %when $i_1\not= i_2$, $l_1=1,\ldots,n_{i_1}$ and $l_2=1,\ldots,n_{i_2}$. 
for $i, j=1,\ldots,k$ with $i\ne j$. The exact expressions of these identity-by-descent measures for up to five offspring  are given in Appendix H. 

The above identity-by-descent measures lead to exact expressions for $M(\mathbf{v}(x))$ and $Q(\mathbf{v}(x))$ deduced in Appendix I. Simplifications are possible in the case where the deme size is large and the probability of extinction is small.

\begin{Proposition}\label{Proposition7}
If $N$ is large enough and $m$ small enough, the infinitesimal mean and variance of the diffusion process in Theorem \ref{Theorem} can be approximated as
\begin{subequations}\label{sec5-eq1}
\begin{align}
M(\mathbf{v}(x))\approx &\;x(1-x)\Big[(\mu_{1,1}-\mu_{2,1})\psi(x)+(\mu_{1,2}-\mu_{2,2})\psi(1-x)-\sigma_{1,1}^2g(x)-\sigma_{1,2}^2h(x)\nonumber\\
&\quad\quad\quad\quad+\sigma_{2,1}^2h(1-x)+\sigma_{2,2}^2g(1-x)\Big],\\
%%%
%%%%
Q(\mathbf{v}(x))\approx &\;x(1-x)\Big[f_{11}+\nu f_2+x(1-x)\psi^2(x)\left(\sigma^2_{1,1}+\sigma^2_{2,1}\right)\nonumber\\
&\quad\quad\quad\quad+x(1-x)\psi^2(1-x)\left(\sigma^2_{1,2}+\sigma^2_{2,2}\right)\Big],
\end{align}
\end{subequations}
where $\nu=Nm$ is a deme-scaled extinction rate and
\begin{subequations}\label{sec5-eq2}
\begin{align}
%%%%
\psi(x)&=f_{21}+xf_{111},\\
%%%
%%%
g(x)&=f_{41}+(4f_{311}+3f_{221})x+6f_{2111}x^2+f_{11111}x^3,\\
%%%
%%%
h(x)&=f_{32}+\left(3f_{221}+f_{2111}\right)(1-x)+xf_{311}+2x(1-x)f_{2111}+x(1-x)^2f_{11111}.
\end{align}
\end{subequations}
%with the identity measures given in Eq. (\ref{identity}).
\end{Proposition}

Moreover, we have the following ancillary result for the identity-by-descent measures, proved in Appendix J.

\begin{Proposition}\label{Proposition8}
%We have 
For a deme size $N$ large enough and a probability of extinction $m$ small enough, we have the approximation
 \begin{equation}\label{identity}
f_{n_1n_2\ldots n_k}\approx \binom{n_1+\cdots+n_k-k}{n_1-1,\ldots,n_k-1}\frac{\prod_{i=1}^{k}\prod_{j=2}^{n_i}\binom{j}{2}}
{\prod_{j=k+1}^{n_1+\cdots+n_k}\left(\binom{j}{2}+\nu\right)}\times\frac{\nu}{\binom{k}{2}+\nu}
\end{equation}
for $n_1, \ldots, n_k \geq 1$ for $k\geq 1$, where $\nu=Nm$.
\end{Proposition}

From now on, the conditions for Propositions \ref{Proposition7} and \ref{Proposition8} will be assumed to hold.

%%%%%%%%%%%%%%%%%%%%%%%%%%%%%%%%%%%%%%%%%%%%%%%%%%%%%%%%%%%%%%%%%%%%%%%%%%%%%%%%%%%%%%%%%%
%%%%%%%%%%%%%%%%%%%%%%%%%%%%%%%%%%%%%%%%%%%%%%%%%%%%%%%%%%%%%%%%%%%%%%%%%%%%%%%%%%%%%%%%%%
%%%%%%%%%%%%%%%%%%%%%%%%%%%%%%%%%%%%%%%%%%%%%%%%%%%%%%%%%%%%%%%%%%%%%%%%%%%%%%%%%%%%%%%%%%
\section{Fixation probability}
In the absence of mutation, the frequency of $S_1$ in the model of Section 2 is a Markov chain that fixes into either $x=1$, all individuals are of type $S_1$, or $x=0$, all individuals are of type $S_2$. 
 On the other hand, the probability of absorption in state $x=1$ at or before time $t$ starting in state $x=p$ at time $0$ in the diffusion process of Theorem 1,  denoted by $P_1(p, t)$, is known to satisfy the backward Kolmogorv equation
\begin{equation}\label{sec4-eq1}
-\frac{\partial P_1(p,t)}{\partial t}=M(\mathbf{v}(x))\frac{\partial P_1(p,t)}{\partial p}+
\frac{Q(\mathbf{v}(x))}{2}\frac{\partial^2P_1(p,t)}{\partial p^2},
\end{equation}
with the boundary conditions $P_1(0,t)=0$ and $P_1(1,t)=1$. See Karlin and Taylor \cite{KT1981} or Ewens \cite{E2004} for more details. Moreover, 
\begin{align}
u_1(p)=\lim_{t\rightarrow\infty}P_1(1,t),
\end{align}
which is the fixation probability of $S_1$ given that $p$ is  the initial frequency of $S_1$, 
satisfies
\begin{equation}\label{sec4-eq2}
0=M(\mathbf{v}(x))\frac{d u_1(p)}{d p}+
\frac{Q(\mathbf{v}(x))}{2}\frac{d^2u_1(p)}{d p^2},
\end{equation}
with the boundary conditions $u_1(0)=0$ and $u_1(1)=1$. This equation has the exact solution
\begin{equation}\label{sec4-eq3}
u_1(p)=\frac{\int_{0}^p\Psi(y)dy}{\int_{0}^1\Psi(y)dy},
\end{equation}
where 
\begin{align}\label{sec4-eq4}
\Psi(y)=\exp\left\{-2\int_{0}^{y}\frac{M(\mathbf{v}(x))}{Q(\mathbf{v}(x))} dx\right\}.
\end{align}
%with
%\begin{align}\label{}
%\frac{M(\mathbf{v}(x))}{Q(\mathbf{v}(x))}=\frac{M(\mathbf{v}(x))}{Q(\mathbf{v}(x))}.
%\end{align}
Similarly, 
\begin{equation}\label{sec4-eq5}
u_2(p)=1-u_1(1-p)=\frac{\int_{1-p}^1\Psi(y)dy}{\int_{0}^1\Psi(y)dy}
\end{equation}
is the fixation probability of $S_2$ given that the initial frequency of $S_2$ is $p$.

Let $F_1$ (respectively, $F_2$) denote the fixation probability of $S_1$ (respectively, $S_2$) introduced as a single mutant in a population that was initially fixed for $S_2$ (respectively, $S_1$). Assuming a sufficiently large number of demes and using the diffusion approximation established in the previous section, we obtain
\begin{subequations}\label{sec4-eq6}
\begin{align}\label{sec4-eq7}
&F_1\approx u_1\left(\frac{1}{Nd}\right)=\frac{\int_{0}^{(Nd)^{-1}}\Psi(y)dy}{\int_{0}^1\Psi(y)dy}\approx\frac{\Psi(0)}{Nd\int_{0}^1\Psi(y)dy},\\
&F_2\approx u_2\left(\frac{1}{Nd}\right)=\frac{\int_{1-(Nd)^{-1}}^{1}\Psi(y)dy}{\int_{0}^1\Psi(y)dy}\approx\frac{\Psi(1)}{Nd\int_{0}^1\Psi(y)dy},
\end{align}
\end{subequations}
since
\begin{subequations}
\begin{align}
&\int_{0}^{(Nd)^{-1}}\Psi(y)dy\approx\frac{\Psi(0)}{Nd},\\
&\int_{1-(Nd)^{-1}}^{1}\Psi(y)dy\approx\frac{\Psi(1)}{Nd}.
\end{align}
\end{subequations}
In the absence of selection, it can be noted that $\Psi(y)=1$ for any $y\in [0,1]$. Then, the above fixation probabilities are equal and given by the initial frequency, that is,
\begin{equation}\label{sec4-eq8}
F_1=F_2=\frac{1}{Nd}.
\end{equation}
In the presence of selection, the evolution of $S_i$ is said to be favoured if its fixation probability when introduced as a single mutant is greater than what it would be in the absence of selection, that is, $F_{i}>(Nd)^{-1}$, for $i=1, 2$ (Nowak \textit{et al.} \cite{NSTF2004}). If the inequality is reversed, that is, $F_{i} < (Nd)^{-1}$, then the evolution of $S_i$ is said to be disfavoured, for $i=1,2$.
It should be noted that the evolution of both types can be favoured or disfavoured by selection. Thus, a third criterion is introduced to determine which type has a higher probability of invading a population consisting entirely of the other type. Selection is said to favour the evolution of $S_1$ more than the evolution of $S_2$ if $F_{1}>F_2$.

Using Eq. (\ref{sec4-eq6}) with the fact that $\Psi(0)=1$ and introducing the function
\begin{equation}\label{sec4-eq4bis}
\Phi(y)=\frac{\Psi(y)}{\Psi(1)}=\exp\left(2\int_{y}^{1}\frac{M(\mathbf{v}(x))}{Q(\mathbf{v}(x))}dx\right),
\end{equation}
where $\Psi$ is defined in Eq. (\ref{sec4-eq4}), we can conclude the following.

\begin{Proposition}\label{Proposition9}

Selection favours the evolution of $S_1$ in the sense that  $F_{1}>(Nd)^{-1}$ as long as 
\begin{equation}\label{sec4-eq9}
\int_{0}^{1}\Psi(y)dy<1.
\end{equation} 
On the other hand, the condition for selection to disfavour the evolution of $S_2$ in the sense that $F_{2}<(Nd)^{-1}$ is
\begin{equation}\label{sec4-eq10}
\int_{0}^{1}\Phi(y)dy>1.
\end{equation}  
%while the reversed inequality implies that selection favours the evolution of $S_2$.
Finally, the condition for selection to favour the evolution of $S_1$ more than the evolution of $S_2$ in the sense that $F_1>F_2$ is
\begin{equation}\label{sec4-eq11}
\Psi(1)<1.
\end{equation}  
If the above three conditions are satisfied, then $F_1>(Nd)^{-1}>F_2$ and we say that selection fully favours the evolution of $S_1$.
\end{Proposition}

\section{Small population-scaled means and variances  }

Let us  assume that the population-scaled means $\mu_{1,1}$, $\mu_{1,2}$, $\mu_{2,1}$ and $\mu_{2,2}$ as well as population-scaled variances $\sigma_{1,1}^2$, $\sigma_{1,2}^2$, $\sigma_{2,1}^2$ and $\sigma_{2,2}^2$ are all of sufficiently small order. 
Then, using Proposition \ref{Proposition7} and the fact that $\nu f_2=\nu/(1+\nu)=f_{11}$, we get the approximation
\begin{align}\label{sec1-case5-eq1}
\frac{M(\mathbf{v}(x))}{Q(\mathbf{v}(x))}\approx\frac{M(\mathbf{v}(x))}{2f_{11}x(1-x)}&=\frac{\psi(x)}{2f_{11}}(\mu_{1,1}-\mu_{2,1})+\frac{\psi(1-x)}{2f_{11}}(\mu_{1,2}-\mu_{2,2})
\nonumber\\
&\quad-\frac{g(x)}{2f_{11}}\sigma_{1,1}^2-\frac{h(x)}{2f_{11}}\sigma_{1,2}^2+\frac{h(1-x)}{2f_{11}}\sigma_{2,1}^2+\frac{g(1-x)}{2f_{11}}\sigma_{2,2}^2.
\end{align}
By using this approximation, along with $e^{u}\approx 1+u$  when $u$ is sufficiently small, we obtain
\begin{subequations}\label{sec1-case5-eq2}
\begin{align}
\Psi(y)&=\exp\left(-2\int_{0}^{y}\frac{M(\mathbf{v}(x))}{Q(\mathbf{v}(x))}\right)\approx 1-\int_{0}^{y}\frac{M(\mathbf{v}(x))}{f_{11}x(1-x)}dx,\\
\Phi(y)&=\exp\left(2\int_{y}^{1}\frac{M(\mathbf{v}(x))}{Q(\mathbf{v}(x))}\right)\approx 1+\int_{y}^{1}\frac{M(\mathbf{v}(x))}{f_{11}x(1-x)}dx.
\end{align}
\end{subequations}
Then, applying the conditions in Proposition \ref{Proposition9} and using the approximations in Proposition \ref{Proposition8} for the identity-by-descent measures yield our first result.

\paragraph{Result 1.}
\textit{
Assuming that the population-scaled means and variances of the payoffs are of the same small order, selection will favour the evolution of $S_1$ if
\begin{align}\label{sec1-case5-eq5}
&\frac{4+\nu}{3(3+\nu)}(\mu_{1,1}-\mu_{2,1})
+\frac{5+2\nu}{3(3+\nu)}(\mu_{1,2}-\mu_{2,2})
-\frac{288+127\nu+ 20\nu^2 + \nu^3}{10(3+\nu)(6+\nu)(10+\nu)}\sigma^2_{1,1}\nonumber\\
&-\frac{188+117\nu+20\nu^2+\nu^3}{10(3+\nu)(6+\nu)(10+\nu)}\sigma^2_{1,2}
+\frac{243+127\nu+20\nu^2+\nu^3}{15(3+\nu)(6+\nu)(10+\nu)}\sigma^2_{2,1}\nonumber\\
&+\frac{231+174\nu+35\nu^2+2\nu^3}{5(3+\nu)(6+\nu)(10+\nu)}\sigma^2_{2,2}>0,
\end{align}
disfavour the evolution of $S_2$ if
\begin{align}\label{sec1-case5-eq6}
&\frac{5+2\nu}{3(3+\nu)}(\mu_{1,1}-\mu_{2,1})
+\frac{4+\nu}{3(3+\nu)}(\mu_{1,2}-\mu_{2,2})
-\frac{231+174\nu+35\nu^2+2\nu^3}{5(3+\nu)(6+\nu)(10+\nu)}\sigma^2_{1,1}\nonumber\\
&-\frac{243+127\nu+20\nu^2+\nu^3}{15(3+\nu)(6+\nu)(10+\nu)}\sigma^2_{1,2}
+\frac{188+117\nu+20\nu^2+\nu^3}{10(3+\nu)(6+\nu)(10+\nu)}\sigma^2_{2,1}\nonumber\\
&+\frac{288+127\nu+ 20\nu^2 + \nu^3}{10(3+\nu)(6+\nu)(10+\nu)}\sigma^2_{2,2}>0,
\end{align}
and favour the evolution of $S_1$ more than the evolution of $S_2$ if
\begin{equation}\label{sec1-case5-eq7}
\mu_{1,1}+\mu_{1,2}-\mu_{2,1}-\mu_{2,2}
+\frac{\nu+5}{2\left(\nu+6\right)}(\sigma^2_{2,2}-\sigma^2_{1,1})+\frac{\nu+7}{6\left(\nu+6\right)}(\sigma^2_{2,1}-\sigma^2_{1,2})
>0.
\end{equation}
}

%Substituting these expressions into Eq. (\ref{sec4-eq6}), we find 
%\begin{align}\label{sec1-case5-eq3}
%F_1&\approx \frac{1}{Nd}+\frac{2}{Nd}\int_{0}^1\int_{0}^y\frac{M(\mathbf{v}(x))}{Q(\mathbf{v}(x))}dxdy\nonumber\\
%&\approx \frac{1}{Nd}+\frac{1}{2Nd}\Bigg[
%\frac{4+\nu}{3(3+\nu)}(\mu_{1,1}-\mu_{2,1})
%+\frac{5+2\nu}{3(3+\nu)}(\mu_{1,2}-\mu_{2,2})\nonumber\\
%&\quad-\frac{288+127\nu+ 20\nu^2 + \nu^3}{10(3+\nu)(6+\nu)(10+\nu)}\sigma^2_{1,1}
%-\frac{188+117\nu+20\nu^2+\nu^3}{10(3+\nu)(6+\nu)(10+\nu)}\sigma^2_{1,2}\nonumber\\
%&\quad+\frac{243+127\nu+20\nu^2+\nu^3}{15(3+\nu)(6+\nu)(10+\nu)}\sigma^2_{2,1}
%+\frac{231+174\nu+35\nu^2+2\nu^3}{5(3+\nu)(6+\nu)(10+\nu)}\sigma^2_{2,2}
%\Bigg]
%\end{align}
%and
%\begin{align}\label{sec1-case5-eq4}
%F_2&\approx \frac{1}{Nd}-\frac{2}{Nd}\int_{0}^1\int_{y}^1\frac{M(\mathbf{v}(x))}{Q(\mathbf{v}(x))}dxdy\nonumber\\
%&\approx \frac{1}{Nd}-\frac{1}{2Nd}\Bigg[
%\frac{5+2\nu}{3(3+\nu)}(\mu_{1,1}-\mu_{2,1})
%+\frac{4+\nu}{3(3+\nu)}(\mu_{1,2}-\mu_{2,2})\nonumber\\
%&\quad-\frac{231+174\nu+35\nu^2+2\nu^3}{5(3+\nu)(6+\nu)(10+\nu)}\sigma^2_{1,1}
%-\frac{243+127\nu+20\nu^2+\nu^3}{15(3+\nu)(6+\nu)(10+\nu)}\sigma^2_{1,2}\nonumber\\
%&\quad+\frac{188+117\nu+20\nu^2+\nu^3}{10(3+\nu)(6+\nu)(10+\nu)}\sigma^2_{2,1}
%+\frac{288+127\nu+ 20\nu^2 + \nu^3}{10(3+\nu)(6+\nu)(10+\nu)}\sigma^2_{2,2}
%\Bigg].
%\end{align}
%Here, we have used Eq. (\ref{sec5-eq2}) and Proposition  \ref{Proposition8}.
It should be noted that the coefficients of $\sigma^2_{1,1}$ and $\sigma^2_{1,2}$ in the above conditions are all negative, while the coefficients of $\sigma^2_{2,1}$ and $\sigma^2_{2,2}$ are all positive. As a consequence, reducing the population-scaled variance of any payoff to $S_1$ or increasing the population-scaled variance of any payoff to $S_2$ makes all the above conditions for the evolution of $S_1$ easier to be satisfied. This is even the only way to make the conditions possible to hold when the population-scaled means of the payoffs verify $\mu_{1,1}<\mu_{2,1}$ and $\mu_{1,2}<\mu_{2,2}$ as in a deterministic Prisoner's dilemma.

%will lead to an increase of $F_1$. Conversely, in the approximation of $F_2$, the coefficients of $\sigma^2_{1,1}$ and $\sigma^2_{1,2}$ are positive, while the coefficients of $\sigma^2_{2,1}$ and $\sigma^2_{2,2}$ are negative. Then, reducing the population-scaled variance of any payoff to $S_1$ or increasing the population-scaled variance of any payoff to $S_2$ will result in a decrease of $F_2$.
%As a result, this scenario will make it more favorable for selection to promote the evolution of cooperation in any sense.

As for the effect of the deme-scaled extinction rate, we note that the coefficient of  $\sigma^2_{2,2}-\sigma^2_{1,1}$ in the condition for selection to favour the evolution of $S_1$ more than the evolution of $S_2$ is an increasing positive function with respect to $\nu$, while the coefficient of $\sigma^2_{2,1}-\sigma^2_{1,2}$ is a decreasing positive function. This means that the effect of variability in payoffs to interacting offspring of the same type on this condition decreases as $\nu$ increases, while the effect of variability in payoffs to interacting offspring of different types decreases. This is surprising since, as $\nu$ increases, interacting offspring are less likely to be of the same type. %(TO BE EXPLAINED) 

%The reason for this is that an increase in $\nu$ leads to less homogeneous demes. Consequently, there will be an increase in the number of interactions between individuals playing different strategies and a decrease in the number of interactions between individuals playing the same strategy.

Finally, as $\nu\rightarrow\infty$, the above conditions for $F_1>(Nd)^{-1}$, $F_2<(Nd)^{-1}$ and $F_1>F_2$ become 
\begin{subequations}
\begin{align}
&\frac{1}{3}(\mu_{1,1}-\mu_{2,1})+\frac{2}{3}(\mu_{1,2}-\mu_{2,2})-\frac{1}{10}\sigma^2_{1,1}-\frac{1}{10}\sigma^2_{1,2}+\frac{1}{15}\sigma^2_{2,1}+\frac{2}{5}\sigma^2_{2,2}>0,\\
&\frac{2}{3}(\mu_{1,1}-\mu_{2,1})+\frac{1}{3}(\mu_{1,2}-\mu_{2,2})-\frac{2}{5}\sigma^2_{1,1}-\frac{1}{15}\sigma^2_{1,2}+\frac{1}{10}\sigma^2_{2,1}+\frac{1}{10}\sigma^2_{2,2}>0,\\
&\mu_{1,1}+\mu_{1,2}-\mu_{2,1}-\mu_{2,2}
+\frac{1}{2}(\sigma^2_{2,2}-\sigma^2_{1,1})+\frac{1}{6}(\sigma^2_{2,1}-\sigma^2_{1,2})
>0,
\end{align}
\end{subequations}
respectively.
Note that similar conditions were derived for well-mixed populations (see Eqs. ($38$, $39$, $40$) in Kroumi \textit{et al.} \cite{KMLL2021} for a Moran model, and Eqs. ($95$, $100$, $104$) in Li and Lessard \cite{LL2020} for a Wright-Fisher model). Note also that, if the population-scaled variances are all equal to $\sigma^2$, then the condition for selection to favour the evolution of $S_1$ more than the evolution of $S_2$ ($F_1>F_2$) does not depend on $\sigma^2$, which is not surprising by symmetry, but it becomes easier for selection to favour of the evolution of $S_1$ ($F_1>(Nd)^{-1}$) as well as the evolution of $S_2$  ($F_2>(Nd)^{-1}$) as $\sigma^2$ increases.%Thus, when the island-scaled extinction rate is high, the evolutionary dynamic in the island model is exactly to this one obtained in a population without structure.

%%%%%%%%%%%%%%%%%%%%%%%%%%%%%%%%%%%%%%%%%%%%%%%%%%%%%%%%%%%%%%%%%%%%%%%%%%%%%%%%%%%%%%%%%%%%%%%%%%%%%%%%%%%%%%%%

\section{Case 1: $\sigma^2_{1,1}=\sigma^2_{1,2}=\sigma^2_{2,1}=0$ and $\sigma^2_{2,2}=\sigma^2>0$}
Here, we suppose that the population-scaled variances of the payoffs different from $\sigma^2_{2,2}=\sigma^2>0$ are insignificant. In this case, we have 
\begin{equation}\label{sec5-case1-eq0}
\frac{M(\mathbf{v}(x))}{Q(\mathbf{v}(x))}=\frac{(\mu_{1,1}-\mu_{2,1})\psi(x)+(\mu_{1,2}-\mu_{2,2})\psi(1-x)+\sigma^2g(1-x)}{2f_{11}+x(1-x)\psi^2(1-x)\sigma^2},
\end{equation}
whose first derivative with respect to $\sigma^2$ is given by
\begin{equation}
%\frac{\partial \frac{M(\mathbf{v}(x))}{Q(\mathbf{v}(x))}}{\partial \sigma^2}=
\frac{2f_{11}g(1-x)+x(1-x)\psi^2(1-x)\Big[(\mu_{2,1}-\mu_{1,1})\psi(x)+(\mu_{2,2}-\mu_{1,2})\psi(1-x)\Big]}
{\left(2f_{11}+x(1-x)\sigma^2\psi^2(1-x)\right)^2}>0,
\end{equation}
for $x\in[0,1]$. 
Using Leibniz integral rule, it can be shown that for every $y\in(0,1]$, the functions 
\begin{align}
\sigma^2\longmapsto\int_{0}^{y}\frac{M(\mathbf{v}(x))}{Q(\mathbf{v}(x))}dx=-\frac{\ln(\Psi(y))}{2}
\end{align}
and 
\begin{align}
\sigma^2\longmapsto\int_{y}^{1}\frac{M(\mathbf{v}(x))}{Q(\mathbf{v}(x))}dx=\frac{\ln(\Phi(y))}{2}
\end{align}
 are strictly increasing on $[0,\infty)$, so that the functions $\sigma^2\longmapsto\Psi(y)$  and $\sigma^2\longmapsto -\Phi(y)$ are strictly 
 decreasing on $[0,\infty)$. %$\frac{\partial\Psi(y)}{\partial\sigma^2}<0$ and $\frac{\partial\Phi(y)}{\partial\sigma^2}>0$ for any $y\in(0,1)$. 
 Hence, an increase in $\sigma^2$ will increase the fixation probability $F_1$ and decrease the fixation probability $F_2$. In addition, increasing $\sigma^2$ will make it possible for selection to favour the evolution of $S_1$.

Using Fatou's lemma, we get
\begin{equation}
\liminf_{\sigma^2\rightarrow\infty}\int_{0}^{y}\frac{M(\mathbf{v}(x))}{Q(\mathbf{v}(x))}dx\geq\int_{0}^{y}\lim_{\sigma^2\rightarrow\infty}\frac{M(\mathbf{v}(x))}{Q(\mathbf{v}(x))}dx=\int_{0}^{y}\frac{g(1-x)}{x(1-x)\psi^2(1-x)}dx=\infty
\end{equation}
for every $y\in(0,1]$,
since 
\begin{equation}
\frac{g(1-x)}{x(1-x)\psi^2(1-x)}\sim \frac{g(1)}{x\psi^2(1)}
\end{equation} 
as $x\rightarrow0^{+}$.
Therefore, we have
\begin{equation}
\lim_{\sigma^2\rightarrow\infty}\int_{0}^{y}\frac{M(\mathbf{v}(x))}{Q(\mathbf{v}(x))}dx=\infty,
\end{equation}
from which 
\begin{equation}\label{sec5-case1-eq1}
\lim_{\sigma^2\rightarrow\infty}\Psi(y)=\lim_{\sigma^2\rightarrow\infty}\exp\left(-2\int_{0}^{y}\frac{M(\mathbf{v}(x))}{Q(\mathbf{v}(x))}dx\right)=0,
\end{equation}
for every $y\in(0,1]$.
On the other hand, when $\sigma^2=0$, we have $M(\mathbf{v}(x))/Q(\mathbf{v}(x))<0$ for $x\in[0,1]$, which implies that 
\begin{equation}\label{sec5-case1-eq2}
\Psi(y)|_{\sigma^2=0}>1\;\mbox{ for every   } y\in(0,1].
\end{equation} 
Combining  (\ref{sec5-case1-eq1}) and (\ref{sec5-case1-eq2}) for $y=1$, and using the fact that the function $\sigma^2\longmapsto\Psi(1)$ is continuous and strictly decreasing on $[0,\infty)$, we deduce that  the equation
\begin{equation}
\Psi(1)=1,
\end{equation}
with respect to $\sigma^2>0$, has a unique solution $\sigma^2_{**}(\nu)>0$,
so that $F_1>F_2$ for $\sigma^2>\sigma^2_{**}(\nu)$ and $F_1<F_2$ for $\sigma^2<\sigma^2_{**}(\nu)$.

Similarly, Fatou's lemma yields
\begin{equation}
\liminf_{\sigma^2\rightarrow\infty}\int_{y}^{1}\frac{M(\mathbf{v}(x))}{Q(\mathbf{v}(x))}dx\geq\int_{y}^{1}\lim_{\sigma^2\rightarrow\infty}\frac{M(\mathbf{v}(x))}{Q(\mathbf{v}(x))}dx=\int_{y}^{1}\frac{g(1-x)}{x(1-x)\psi^2(1-x)}dx=\infty
\end{equation}
for every $y\in[0,1)$,
since 
\begin{equation}
\frac{g(1-x)}{x(1-x)\psi^2(1-x)}\sim \frac{g(0)}{(1-x)\psi^2(0)}
\end{equation}
as $x\rightarrow1^{-}$,
which implies
\begin{equation}
\lim_{\sigma^2\rightarrow\infty}\int_{y}^{1}\frac{M(\mathbf{v}(x))}{Q(\mathbf{v}(x))}dx=\infty
\end{equation}
and
  \begin{equation}
\lim_{\sigma^2\rightarrow\infty}\Phi(y)=\lim_{\sigma^2\rightarrow\infty}\exp\left(2\int_{y}^{1}\frac{M(\mathbf{v}(x))}{Q(\mathbf{v}(x))}dx\right)=\infty.
\end{equation}
Using again Fatou's lemma, we have
\begin{equation}
\liminf_{\sigma^2\rightarrow\infty}\int_{0}^{1}\Phi(y)dy\geq\int_{0}^{1}\lim_{\sigma^2\rightarrow\infty}\Phi(x)dx=\infty,
\end{equation}
from which 
\begin{equation}\label{sec5-case1-eq3}
\lim_{\sigma^2\rightarrow\infty}\int_{0}^{1}\Phi(y)dy=\infty.
\end{equation}
On the other hand,  the fact that  $M(\mathbf{v}(x))/Q(\mathbf{v}(x))<0$ for $x\in[0,1]$ when $\sigma^2=0$ implies that
\begin{equation}\label{sec5-case1-eq4}
\int_{0}^{1}\Phi(y)dy|_{\sigma^2=0}<1.
\end{equation}
Combining  (\ref{sec5-case1-eq3}) and (\ref{sec5-case1-eq4}), and using the fact that the function $\sigma^2\longmapsto\int_{0}^{1}\Phi(y)dy$ is continuous and strictly increasing on $[0,\infty)$, the equation
\begin{equation}
\int_{0}^{1}\Phi(y)dy=1
\end{equation}
with respect to $\sigma^2>0$,
has a unique solution $\sigma^2_{***}(\nu)>0$,
so that $F_2>(Nd)^{-1}$ for $\sigma^2<\sigma^2_{***}(\nu)$ and $F_2<(Nd)^{-1}$ for $\sigma^2>\sigma^2_{***}(\nu)$.

Note that, owing to Eq. (\ref{sec5-case1-eq0}), there exists a $\sigma^2_0>0$ such that if 
$\sigma^2>\sigma^2_0$, then $M(\mathbf{v}(x))>0$ for $x\in[0,1]$, from which $|\Psi(y)|<1$ for every $y\in (0,1]$. Then, adding this conclusion to Eq. (\ref{sec5-case1-eq1}), the Lebesgue dominated convergence theorem ensures that
\begin{equation}\label{sec5-case1-eq5}
\lim_{\sigma^2\rightarrow\infty}\int_{0}^{1}\Psi(y)dy=0.
\end{equation}
On the other hand, a consequence of Eq. (\ref{sec5-case1-eq2}) is that
\begin{equation}\label{sec5-case1-eq6}
\int_{0}^{1}\Psi(y)dy|_{\sigma^2=0}>1.
\end{equation}
Combining  (\ref{sec5-case1-eq5}) and (\ref{sec5-case1-eq6}), and using the fact that the function $\sigma^2\longmapsto\int_{0}^{1}\Psi(y)dy$ is continuous and strictly decreasing on $[0,\infty)$, the equation
\begin{equation}
\int_{0}^{1}\Psi(y)dy=1
\end{equation}
with respect to $\sigma^2>0$,
has a unique solution $\sigma^2_{*}(\nu)>0$,
so that  $F_1<(Nd)^{-1}$ for $\sigma^2<\sigma^2_{*}(\nu)$ and $F_1>(Nd)^{-1}$ for $\sigma^2>\sigma^2_{*}(\nu)$.

Next, we study the relationships between the threshold values $\sigma^2_{*}(\nu)$, $\sigma^2_{**}(\nu)$ and $\sigma^2_{***}(\nu)$ for $\sigma^2=\sigma^2_{2,2}$ and the effect of the deme-scaled extinction rate $\nu$ on these values when $\sigma^2$ and the population-scaled means of all the payoffs are of the same small order, while $\sigma^2_{1,1}=\sigma^2_{1,2}=\sigma^2_{2,1}=0$.
By using Eqs. (\ref{sec1-case5-eq5})-(\ref{sec1-case5-eq7}), the threshold values are found to be
\begin{subequations}\label{sec5-case1-eq7}
\begin{align}
\sigma^2_{*}(\nu)&=\frac{5(6+\nu)(10+\nu)}{3(231+174\nu+35\nu^2+2\nu^3)}\Big[(4+\nu)(\mu_{2,1}-\mu_{1,1})+(5+2\nu)
(\mu_{2,2}-\mu_{1,2})\Big],\\
\sigma^2_{**}(\nu)&=\frac{2(6+\nu)}{5+\nu}\Big(\mu_{2,1}+\mu_{2,2}-\mu_{1,1}-\mu_{1,2}\Big),\\
\sigma^2_{***}(\nu)&=\frac{10(6+\nu)(10+\nu)}{3(288+127\nu+20\nu^2+\nu^3)}\Big[(5+2\nu)(\mu_{2,1}-\mu_{1,1})+(4+\nu)
(\mu_{2,2}-\mu_{1,2})\Big].
\end{align}
\end{subequations}
It a worth noting, which can be easily checked, that the coefficients of $\mu_{2,1}-\mu_{1,1}$ and $\mu_{2,2}-\mu_{1,2}$ in the first two expressions above %$\nu\longmapsto\sigma^2_{*}(\nu)$ and $\nu\longmapsto\sigma^2_{**}(\nu)$ 
are decreasing functions with respect to $\nu$ in $(0,\infty)$, while those on the third %$\nu\longmapsto\sigma^2_{***}(\nu)$ is an
are increasing functions. %This can be deduced by analyzing the coefficients of $\mu_{2,1}-\mu_{1,1}$ and $\mu_{2,2}-\mu_{1,2}$ in Eq. (\ref{sec5-case1-eq7}).
Therefore, in the case where $\mu_{2,1}-\mu_{1,1}>0$ and $\mu_{2,2}-\mu_{1,2}>0$ as in a Prisoner's dilemma, the conditions for $F_1>(Nd)^{-1}$ and $F_1 > F_2$, given by $\sigma^2> \sigma^2_{*}(\nu)>0$ and $\sigma^2> \sigma^2_{**}(\nu)>0$, respectively, become less stringent as $\nu$ increases, while the condition for $F_2<(Nd)^{-1}$,  given by $\sigma^2> \sigma^2_{***}(\nu)>0$, becomes more stringent. 
%This implies that increasing the deme-scaled extinction rate will make it more difficult for $F_2$ to exceed $(Nd)^{-1}$ and easier for both $F_1$ to exceed $(Nd)^{-1}$ and $F_1$ to be greater than $F_2$.
Furthermore, in this case, we have that
\begin{align}
&\sigma^2_{***}(\nu)-\sigma^2_{**}(\nu)\nonumber\\
&=\frac{2(1+\nu)(6+\nu)\Big[(136+33\nu+2\nu^2)(\mu_{2,1}-\mu_{1,1})  + (386+108\nu+7\nu^2)(\mu_{2,2}-\mu_{1,2})\Big]}{3(5+\nu)(288+127\nu+20\nu^2+\nu^3)}
\end{align}
and
\begin{align}
&\sigma^2_{**}(\nu)-\sigma^2_{*}(\nu)\nonumber\\
&=\frac{
(1+\nu) (6+\nu)\Big[(386 + 108\nu + 7\nu^2)(\mu_{2,1}-\mu_{1,1})
+(136 + 33 \nu + 2 \nu^2)(\mu_{2,2}-\mu_{1,2})\Big]
}{3(5+\nu) (231+174\nu+35\nu^2+2\nu^3)}
\end{align}
are positive,
from which
\begin{equation}
\sigma^2_{*}(\nu)<\sigma^2_{**}(\nu)<\sigma^2_{***}(\nu),
\end{equation}
for $\nu>0$. These inequalities imply that the conditions for $F_1<(Nd)^{-1}$ and $F_2<(Nd)^{-1}$ cannot occur simultaneously for any $\sigma_{2}^2>0$ and any $\nu>0$.

The above properties are illustrated in Figure \ref{figure1} in the special case  where 
the population-scaled means of the payoffs are given by $\mu_{1,1}=\mu_b-\mu_c$, $\mu_{1,2}=-\mu_c$, $\mu_{2,1}=\mu_b$ and $\mu_{2,2}=0$, where $\mu_b>\mu_c>0$, which corresponds to an additive Prisoner's dilemma. 

%Upon substituting these values in Eq. (\ref{sec5-case1-eq7}), we can derive the expressions
%\begin{subequations}\label{sec5-case1-eq8}
%\begin{align}
%\sigma^2_{*}(\nu)&=\frac{5(3+\nu)(6+\nu)(10+\nu)}{231+174\nu+35\nu^2+2\nu^3}\mu_{c},\\
%\sigma^2_{**}(\nu)&=\frac{4(6+\nu)}{5+\nu}\mu_{c},\\
%\sigma^2_{***}(\nu)&=\frac{10(3+\nu)(6+\nu)(10+\nu)}{288+127\nu+20\nu^2+\nu^3}\mu_{c}.
%\end{align}
%\end{subequations}
%It is easy to check that $\nu\longmapsto\sigma^2_{*}(\nu)$ and $\nu\longmapsto\sigma^2_{**}(\nu)$ are decreasing functions on $(0,\infty)$, while $\nu\longmapsto\sigma^2_{***}(\nu)$ is an increasing function on the same interval (see Figure \ref{figure1}). Thus, increasing the deme-scaled extinction rate will make it more difficult for $F_2$ to exceed $(Nd)^{-1}$ and easier for both $F_1$ to exceed $(Nd)^{-1}$ and $F_1$ to be greater than $F_2$. 

%%%%%%%%%%%%%%%%%%%%%%%%%%%%%%%%%%%%%%%%%%%%%%%%%%%%%%%%%%%%%%%%%%%%%%%%%%%%%%%%%%%%%%%%%%%%%%%%%%%%%%%%%%%%%%%%%%%%%%%%%%%%%%%%%%%%%%%%%%%%%%%%%%%%%%%%%%%%%%%%%%%%%%%%%%%%%%%%%%%%%%%%%%%%%%%%%%%%%%%%%%%%%%%%%%%%%%%%%%%%%%%%%%%%%%%%%%%%%%%%%%%%%%%%%%%%%%%%%%%%%%%%%%%%%%%%%%%%%%%%%%%%%%%%%%%%%%%%%%%%%%%
\begin{figure}[hbt!]
\includegraphics[height=6cm, width=10cm]{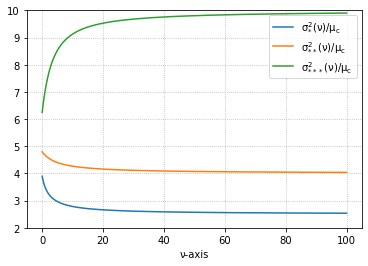}
\caption{\textit{The curves represent the threshold values $\sigma^2_{*}(\nu)/\mu_c$, $\sigma^2_{**}(\nu)/\mu_c$ and $\sigma^2_{***}(\nu)/\mu_c$ with respect to the deme-scaled extinction rate $\nu$ in Case 1 where $\sigma^2_{1,1}=\sigma^2_{1,2}=\sigma^2_{2,1}=0$, $\sigma^2_{2,2}=\sigma^2>0$ with $\mu_{1,1}=\mu_b-\mu_c$, $\mu_{1,2}=-\mu_c$, $\mu_{2,1}=\mu_b$, $\mu_{2,2}=0$ and $\mu_b>\mu_c>0$. These curves indicate the minimum values that $\sigma_{}^2/\mu_c$ must exceed in order to have $F_1>(Nd)^{-1}$, $F_1 > F_2$ and  $F_2<(Nd)^{-1}$, respectively.
Since the inequalities 
$\sigma^2_{*}(\nu)<\sigma^2_{**}(\nu)<\sigma^2_{***}(\nu)$ hold for every $\nu>0$, the scenario where $F_1<(Nd)^{-1}$ and $F_2<(Nd)^{-1}$ cannot occur.
}}
\label{figure1}
\end{figure}

Let us summarize our findings in Case 1.

%%%%%%%%%
%%%%%%%%%
\paragraph{Result 2.}
\textit{
If the population-scaled variance of the payoff to $S_2$ against $S_2$, represented by $\sigma^2_{}$, is large enough, while the population-scaled variances  of the other payoffs are negligible, then the inequalities $F_1>(Nd)^{-1}$, $F_1 > F_2$ and $F_2<(Nd)^{-1}$ will hold. Moreover, when the population-scaled means of the payoffs to $S_1$ against $S_1$ and $S_2$ are less than the corresponding ones to $S_2$ and all the population-scaled parameters are of the same small order, these inequalities will occur in the listed order as $\sigma^2$ increases, earlier for the first two and later in the third one as the deme-scaled extinction rate $\nu$ increases.
%increase the fixation probability of $S_1$ introduced as a single mutant ($F_1$) and decrease the corresponding fixation probability for $S_2$ ($F_2$). Moreover, there exists $\sigma_0^2$ such that, if $\sigma^2_{}>\sigma_0^2$, then selection fully favours the evolution of $S_1$ in the sense that $F_1>(Nd)^{-1}>F_2$.
}
 %%%%%%%%%
 %%%%%%%%%

%%%%%%%%%%%%%%%%%%%%%%%%%%%%%%%%%%%%%%%%%%%%%%%%%%%%%%%%%%%%%%%%%%%%%%%%%%%%%%%%%%%%%%%%%%%%%%%%%%%%%%%%%%%%%%%%%%%%%%%%%%%%%%%%%%%%%%%%%%%%%%%%%%%%%%%%%%%%%%%%%%%%%%%%%%%%%%%%%%%%%%%%%%%%%%%%%%%%%%%%%%%%%%%%%%%%%%%%%%%%%%%%%%%%%%%%%%%%%%%%%%%%%%%%%%%%%%%%%%%%%%%%%%%%%%%%%%%%%%%%%%%%%%%%%%%%%%%%%%%%%%%%

\section{Case 2: $\sigma^2_{1,1}=\sigma^2_{1,2}=\sigma^2_{2,2}=0$ and $\sigma^2_{2,1}=\sigma^2>0$}

Proceeding as previously in the case where the population-scaled variances of the payoffs different from $\sigma^2_{2,1}=\sigma^2>0$ are insignificant, it can be shown that an increase in $\sigma^2$ increases $F_1$ and decreases $F_2$. Moreover, there exist three threshold values such that
%Now, suppose that all other population-scaled variances of the different payoffs are insignificant compared to $\sigma^2_{2,1}$. In such a case, we have 
%\begin{equation}\label{sec5-case2-eq1}
%\Psi(x)=\frac{(\mu_{1,1}-\mu_{2,1})\psi(x)+(\mu_{1,2}-\mu_{2,2})\psi(1-x)+\sigma^2h(1-x)}
%{2f_{11}+x(1-x)\psi^2(x)\sigma^2},
%\end{equation}
%where its first derivative satisfies
%\begin{equation}\label{sec5-case2-eq2}
%\begin{split}
%\frac{\partial \Psi(x)}{\partial \sigma^2}&=\frac{2f_{11}h(1-x)+x(1-x)\psi^2(x)\Big[(\mu_{2,1}-\mu_{1,1})\psi(x)+(\mu_{2,2}-\mu_{1,2})\psi(1-x)\Big]}
%{\left(2f_{11}+x(1-x)\sigma^2\psi^2(1-x)\right)^2}>0,
%\end{split}
%\end{equation}
%for any $x\in[0,1]$. 
%
%Using a similar approach to the previous case, the function $\sigma^2\longmapsto F_1$ (resp. $\sigma^2\longmapsto F_2$) is an increasing (resp. decreasing) on $(0,\infty)$. Moreover, there exist three positive threshold values of $\sigma^2$, denoted as $\sigma^2_{*}(\nu)$, $\sigma^2_{**}(\nu)$, and $\sigma^2_{***}(\nu)$, such that:
\begin{subequations}\label{sec5-case2-eq3}
\begin{align}
&F_1>(Nd)^{-1}\mbox{   as soon as  }\sigma^2> \sigma^2_{*}(\nu),\\
&F_1>F_2\mbox{  as soon as  }\sigma^2> \sigma^2_{**}(\nu),\\
&F_2<(Nd)^{-1}\mbox{   as soon as  }\sigma^2> \sigma^2_{***}(\nu).
\end{align}
\end{subequations} 
If $\sigma^2$ and the population-scaled means of all the payoffs are of the same small order, then the threshold values are given by
%We will now examine how $\sigma^2_{*}(\nu)$, $\sigma^2_{**}(\nu)$, and $\sigma^2_{***}(\nu)$ are related and how they are affected by $\nu$, particularly when the scaled means and variances are of the same small order. By using Eqs. (\ref{sec1-case5-eq5})-(\ref{sec1-case5-eq7}) with the new expressions, we have
\begin{subequations}\label{sec5-case2-eq4}
\begin{align}
\sigma^2_{*}(\nu)&=\frac{5(6+\nu)(10+\nu)}{243+127\nu+20\nu^2+\nu^3}\Big[(4+\nu)(\mu_{2,1}-\mu_{1,1})+(5+2\nu)
(\mu_{2,2}-\mu_{1,2})\Big],\\
\sigma^2_{**}(\nu)&=\frac{6(6+\nu)}{7+\nu}\Big(\mu_{2,1}+\mu_{2,2}-\mu_{1,1}-\mu_{1,2}\Big),\\
\sigma^2_{***}(\nu)&=\frac{10(6+\nu)(10+\nu)}{3(188+117\nu+20\nu^2+\nu^3)}\Big[(5+2\nu)(\mu_{2,1}-\mu_{1,1})+(4+\nu)
(\mu_{2,2}-\mu_{1,2})\Big].
\end{align}
\end{subequations}
Even in the  case where $\mu_{2,1}-\mu_{1,1}>0$ and $\mu_{2,2}-\mu_{1,2}>0$, the order of these threshold values is not clear. In the particular case where
$\mu_{1,1}=\mu_b-\mu_c$, $\mu_{1,2}=-\mu_c$, $\mu_{2,1}=\mu_b$ and $\mu_{2,2}=0$ with $\mu_b>\mu_c>0$, however, we have
\begin{subequations}
\begin{align}
\sigma^2_{*}(\nu)-\sigma^2_{**}(\nu)&=\frac{3(1+\nu)(6+\nu)^2(13+\nu)}{(7+\nu)(243+127\nu+20\nu^2+\nu^3)}\mu_c>0,\\
 \sigma^2_{**}(\nu)-\sigma^2_{***}(\nu)&=\frac{2(1+\nu)(6+\nu)^2(13+\nu)}{(7+\nu)(188+117\nu+20\nu^2+\nu^3)}\mu_c>0,
\end{align}
\end{subequations}
so that
\begin{equation}
\sigma^2_{***}(\nu)<\sigma^2_{**}(\nu)<\sigma^2_{*}(\nu),
\end{equation}
for $\nu>0$. These inequalities imply that the conditions for $F_1>(Nd)^{-1}$ and $F_2>(Nd)^{-1}$ cannot simultaneously occur. Moreover, while $\sigma^2_{*}(\nu)$ and $\sigma^2_{**}(\nu)$ are monotonic functions, actually monotonically increasing functions, with respect to $\nu$ in the whole interval $(0, \infty)$, 
this is not necessarily the case for $\sigma^2_{***}(\nu)$ which can decrease first and then increase to a global maximum as $\nu$ increases. This is illustrated in Figure \ref{figure2}. 

We are now ready to state our next result.

\paragraph{Result 3.}
\textit{
If the population-scaled variance of the payoff to $S_2$ against $S_1$, represented by $\sigma^2_{}$, is large enough, while the population-scaled variances  of the other payoffs are negligible, then the inequalities $F_1>(Nd)^{-1}$, $F_1 > F_2$ and $F_2<(Nd)^{-1}$ will hold. Moreover, when all the population-scaled parameters are of the same small order, these inequalities may occur in the reversed order as $\sigma^2$ increases, and all later  as the deme-scaled extinction rate $\nu$ tends to infinity.
%Increasing $\sigma^2_{2,1}$, the scaled variance of a defector's payoff when interacting with a cooperator, in the absence of fluctuations in other payoffs, will result in an increase in the fixation probability of cooperators ($F_1$) and a decrease in the fixation probability of defectors ($F_2$). This shift towards higher $F_1$ and lower $F_2$ may make it possible for selection to fully favor the evolution of cooperation. In particular, there exists a constant value $\sigma_0^2$ such that for any $\sigma^2_{2,1}>\sigma_0^2$, the evolution of cooperation is fully favored by selection, with $F_1>(Nd)^{-1}>F_2$.
}

%%%%%%%%%%%%%%%%%%%%%%%%%%%%%%%%%%%%%%%%%%%%%%%%%%%%%%%%%%%%%%%%%%%%%%%%%%%%%%%%%%%%%%%%%%%
%%%%%%%%%%%%%%%%%%%%%%%%%%%%%%%%%%%%%%%%%%%%%%%%%%%%%%%%%%%%%%%%%%%%%%%%%%%%%%%%%%%%%%%%%%%%%%%%%%%%%%%%%%%%%%%%%%%%%%%%%%%%%%%%%%%%%%%%%%%%%%%%%%%%%%%%%%%%%%%%%%%%%%%%%%%%%%%%%%%%%%%%%%%%%%%%%%
%%%%%%%%%%%%%%%%%%%%%%%%%%%%%%%%%%%%%%%%%%%%%%%%%%%%%%%%%%%%%%%%%%%%%%%%%%%%%%%%%%%%%%%%%%%%%%%%%
\begin{figure}[hbt!]
\includegraphics[height=6cm, width=10cm]{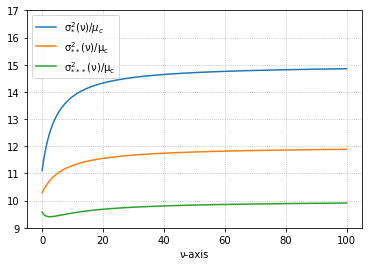}
\caption{\textit{The curves represent the minimum threshold values $\sigma^2_{*}(\nu)/\mu_c$, $\sigma^2_{**}(\nu)/\mu_c$ and $\sigma^2_{***}(\nu)/\mu_c$ for $\sigma_{}^2/\mu_c$ to have $F_1>(Nd)^{-1}$, $F_1 > F_2$ and  $F_2<(Nd)^{-1}$, respectively, in Case 2 where $\sigma^2_{1,1}=\sigma^2_{1,2}=\sigma^2_{2,2}=0$, $\sigma^2_{2,1}=\sigma^2>0$ with $\mu_{1,1}=\mu_b-\mu_c$, $\mu_{1,2}=-\mu_c$, $\mu_{2,1}=\mu_b$, $\mu_{2,2}=0$ and $\mu_b>\mu_c>0$.  Note that the worst scenario for the three conditions to hold occurs when the deme-scaled extinction rate $\nu\rightarrow\infty$.
%the best scenario for the first two conditions to hold occurs when $\nu\rightarrow0$, representing a population structured into isolated demes. On the other hand, the worst scenario for $F_1>(Nd)^{-1}$, $F_1>F_2$, and $F_2<(Nd)^{-1}$ arises when $\nu\rightarrow\infty$, representing a population without structure.
}}
\label{figure2}
\end{figure}
%%%%%%%%%%%%%%%%%%%%%%%%%%%%%%%%%%%%%%%%%%%%%%%%%%%%%%%%%%%%%%%%%%%%%%%%%%%%%%%%%%%%%%%%%%%
%%%%%%%%%%%%%%%%%%%%%%%%%%%%%%%%%%%%%%%%%%%%%%%%%%%%%%%%%%%%%%%%%%%%%%%%%%%%%%%%%%%%%%%%%%%%%%%%%%%%%%%%%%%%%%%%%%%%%%%%%%%%%%%%%%%%%%%%%%%%%%%%%%%%%%%%%%%%%%%%%%%%%%%%%%%%%%%%%%%%%%%%%%%%%%%%%%
%%%%%%%%%%%%%%%%%%%%%%%%%%%%%%%%%%%%%%%%%%%%%%%%%%%%%%%%%%%%%%%%%%%%%%%%%%%%%%%%%%%%%%%%%%%%%%%%%

\section{Case 3: $\sigma^2_{1,1}=\sigma^2_{1,2}=0$ and $\sigma^2_{2,2}=\sigma^2_{2,1}=\sigma^2>0$}

Finally, in the case where the population-scaled variances of the payoffs to $S_1$ against $S_1$ and $S_2$ are negligible compared to those to $S_2$ given by $\sigma^2>0$, it can be shown as previously that $F_1$ increases and $F_2$ decreases as $\sigma^2$ increases such that $F_1>(Nd)^{-1}$, $F_1 > F_2$ and  $F_2<(Nd)^{-1}$ occur as soon as $\sigma^2 > \sigma^2_{*}(\nu),$ $\sigma^2_{**}(\nu)$ and $\sigma^2_{***}(\nu)$, respectively. Moreover, if the population-scaled parameters are of the same small order, then the threshold values are given by
\begin{subequations}\label{sec5-case3-eq4}
\begin{align}
\sigma^2_{*}(\nu)&=\frac{5(6+\nu)(10+\nu)}{936+ 649\nu+125\nu^2+7\nu^3}\Big[(4+\nu)(\mu_{2,1}-\mu_{1,1})+(5+2\nu)
(\mu_{2,2}-\mu_{1,2})\Big],\\
\sigma^2_{**}(\nu)&=\frac{3(6+\nu)}{11+2\nu}\Big(\mu_{2,1}+\mu_{2,2}-\mu_{1,1}-\mu_{1,2}\Big),\\
\sigma^2_{***}(\nu)&=\frac{5(6+\nu)(10+\nu)}{3(238+122\nu+20\nu^2+\nu^3)}\Big[(5+2\nu)(\mu_{2,1}-\mu_{1,1})+(4+\nu)
(\mu_{2,2}-\mu_{1,2})\Big].
\end{align}
\end{subequations}
%coefficients of $\mu_{2,1}-\mu_{1,1}$ and $\mu_{2,2}-\mu_{1,2}$ in the first two expressions above %$\nu\longmapsto\sigma^2_{*}(\nu)$ and $\nu\longmapsto\sigma^2_{**}(\nu)$ 
%are decreasing functions with respect to $\nu$ in $(0,\infty)$, while those on the third %$\nu\longmapsto\sigma^2_{***}(\nu)$ is an
%are increasing functions.
Here, the coefficients of $\mu_{2,1}-\mu_{1,1}$ and $\mu_{2,2}-\mu_{1,2}$ in the expressions of
$\sigma^2_{*}(\nu)$ and $\sigma^2_{**}(\nu)$ are decreasing functions on $(0,\infty)$, while the first coefficient is increasing and the second decreasing in the expression of
$\sigma^2_{***}(\nu)$. 
Moreover, when $\mu_{2,1}>\mu_{1,1}$ and $\mu_{2,2}>\mu_{1,2}$, the differences
%\begin{subequations}\label{sec5-case3-eq5}
\begin{align}\label{sec5-case3-eq5}
&\sigma^2_{**}(\nu)-\sigma^2_{*}(\nu)\nonumber\\
&\quad=\frac{(1+\nu)(6+\nu)\Big[(608+169\nu+11\nu^2)(\mu_{2,1}-\mu_{1,1})+(58+14\nu+\nu^2)(\mu_{2,2}-\mu_{1,2})\Big]}
{(11+2\nu)(936+649\nu+125\nu^2+7\nu^3)}
%%%%
%%%%
%&\sigma^2_{***}(\nu)-\sigma^2_{**}(\nu)\nonumber\\
%&\quad=\frac{
%(1+\nu) (6+\nu)\Big[(608+169\nu+11\nu^2)(\mu_{2,1}-\mu_{1,1})
%+(58+14\nu+\nu^2)(\mu_{2,2}-\mu_{1,2})\Big]
%}{3(11+2\nu)(238+122\nu+20\nu^2+\nu^3)},
\end{align}
and
\begin{align}\label{sec5-case3-eq5bis}
%&\sigma^2_{**}(\nu)-\sigma^2_{*}(\nu)\nonumber\\
%&\quad=\frac{(1+\nu)(6+\nu)\Big[(608+169\nu+11\nu^2)(\mu_{2,1}-\mu_{1,1})+(58+14\nu+\nu^2)(\mu_{2,2}-\mu_{1,2})\Big]}
%{(11+2\nu)(936+649\nu+125\nu^2+7\nu^3)},\\
%%%%
%%%%
&\sigma^2_{***}(\nu)-\sigma^2_{**}(\nu)\nonumber\\
&\quad=\frac{
(1+\nu) (6+\nu)\Big[(608+169\nu+11\nu^2)(\mu_{2,1}-\mu_{1,1})
+(58+14\nu+\nu^2)(\mu_{2,2}-\mu_{1,2})\Big]
}{3(11+2\nu)(238+122\nu+20\nu^2+\nu^3)}
\end{align}
%\end{subequations}
are both positive, from which  we have
\begin{equation}\label{sec5-case3-eq6}
\sigma^2_{*}(\nu)<\sigma^2_{**}(\nu)<\sigma^2_{***}(\nu),
\end{equation}
for $\nu>0$. This is similar to Case 1.

Figure \ref{figure3} shows the curves of $\sigma^2_{*}(\nu)/\mu_c$, $\sigma^2_{**}(\nu)/\mu_c$ and $\sigma^2_{***}(\nu)\mu_c$
in the particular case where
$\mu_{1,1}=\mu_b-\mu_c$, $\mu_{1,2}=-\mu_c$, $\mu_{2,1}=\mu_b$ and $\mu_{2,2}=0$ with $\mu_b>\mu_c>0$. There is little difference with Figure \ref{figure1} in Case 1, which suggests that the variability in the payoff to $S_2$ against $S_2$ is more important than the variability in the payoff to $S_2$ against $S_1$ in conditions for the evolution of $S_1$. 
%Consider the scenario in which the scaled means are described by the additive prisoner's dilemma. The quantities in 
%Eq. (\ref{sec5-case3-eq4}) become
%\begin{subequations}\label{sec5-case3-eq7}
%\begin{align}
%\sigma^2_{*}(\nu)&=\frac{15(3+\nu)(6+\nu)(10+\nu)}{936+ 649\nu+125\nu^2+7\nu^3}\mu_{c},\\
%\sigma^2_{**}(\nu)&=\frac{6(6+\nu)}{11+2\nu}\mu_{c},\\
%\sigma^2_{***}(\nu)&=\frac{5(3+\nu)(6+\nu)(10+\nu)}{(238+122\nu+20\nu^2+\nu^3)}\mu_{c}.
%\end{align}
%\end{subequations}
%It is clear that the functions $\nu\longmapsto\sigma^2_{*}(\nu)$ and $\nu\longmapsto\sigma^2_{**}(\nu)$ are decreasing, while $\nu\longmapsto\sigma^2_{}(\nu)$ is increasing. Consequently, as $\nu$ increases, the conditions for $F_1>(Nd)^{-1}$ and $F_1>F_2$ become less stringent, while the condition for $F_2<(Nd)^{-1}$ becomes stronger.
%See Figure \ref{figure3}, which shows the curves of $\sigma^2_{*}(\nu)/\mu_c$, $\sigma^2_{**}(\nu)/\mu_c$, and $\sigma^2_{***}(\nu)\mu_c$. 

\paragraph{Result 4.}
\textit{
The conditions for $F_1>(Nd)^{-1}$, $F_1 > F_2$ and $F_2<(Nd)^{-1}$  in Case 3, where
 the population-scaled variances  of the payoffs to $S_2$ against $S_1$ and $S_2$ are equal to $\sigma^2>0$, while the population-scaled variances  of the payoffs to $S_1$ against $S_1$ and $S_2$ are negligible in comparison, are similar to the conditions obtained in Case 1.
%Increasing 
%$\sigma^2_{2,2}$, the scaled variance of the payoff to a defector when it interacts with another defector, or
%$\sigma^2_{2,1}$, the scaled variance of a defector's payoff when interacting with a cooperator, in the absence of fluctuations in other payoffs, will increase the fixation probability of cooperators ($F_1$) and decrease in the fixation probability of defectors ($F_2$).
}

%%%%%%%%%%%%%%%%%%%%%%%%%%%%%%%%%%%%%%%%%%%%%%%%%%%%%%%%%%%%%%%%%%%%%%%%%%%%%%%%%%%%%%%%%%%
%%%%%%%%%%%%%%%%%%%%%%%%%%%%%%%%%%%%%%%%%%%%%%%%%%%%%%%%%%%%%%%%%%%%%%%%%%%%%%%%%%%%%%%%%%%%%%%%%%%%%%%%%%%%%%%%%%%%%%%%%%%%%%%%%%%%%%%%%%%%%%%%%%%%%%%%%%%%%%%%%%%%%%%%%%%%%%%%%%%%%%%%%%%%%%%%%%
%%%%%%%%%%%%%%%%%%%%%%%%%%%%%%%%%%%%%%%%%%%%%%%%%%%%%%%%%%%%%%%%%%%%%%%%%%%%%%%%%%%%%%%%%%%%%%%%%
\begin{figure}[hbt!]
\includegraphics[height=6cm, width=10cm]{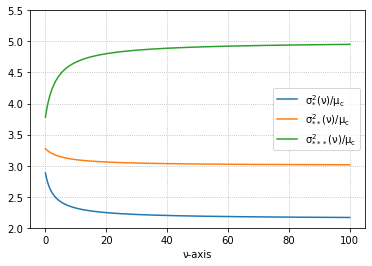}
\caption{\textit{The plotted curves represent the minimum threshold values $\sigma^2_{*}(\nu)/\mu_c$, $\sigma^2_{**}(\nu)/\mu_c$ and $\sigma^2_{***}(\nu)/\mu_c$ for $\sigma_{}^2/\mu_c$ to have $F_1>(Nd)^{-1}$, $F_1 > F_2$ and  $F_2<(Nd)^{-1}$, respectively, in  Case 3 where $\sigma^2_{1,1}=\sigma^2_{1,2}=0$, $\sigma^2_{2,2}=\sigma^2_{2,1}=\sigma^2>0$ with $\mu_{1,1}=\mu_b-\mu_c$, $\mu_{1,2}=-\mu_c$, $\mu_{2,1}=\mu_b$, $\mu_{2,2}=0$ and $\mu_b>\mu_c>0$. 
%These thresholds indicate the values that $\sigma_{2,1}^2/\mu_c=\sigma_{2,2}^2/\mu_c$ must exceed in order for selection to favor the evolution of $S_1$, to favor the evolution of $S_1$ more than the evolution of $S_2$, or to disfavor the evolution of $S_2$. 
Like in Case 1, the best scenario for the first two conditions to hold occurs when $\nu\rightarrow\infty$, while the best scenario for the third condition occurs when $\nu\rightarrow0$.}
}
\label{figure3}
\end{figure}
%%%%%%%%%%%%%%%%%%%%%%%%%%%%%%%%%%%%%%%%%%%%%%%%%%%%%%%%%%%%%%%%%%%%%%%%%%%%%%%%%%%%%%%%%%%%%%%%%%%%%%%%%%%%%%%%%%%%%%%%%%%%%%%%%%%%%%%%%%%%%%%%%%%%%%%%%%%%%%%%%%%%%%%%%%%%%%%%%%%%%%%%%%%%%%%%%%
%%%%%%%%%%%%%%%%%%%%%%%%%%%%%%%%%%%%%%%%%%%%%%%%%%%%%%%%%%%%%%%%%%%%%%%%%%%%%%%%%%%%%%%%%%%%%%%%%%%%%%%%%%%%%%%%%%%%%%%%%%%%%%%%%%%%%%%%%%%%%%%%%%%%%%%%%%%%%%%%%%%%%%%%%%%%%%%%%%%%%%%%%%%%%%%%%%

\section{Discussion}

Previous studies examining the island model primarily focused on constant payoffs. This paper examines a haploid population that is subdivided into $d$ demes, each containing $N\geq2$ individuals. The payoffs received by offspring using the strategies $S_1$ and $S_2$ in random pairwise interactions within demes are subject to stochastic fluctuations from one generation to the next such that their means and variances are proportional to the inverse of the population size $Nd$. Following viability selection within demes and allowing for local extinction and uniform recolonization with probability $m\in (0, 1)$,  the next generation starts with $N$ offspring chosen at random in each deme according to a Wright-Fisher model.

Having examined the conditions outlined in Ethier and Nagylaki (\cite{EN1980}), we have shown that the exact discrete-time Markov chain, which is determined by the deme type frequencies, involves two time scales so that it converges to a diffusion process for the strategy frequencies as $d$ goes to infinity. This convergence allows us to compute the probability for strategy $S_1$ to become fixed in the population following its introduction as a single mutant, represented by $F_1$, and the corresponding fixation probability for $S_2$, represented by $F_2$.  As a result, we have examined how stochastic fluctuations in the payoffs  impact the evolutionary dynamics. It is important to highlight that the infinitesimal mean and variance of the limiting diffusion process are expressed in terms of the population-scaled means and variances of the payoffs, represented by $\mu_{k,l}$ and $\sigma^2_{k,l}$, respectively, for $k,l=1,2$, with coefficients that depend on identity-by-descent measures between offspring under neutrality besides. These have been deduced exactly and shown to depend only on the deme-scaled extinction rate $\nu=Nm$ for $N$ large enough and $m$ small enough (see Appendices H and J). 
It is worth noting that if $\nu\rightarrow\infty$, the limiting diffusion process is exactly the one obtained in a large well-mixed population  (Li and Lessard \cite{LL2020}).

In the particular context of the deterministic Prisoner's dilemma with $S_1$ for cooperation and $S_2$ for defection, which means that $\sigma^2_{1,1}=\sigma^2_{1,2}=\sigma^2_{2,1}=\sigma^2_{2,2}=0$ with $\mu_{1,1}<\mu_{2,1}$ and $\mu_{1,2}<\mu_{2,2}$, our findings indicate that selection fully disfavours the evolution of cooperation in the sense that $F_1<(Nd)^{-1}<F_2$ for any $\nu>0$. Here $(Nd)^{-1}$ is the probability of ultimate fixation of a new mutant under neutrality.

We have investigated the effects of fluctuations in the payoffs for defection, represented by $S_2$, across various scenarios. It is reasonable to expect that there would be greater uncertainty in the payoffs associated with defection compared to cooperation. Several factors contribute to this difference, one being that defectors often find themselves more isolated than cooperators. Consequently, they may not benefit or suffer as much from fluctuations in the surrounding resources caused by environmental variations.
Furthermore, defectors may face punishments or a lack of rewards from other individuals in the population, adding to the uncertainty about their payoffs. While defection might offer higher expected payoffs compared to cooperation, this increased isolation and exposure to external factors can also result in higher variances in their payoffs. As a consequence, this elevated variance in the payoffs for defection can significantly influence the evolutionary dynamics and the overall prevalence of cooperative strategies in the population.

In the first scenario, we assume that all the population-scaled variances of the payoffs are negligible compared to $\sigma_{2,2}^2$, which represents the population-scaled variance of the payoff to  $S_2$ against $S_2$. As we increase $\sigma_{2,2}^2$, the fixation probability $F_1$ increases, while the fixation probability $F_2$ decreases for any deme-scaled extinction rate $\nu>0$. Additionally, we have identified three thresholds, namely   $\sigma^2_{*}(\nu)$, $\sigma^2_{**}(\nu)$ and $\sigma^2_{***}(\nu)$, such that  the conditions  $F_1 > (Nd)^{-1}$, $F_1>F_2$, and $F_2<(Nd)^{-1}$  hold when $\sigma_{2,2}^2>\sigma^2_{*}(\nu)$, $\sigma_{2,2}^2>\sigma^2_{**}(\nu)$ and $\sigma_{2,2}^2>\sigma^2_{***}(\nu)$, respectively.  Furthermore, when all the population-scaled means of the payoffs, represented by  $\mu_{1,1}, \mu_{1,2}, \mu_{2,1}, \mu_{2,2}$, are of the same small order, these thresholds satisfy the inequalities $\sigma^2_{*}(\nu)<\sigma^2_{**}(\nu)<\sigma^2_{***}(\nu)$.
In the case where the population-scaled means determine an additive prisoner's dilemma ($\mu_{1,1}=\mu_b-\mu_c$, $\mu_{1,2}=-\mu_c$, $\mu_{2,1}=\mu_b$, and $\mu_{2,2}=0$), we have shown that  $\sigma^2_{*}(\nu)$ and $\sigma^2_{**}(\nu)$ are decreasing with respect to $\nu>0$ while  $\sigma^2_{***}(\nu)$ is increasing. This means that increasing the deme-scaled extinction rate will lessen the conditions for selection to favour the evolution of $S_1$ and to favour the evolution of $S_1$ more than the evolution of $S_2$. Conversely, this will reinforce the condition for selection to disfavour the evolution of $S_2$ (see Figure \ref{figure1}). Under such a scenario, it is impossible for selection to disfavour the evolution of both cooperation and defection for any $\sigma_{2,2}^2>0$ and $\nu>0$.

In the second scenario, we consider that all the population-scaled variances of the payoffs other than $\sigma_{2,1}^2$, for the payoff to $S_2$ against $S_1$, are insignificant.  Similarly to the first scenario, we have shown that an increase in $\sigma_{2,1}^2$ will increase the fixation probability $F_1$ and decrease the fixation probability $F_2$. Moreover, this scenario will make it possible for selection to favour the evolution of $S_1$ in any sense. More precisely, there exist three thresholds $\sigma^2_{*}(\nu)$, $\sigma^2_{**}(\nu)$ and $\sigma^2_{***}(\nu)$, where
$F_1 > (Nd)^{-1}$, $F_1>F_2$ and $F_2<(Nd)^{-1}$ as long as $\sigma_{2,1}^2>\sigma^2_{*}(\nu)$, $\sigma_{2,1}^2>\sigma^2_{**}(\nu)$ and $\sigma_{2,1}^2>\sigma^2_{***}(\nu)$, respectively. Here, there is no clear relationship between these thresholds. However, when the population-scaled means of the payoffs are of the same small order and determine an additive Prisoner's dilemma, we find that $\sigma^2_{***}(\nu)<\sigma^2_{**}(\nu)<\sigma^2_{*}(\nu)$ for any  $\nu>0$. Furthermore, the functions
$\sigma^2_{*}(\nu)$ and $\sigma^2_{**}(\nu)$ are increasing, while $\sigma^2_{***}(\nu)$ is decreasing. This means that increasing the deme-scaled extinction rate will strengthen the conditions for $F_1>(Nd)^{-1}$ and $F_1>F_2$ and lessen the condition for $F_2<(Nd)^{-1}$ (see Figure \ref{figure2}). Under such a scenario, it is impossible for selection to favour the evolution of both cooperation and defection for any $\sigma_{2,2}^2>0$ and $\nu>0$.

In the third scenario, we assume that the population-scaled variances of the payoffs to $S_1$ are insignificant compared to the population-scaled variances of the payoffs to $S_2$ against $S_1$ and to $S_2$ against $S_2$, given by $\sigma^2$. It is important to highlight that increasing $\sigma^2$ will lead to an increase in $F_1$ and a decrease in $F_2$. In addition, it is crucial to note that all the conclusions from the first scenario remain valid (refer to Figure \ref{figure3}). By comparing Figures \ref{figure1}, \ref{figure2} and \ref{figure3}, it becomes evident that for increasing $F_1$, it is more effective to increase the population-scaled variances of the payoffs to both $S_2$ against $S_1$ and $S_2$ against $S_1$ simultaneously, rather than concentrating the increase on just one of these variances.

In general, assuming that all population-scaled means and variances of the payoffs are of the same small order, we have shown that decreasing the population-scaled variance of any payoff to $S_1$  or increasing the population-scaled variance of any payoff to $S_2$ increases the fixation probability $F_1$ and decreases the fixation probability $F_2$. Therefore, either increases the difference $F_1-F_2$. 
As the deme-scaled extinction rate $\nu$ decreases, which rises the level of identity by descent among offspring within demes, the weight of $\sigma_{2,1}^2-\sigma_{1,2}^2$ in this difference increases, while the weight of $\sigma_{2,2}^2-\sigma_{1,1}^2$ decreases. This may be counterintuitive, since then offspring of the same type are more likely to interact than offspring of different types. But we have to remind ourselves that then offspring are also more likely to be in competition with offspring of the same type in the same deme. Note also that $F_1-F_2$ does not depend on $\nu$ in the case of all deterministic payoffs.

 Our results extend previous studies that have explored the stochastic Prisoner's dilemma in a large but finite well-mixed population which corresponds to the case where $\nu \rightarrow \infty$ (see Li and Lessard \cite{LL2020} for a Wright-Fisher model and by Kroumi \textit{et al.} \cite{KMLL2021} for a Moran model). In these studies, it was shown that more variability in the payoffs for defection than in the payoffs for cooperation is required for selection to favour the evolution of cooperation in any sense. Note that
a similar conclusion was drawn in the context of an infinite well-mixed population  (Zheng \textit{et al.} \cite{ZLLT2018,ZLLT2017}). In this setting, the same condition was found to be necessary  for the fixation of cooperation  to be stochastically locally stable, and for the fixation of defection  to be stochastically locally unstable. In an island model, the weights of the variances of the different payoffs may differ, but the general conclusion still holds.

Finally, the island model with the possibility of extinction and recolonization can be extended to study groups distributional dynamics stemming from the possibility of opting out in repeated games  (Cressman and Krivan \cite{CK2022}, Krivan and Cressman \cite{KC2020}, Li and Lessard \cite{LL2021}, Zhang \emph{et al.} \cite{Zetal2016}). In the general case, however, the extinction probability or rate depends on the group composition, which makes the analysis more challenging.

%These consistent results across different population models and settings emphasize the crucial role that payoff variability plays in shaping the evolutionary dynamics of cooperation and defection. The benefit of larger variances in defection payoffs for the promotion of cooperation highlights the sensitivity of cooperative strategies to uncertainty in the environment.

%%%%%%%%%%%%%%%%%%%%%%%%%%%%%%%%%%%%%%%%%%%%%

%%%%%%%%%%%%%%%%%%%%%%%%%%%%%%%%%%%%%%%%%%%%%%%%%%%%%%%%%%%%%%%%%%%%%%%%%%%%%%%%%%%%%%%%%%%%%%%
%%%%%%%%%%%%%%%%%%%%%%%%%%%%%%%%%%%%%%%%%%%%%%%%%%%%%%%%%%%%%%%%%%%%%%%%%%%%%%%%%%%%%%%%%%%%%%%

\section*{Appendices  }

\subsection*{Appendix A: Proof of Lemma 1  }

%\begin{proof}
By using Taylor's theorem with the Lagrange form of the remainder, there exists a real random variable $\xi_i$ such that $|\xi_i|<|\tilde{a}_{i}|< 1$ and
\begin{equation}\label{sec2-eq8}
\frac{1}{1+\tilde{a}_{i}}=1-\tilde{a}_{i}+\tilde{a}_{i}^2-\frac{\tilde{a}_{i}^3}{(1+\xi_i)^{3}},
\end{equation}
from which we deduce 
\begin{align}\label{sec2-eq9}
E_{\mathbf{z}}\left[\tilde{x}\right]&=\sum_{i=0}^{N}z_ix_iE_{\mathbf{z}}\left[\frac{1+\tilde{a}_{1,i}}{1+\tilde{a}_{i}}\right]\nonumber\\
&=\sum_{i=0}^{N}z_ix_iE_{\mathbf{z}}\left[\left(1+\tilde{a}_{1,i}\right)\left(1-\tilde{a}_{i}+\tilde{a}_{i}^2\right)\right]+o(d^{-1})\nonumber\\
&=x+\sum_{i=0}^{N}z_ix_iE_{\mathbf{z}}\left[\tilde{a}_{1,i}-\tilde{a}_{i}\right]-\sum_{i=0}^{N}z_ix_iE_{\mathbf{z}}\left[\left(\tilde{a}_{1,i}-\tilde{a}_{i}\right)\tilde{a}_{i}\right]+o(d^{-1}).
\end{align}
On the other hand, we have 
\begin{align}\label{sec2-eq10}
&\sum_{i=0}^{N}z_ix_iE_{\mathbf{z}}\left[\tilde{a}_{1,i}-\tilde{a}_{i}\right]\nonumber\\
&=\sum_{i=0}^{N}z_ix_iE_{\mathbf{z}}\left[\tilde{a}_{1,i}\right]
-\sum_{i=0}^{N}z_ix_i\sum_{j=0}^{N}z_jE_{\mathbf{z}}\left[\tilde{a}_{j}\right]-
\left(
\sum_{i=0}^{N}z_ix_iE_{\mathbf{z}}\left[\tilde{a}_{i}\right]
-\sum_{i=0}^{N}z_ix_i\sum_{j=0}^{N}z_jE_{\mathbf{z}}\left[\tilde{a}_{j}\right]
\right)\nonumber\\
&=Cov_{\mathbf{z}}(q_{1,I_1},a_{I_1})-Cov_{\mathbf{z}}(q_{1,I_1},a_{I_2})
\end{align}
and
\begin{align}\label{sec2-eq11}
&\sum_{i=0}^{N}z_ix_iE_{\mathbf{z}}\left[\left(\tilde{a}_{1,i}-\tilde{a}_{i}\right)\tilde{a}_{i}\right]\nonumber\\
&=\sum_{i=0}^{N}z_ix_iE_{\mathbf{z}}\left[\tilde{a}_{1,i}\tilde{a}_{i}\right]-\sum_{i=0}^{N}z_ix_i\sum_{j=0}^{N}z_jE_{\mathbf{z}}\left[\tilde{a}_{j}^2\right]-
\left(
\sum_{i=0}^{N}z_ix_iE_{\mathbf{z}}\left[\tilde{a}_{i}^2\right]
-
\sum_{i=0}^{N}z_ix_i\sum_{j=0}^{N}z_jE_{\mathbf{z}}\left[\tilde{a}_{j}^2\right]
\right)\nonumber\\
&=Cov_{\mathbf{z}}(q_{1,I_1},a_{I_1}a_{I_2})-Cov_{\mathbf{z}}(q_{1,I_1},a_{I_2}a_{I_3}),
\end{align}
where $I_1$, $I_2$, and $I_3$ are randomly chosen offspring produced in the same deme chosen at random before viability selection and recolonization.
Inserting Eqs. (\ref{sec2-eq10}) and (\ref{sec2-eq11}) in  Eq. (\ref{sec2-eq9}) and then using the expression for the payoff of an offspring in Eq. (\ref{sec1-eq4}) and the assumptions in Eq. (\ref{sec1-eq2}) lead to the stated result.
%\end{proof}

\subsection*{Appendix B: Proof of Lemma 2  }
%\begin{proof}
By a similar argument as in the previous proof, we have
\begin{equation}\label{sec2-eq13}
\begin{split}
E_{\mathbf{z}}[\tilde{\tilde{x}}]&=\sum_{i=0}^{N}z_ix_iE_{\mathbf{z}}\left[\frac{1+\tilde{a}_{1,i}}{1+\tilde{a}}\right]\\
&=\sum_{i=0}^{N}z_ix_iE_{\mathbf{z}}\left[\left(1+\tilde{a}_{1,i}\right)\left(1-\tilde{a}+\tilde{a}^2\right)\right]+o(d^{-1})\\
&=x+\sum_{i=0}^{N}z_ix_iE_{\mathbf{z}}\left[\tilde{a}_{1,i}-\tilde{a}-(\tilde{a}_{1,i}-\tilde{a})\tilde{a}\right]+o(d^{-1}).
\end{split}
\end{equation}
On the other hand, we have the identities
\begin{subequations}\label{sec2-eq14}
\begin{align}
Cov_{\mathbf{z}}(q_{1,I_1},a_{I_1})&=\sum_{i=0}^{N}z_ix_iE_{\mathbf{z}}\left[\tilde{a}_{1,i}-\tilde{a}\right],\\
Cov_{\mathbf{z}}(q_{1,I_1},a_{I_1}a_{J_1})&=\sum_{i=0}^{N}z_ix_iE_{\mathbf{z}}\left[(\tilde{a}_{1,i}-\tilde{a})\tilde{a}\right],
\end{align}
\end{subequations}
where $I_1$ and $J_1$ are randomly chosen offspring produced in demes chosen at random and independently before viability selection and recolonization. The stated result  is obtained by substituting these identities into Eq. (\ref{sec2-eq13}) and then using Eqs. (\ref{sec1-eq4} ) and (\ref{sec1-eq2}).
%\end{proof}

\subsection*{Appendix C: Proof of Lemma 3  }
%\begin{proof}
Proceeding as in the proof of Lemma 1, we have
\begin{align}\label{sec2-eq20}
E_{\mathbf{z}}\left[\tilde{x}^2\right]&=\sum_{i,j=0}^{N}z_iz_jx_ix_jE_{\mathbf{z}}\left[\frac{1+\tilde{a}_{1,i}}{1+\tilde{a}_{i}}\frac{1+\tilde{a}_{1,j}}{1+\tilde{a}_{j}}\right]\nonumber\\
&=\sum_{i,j=0}^{N}z_iz_jx_ix_jE_{\mathbf{z}}\Big[\left(1+\tilde{a}_{1,i}\right)\left(1+\tilde{a}_{1,j}\right)\left(1-\tilde{a}_{i}+\tilde{a}_{i}^2\right)\left(1-\tilde{a}_{j}+\tilde{a}_{j}^2\right)\Big]+o(d^{-1})\nonumber\\
&=x^2+\sum_{i,j=0}^{N}z_iz_jx_ix_j
E_{\mathbf{z}}\Big[\tilde{a}_{1,i}-\tilde{a}_{i}+\tilde{a}_{1,j}-\tilde{a}_{j}-\left(\tilde{a}_{1,i}-\tilde{a}_{i}\right)\tilde{a}_{i}-\left(\tilde{a}_{1,j}-\tilde{a}_{j}\right)\tilde{a}_{j}\nonumber\\
&\quad\quad\quad\quad\quad\quad\quad\quad\quad\quad+\left(\tilde{a}_{1,i}-\tilde{a}_{i}\right)\tilde{a}_{1,j}-\left(\tilde{a}_{1,i}-\tilde{a}_{i}\right)\tilde{a}_{j}\Big]
+o(d^{-1}).
\end{align}
Moreover, we have the identities 
\begin{subequations}\label{sec2-eq21}
\begin{align}
\sum_{i,j=0}^{N}z_iz_jx_ix_jE_{\mathbf{z}}\left[\tilde{a}_{1,i}-\tilde{a}_{i}\right]
&=x\Big[Cov_{\mathbf{z}}\left(q_{1,I_1},a_{I_1}\right)-Cov_{\mathbf{z}}\left(q_{1,I_1},a_{I_2}\right)\Big],\\
%%%%
%%%%%
\sum_{i,j=0}^{N}z_iz_jx_ix_jE_{\mathbf{z}}\left[(\tilde{a}_{1,i}-\tilde{a}_{i})\tilde{a}_{i}\right]
&=x\Big[Cov_{\mathbf{z}}\left(q_{1,I_1},a_{I_1}a_{I_2}\right)-Cov_{\mathbf{z}}\left(q_{1,I_1},a_{I_2}a_{I_3}\right)\Big],\\
%%%%
%%%%%
\sum_{i,j=0}^{N}z_iz_jx_ix_jE_{\mathbf{z}}\left[(\tilde{a}_{1,i}-\tilde{a}_{i})\tilde{a}_{1,j}\right]
&=Cov_{\mathbf{z}}\left(q_{1,I_1},q_{1,J_1}a_{I_1}a_{J_1}\right)-Cov_{\mathbf{z}}\left(q_{1,I_1},q_{1,J_1}a_{I_2}a_{J_1}\right),\\
%%%%
%%%%%
\sum_{i,j=0}^{N}z_iz_jx_ix_jE_{\mathbf{z}}\left[(\tilde{a}_{1,i}-\tilde{a}_{i})\tilde{a}_{j}\right]
&=Cov_{\mathbf{z}}\left(q_{1,I_1},q_{1,J_1}a_{I_1}a_{J_2}\right)-Cov_{\mathbf{z}}\left(q_{1,I_1},q_{1,J_1}a_{I_2}a_{J_2}\right),
\end{align}
\end{subequations}
where $I_1$, $I_2$ are two randomly chosen offspring produced in the same deme chosen at random, and similarly for $J_1$, $J_2$ in an independent way.
Inserting these identities into Eq. (\ref{sec2-eq20}) yields
\begin{align}\label{sec2-eq22}
E_{\mathbf{z}}\left[\tilde{x}^2\right]
%%%
%%%
&=x^2+2x\Big[Cov_{\mathbf{z}}\left(q_{1,I_1},a_{I_1}\right)-Cov_{\mathbf{z}}\left(q_{1,I_1},a_{I_2}\right)\Big]
-2x\Big[Cov_{\mathbf{z}}\left(q_{1,I_1},a_{I_1}a_{I_2}\right)\nonumber\\
&\quad\quad\quad-Cov_{\mathbf{z}}\left(q_{1,I_1},a_{I_2}a_{I_3}\right)\Big]
+Cov_{\mathbf{z}}\left(q_{1,I_1},q_{1,J_1}a_{I_1}a_{J_1}\right)+
Cov_{\mathbf{z}}\left(q_{1,I_1},q_{1,J_1}a_{I_2}a_{J_2}\right)\nonumber\\
&\quad\quad\quad-Cov_{\mathbf{z}}\left(q_{1,I_1},q_{1,J_1}a_{I_2}a_{J_1}\right)
-Cov_{\mathbf{z}}\left(q_{1,I_1},q_{1,J_1}a_{I_1}a_{J_2}\right)+o\left(d^{-1}\right)\nonumber\\
%%%%
%%%%
&=x^2+\frac{1}{Nd}\Big[
2x\Big(Cov_{\mathbf{z}}\left(q_{1,I_1},\mu_{I_1}\right)+Cov_{\mathbf{z}}\left(q_{1,I_1},\sigma_{I_2,I_3}\right)
-Cov_{\mathbf{z}}\left(q_{1,I_1},\mu_{I_2}\right)\nonumber\\
&\quad\quad\quad-Cov_{\mathbf{z}}\left(q_{1,I_1},\sigma_{I_1,I_2}\right)
\Big)+Cov_{\mathbf{z}}\left(q_{1,I_1},q_{1,J_1}\sigma_{I_1,J_1}\right)+Cov_{\mathbf{z}}\left(q_{1,I_1},q_{1,J_1}\sigma_{I_2,J_2}\right)
\nonumber\\
&\quad\quad\quad-Cov_{\mathbf{z}}\left(q_{1,I_1},q_{1,J_1}\sigma_{I_2,J_1}\right)-Cov_{\mathbf{z}}\left(q_{1,I_1},q_{1,J_1}\sigma_{I_1,J_2}\right)\Big]+o\left(d^{-1}\right).
\end{align}
Combining the last equation and Lemma \ref{lemma1} completes the proof of Lemma 3.
%\end{proof}

\subsection*{Appendix D: Proof of Lemma 4  }

%\begin{proof}
Proceeding as in the proof of Lemma 2, we obtain 
\begin{align}\label{sec2-eq24}
E_{\mathbf{z}}\left[\tilde{\tilde{x}}^2\right]
&=\sum_{i,j=0}^{N}z_iz_jx_ix_jE\Big[\left(1+\tilde{a}_{1,i}\right)\left(1+\tilde{a}_{1,j}\right)\left(1-\tilde{a}+\tilde{a}^2\right)^2\Big]+o(d^{-1})\nonumber\\
%%%%
%%%%
&=x^2+\sum_{i,j=0}^{N}z_iz_jx_ix_j
E_{\mathbf{z}}\Big[\tilde{a}_{1,i}-\tilde{a}+\tilde{a}_{1,j}-\tilde{a}-2\left(\tilde{a}_{1,i}-\tilde{a}\right)\tilde{a}\nonumber\\
&\quad\quad\quad\quad\quad\quad\quad\quad\quad\quad-2\left(\tilde{a}_{1,j}-\tilde{a}\right)\tilde{a}+\tilde{a}_{1,i}\tilde{a}_{1,j}-\tilde{a}\tilde{a}\Big]+o(d^{-1})\nonumber\\%%%%
%%%%
&=x^2+\sum_{i,j=0}^{N}z_iz_jx_ix_jE_{\mathbf{z}}\Big[2(\tilde{a}_{1,i}-\tilde{a})-
4(\tilde{a}_{1,j}-\tilde{a})\tilde{a}
+\tilde{a}_{1,i}\tilde{a}_{1,j}-\tilde{a}^2
\Big]+o(d^{-1})\nonumber\\
%%%%
%%%%
&=x^2+ 2x\,Cov_{\mathbf{z}}\left(q_{1,I_1},a_{I_1}\right)+Cov_{\mathbf{z}}\left(q_{1,I_1}q_{1,J_1},a_{I_1}a_{J_1}\right)-4x\,Cov_{\mathbf{z}}\left(q_{1,I_1},a_{I_1}a_{J_1}\right)+o(d^{-1})\nonumber\\
&=x^2+ \frac{1}{Nd}\Big[
2x\,Cov_{\mathbf{z}}\left(q_{1,I_1},\mu_{I_1}\right)+Cov_{\mathbf{z}}\left(q_{1,I_1}q_{1,J_1},\sigma_{I_1,J_1}\right)\nonumber\\
&\quad\quad\quad\quad\quad-4x\,Cov_{\mathbf{z}}\left(q_{1,I_1},\sigma_{I_1,J_1}\right)\Big]+o\left(d^{-1}\right),
\end{align}
where $I_1$ and $J_1$ are randomly chosen offspring produced in demes independently chosen at random.
The proof of Lemma 4 can be completed by using Lemma \ref{lemma2} in combination with the above equation.
%\end{proof}

\subsection*{Appendix E: Proof of Lemma 5  }

%\begin{proof}
Proceeding as in the proofs of Lemmas 3 and 4, we get
\begin{align}\label{sec2-eq26}
E_{\mathbf{z}}\left[\tilde{x}\tilde{\tilde{x}}\right]&=\sum_{i,j=0}^{N}z_iz_jx_ix_jE_{\mathbf{z}}\Big[\left(1+\tilde{a}_{1,i}\right)\left(1+\tilde{a}_{1,j}\right)\left(1-\tilde{a}_{i}+\tilde{a}_{i}^2\right)\left(1-\tilde{a}+\tilde{a}^2\right)\Big]+o(d^{-1})\nonumber\\
&=x^2+\sum_{i,j=0}^{N}z_iz_jx_ix_j
E_{\mathbf{z}}\Big[\tilde{a}_{1,i}-\tilde{a}_{i}+\tilde{a}_{1,j}-\tilde{a}-\left(\tilde{a}_{1,i}-\tilde{a}_{i}\right)\tilde{a}_{i}-\left(\tilde{a}_{1,j}-\tilde{a}\right)\tilde{a}\nonumber\\
&\quad\quad\quad\quad\quad\quad\quad\quad\quad\quad+\left(\tilde{a}_{1,i}-\tilde{a}_{i}\right)\tilde{a}_{1,j}-\left(\tilde{a}_{1,i}-\tilde{a}_{i}\right)\tilde{a}\Big]
+o(d^{-1}).\nonumber\\
%%%%
%%%%
&=x^2+
x\,Cov_{\mathbf{z}}(q_{1,I_1},a_{I_1})-x\,Cov_{\mathbf{z}}(q_{1,I_1},a_{I_2})
+x\,Cov_{\mathbf{z}}(q_{1,I_1},a_{I_2}a_{I_3})\nonumber\\
&\quad-x\,Cov_{\mathbf{z}}(q_{1,I_1},a_{I_1}a_{I_2})
+Cov_{\mathbf{z}}(q_{1,I_1},q_{1,J_1}a_{I_1}a_{J_1})
+x\,Cov_{\mathbf{z}}(q_{1,I_1},a_{I_1})\nonumber\\
&
\quad+x\,Cov_{\mathbf{z}}(q_{1,I_1},a_{I_2}a_{J_1})
-x\,Cov_{\mathbf{z}}(q_{1,I_1},a_{I_1}a_{J_1})
-x\,Cov_{\mathbf{z}}(q_{1,I_1},a_{I_1}a_{J_1})\nonumber\\
&\quad-Cov_{\mathbf{z}}(q_{1,I_1},q_{1,J_1}a_{I_2}a_{J_1})
+o(d^{-1})\nonumber\\
%%%
%%%
&=x^2+\frac{1}{Nd}\Big[
2x\,Cov_{\mathbf{z}}(q_{I_1},\mu_{I_1})
-2x\,Cov_{\mathbf{z}}(q_{1,I_1},\sigma_{I_1,J_1})
+x\,Cov_{\mathbf{z}}(q_{1,I_1},\sigma_{I_2,I_3})\nonumber\\
&\quad\quad\quad\quad\quad-x\,Cov_{\mathbf{z}}(q_{1,I_1},\sigma_{I_1,I_2})
-Cov_{\mathbf{z}}(q_{1,I_1},\mu_{I_2})
+x\,Cov_{\mathbf{z}}(q_{1,I_1},\sigma_{I_2,J_1})\nonumber\\
&\quad\quad\quad\quad\quad+Cov_{\mathbf{z}}(q_{1,I_1},q_{1,J_1}\sigma_{I_1,J_1})
-Cov_{\mathbf{z}}(q_{1,I_1},q_{1,J_1}\sigma_{I_2,J_1})
\Big]+o(d^{-1}).
\end{align}
Here, $I_1$ and $I_2$ are randomly chosen offspring produced in the same deme chosen at random, while $J_2$ is an independently chosen offspring produced in the whole population. 
The proof of Lemma 5 is completed by combining Lemmas \ref{lemma1} and \ref{lemma2} with the above equation.
%\end{proof}

\subsection*{Appendix F: Proof of Proposition 3 }

%\begin{proof}
Let each deme at the beginning of generation $0$ be labeled by $(i, d_i)$ where $i$ is the type of the deme and $d_i$ its rank when the demes of the same type are placed in an arbitrary order, for $i=0, 1, \ldots, N$ and $d_i=1, \ldots, dz_i$. Thus, the frequency of $S_1$ at the beginning of generation $1$ can be expressed as
\begin{equation}\label{sec2-eq31}
X(1)=\frac{1}{d}\sum_{(i, d_i)}^{}X_{i,d_i}(1),
\end{equation}
where $X_{i,d_i}(1)$ represents the frequency of type $S_1$ in deme $(i,d_i)$ at the beginning of generation $1$. Note that this random variable is independent of  $X_{j,d_j}(1)$ for all $(j, d_j) \ne (i, d_i)$ and satisfies
\begin{equation}\label{sec2-eq32}
E_{\mathbf{z}}\left[X_{i,d_i}(1)\Big|(a_{k,l})_{}\right]=\tilde{\tilde{x}}_{i}.
\end{equation}
Defining
\begin{equation}\label{sec2-eq33}
U_{i,d_i}=X_{i,d_i}(1)-\tilde{\tilde{x}}_{i},
\end{equation}
 we have $\left|U_{i,d_i}\right|\leq1$ and 
$E_{\mathbf{z}}\left[U_{i,d_i}\right]=0$.
Moreover, for the third central moment of the  frequency of type $S_1$ at the beginning of generation $1$, we have
\begin{align}\label{sec2-eq34}
E_{\mathbf{z}}\left[\Big(X(1)-E_{\mathbf{z}}\left[X(1)\right]\Big)^3\right]&=E_{\mathbf{z}}\Bigg[\Bigg(\frac{1}{d}\sum_{(i,d_i)}U_{i,d_i}\Bigg)^3\Bigg]\nonumber\\
&=\frac{1}{d^3}\Bigg(\sum_{(i,d_i)}E_{\mathbf{z}}\left[U_{i,d_i}^3\right]
+3\sum_{(i,d_i), (j,d_j)}E_{\mathbf{z}}\left[U_{i,d_i}^2\right]E_{\mathbf{z}}\left[U_{j,d_{j}}\right]\nonumber\\
&\quad\quad\quad+\sum_{\textrm{different }(i,d_i), (j,d_j), (k,d_k)}E_{\mathbf{z}}\left[U_{i,d_i}\right]E_{\mathbf{z}}\left[U_{j,d_{j}}\right]E_{\mathbf{z}}\left[U_{k,d_{k}}\right]
\Bigg)\nonumber\\
&=\frac{1}{d^3}\sum_{(i,d_i)}E_{\mathbf{z}}\left[U_{i,d_i}^3\right],%\frac{1}{d^3}\sum_{(i,d_i)}1=
%\leq \frac{1}{d^2}=o(d^{-1}).
\end{align}
which is bounded in absolute value by $1/d^2$.
Similarly, for the fourth central moment of the frequency of $S_1$ at the beginning of generation $1$, we have
\begin{align}\label{sec2-eq35}
&E_{\mathbf{z}}\left[\left(X(1)-E_{\mathbf{z}}\left[X(1)\right]\right)^4\right]=E_{\mathbf{z}}\Bigg[\Bigg(\frac{1}{d}\sum_{(i,d_i)}U_{i,d_i}\Bigg)^4\Bigg]\nonumber\\
&=\frac{1}{d^4}\Bigg[\sum_{(i,d_i)}E_{\mathbf{z}}\left[U_{i,d_i}^4\right]
+
4\sum_{(i,d_i)\neq (j,d_j)}E_{\mathbf{z}}\left[U_{i,d_i}^3\right]E_{\mathbf{z}}\left[U_{j,d_{j}}\right]
+6\sum_{(i,d_i)\neq (j,d_j)}E_{\mathbf{z}}\left[U_{i,d_i}^2\right]E_{\mathbf{z}}\left[U_{j,d_{j}}^2\right]\nonumber\\
&\quad\quad+12\sum_{\textrm{different }(i,d_i), (j,d_j), (k,d_k)}E_{\mathbf{z}}\left[U_{i,d_i}^2\right]E_{\mathbf{z}}\left[U_{j,d_{j}}\right]E_{\mathbf{z}}\left[U_{k,d_{k}}
\right]\nonumber\\
&\quad\quad+\sum_{\textrm{different }(i,d_i), (j,d_j), (k,d_k), (l,d_l)}E_{\mathbf{z}}\left[U_{i,d_i}\right]E_{\mathbf{z}}\left[U_{j,d_{j}}\right]E_{\mathbf{z}}\left[U_{k,d_{k}}\right]E_{\mathbf{z}}\left[U_{l,d_{l}}\right]
\Bigg]\nonumber\\
&=\frac{1}{d^4}\Bigg[\sum_{(i,d_i)}E_{\mathbf{z}}\left[U_{i,d_i}^4\right]+
6\sum_{(i,d_i)\neq (j,d_j)}E_{\mathbf{z}}\left[U_{i,d_i}^2\right]E_{\mathbf{z}}\left[U_{j,d_{j}}^2\right]\Bigg]
\nonumber\\
&\leq\frac{1}{d^4}\Big[d+
6d(d-1)\Bigg].%=o(d^{-1}).
\end{align}
On the other hand, we have
\begin{align}\label{sec2-eq36}
E_{\mathbf{z}}\left[\left(X(1) - X(0)\right)^4\right]&=E_{\mathbf{z}}\left[
\left(X(1)-E_{\mathbf{z}}[X(1)]+E_{\mathbf{z}}[X(1)-X(0)]\right)^4
\right]\nonumber\\
&=E_{\mathbf{z}}\left[
\left(X(1)-E_{\mathbf{z}}[X(1)]\right)^4
\right]
+
4E_{\mathbf{z}}\left[X(1)-X(0)\right]E_{\mathbf{z}}\left[
\left(X(1)-E_{\mathbf{z}}[X(1)]\right)^3
\right]\nonumber\\
&\quad+6\left(E_{\mathbf{z}}\left[X(1)-X(0)\right]\right)^2E_{\mathbf{z}}\left[
\left(X(1)-E_{\mathbf{z}}[X(1)]\right)^2
\right] \nonumber\\
&\quad+ \left(E_{\mathbf{z}}\left[X(1)-X(0)\right]\right)^4.
\end{align}
This is a function $o(d^{-1})$ owing to Eqs. (\ref{sec2-eq5}), (\ref{sec2-eq34}), (\ref{sec2-eq35}) and the fact that $X(1)$ is bounded in absolute value by $1$.
%\end{proof}

\subsection*{Appendix G: Proof of Proposition 5 }
%\begin{proof}
Note that we have
\begin{equation}\label{sec2-eq49}
\begin{split}
Var_{\mathbf{z}}\left[Y_i(1)-Y_i(0)\right]=Var_{\mathbf{z}}\left[Z_i(1)-v_i(X(1))\right]\leq 
2Var_{\mathbf{z}}\left[Z_i(1)\right]+2Var_{\mathbf{z}}\left[v_i(X(1))\right].
\end{split}
\end{equation}
Therefore, in order to establish (\ref{sec2-eq48}), it suffices to show  $Var_{\mathbf{z}}\left[Z_i(1)\right]=o(1)$, which is a direct consequence of Eq. (\ref{sec2-eq3}), and that $Var_{\mathbf{z}}\left[v_i(X(1))\right]=o(1)$. For this latter variance, we have
\begin{align}\label{sec2-eq50}
Var_{\mathbf{z}}\left[v_i(X(1))\right]&=Var_{\mathbf{z}}\left[v_i(x)+\sum_{j=1}^{N}r_j\left(X(1)-X(0)\right)^j\right]\nonumber\\
%&=Var_{\mathbf{z}}\left[\sum_{j=1}^{N}r_j\left(X(1)-X(0)\right)^j\right]\nonumber\\
&=\sum_{j_1,j_2=1}^{N}r_{j_1}r_{j_2}Cov_{\mathbf{z}}\left[\left(X(1)-X(0)\right)^{j_1},\left(X(1)-X(0)\right)^{j_2}\right]\nonumber\\
&=\sum_{j_1,j_2=1}^{N}r_{j_1}r_{j_2}\Bigg(
E_{\mathbf{z}}\left[\left(X(1)-X(0)\right)^{j_1+j_2}\right]\nonumber\\
&\quad\quad\quad\quad\quad\quad\quad-E_{\mathbf{z}}\left[\left(X(1)-X(0)\right)^{j_1}\right]E_{\mathbf{z}}\left[\left(X(1)-X(0)\right)^{j_2}\right]
\Bigg).
\end{align}
This is a function $o(1)$ owing to Propositions 1 and 2 besides the fact that $X(1)-X(0)$ is bounded in absolute value by $1$.
%\end{proof}

\subsection*{Appendix H: Exact identity-by-descent measures}

In this Appendix, we deduce the exact identity-by-descent measures defined in Section 4 for up to five offspring $I_1, \ldots, I_5$ chosen at random in the same deme at the beginning of a given generation in a neutral population of an infinite number of demes of fixed size $N$ that is in the stationary state.

\subsubsection*{H1: identity-by-descent measures for two offspring}

%Let $I_1$ and $I_2$ be two offspring chosen at random in the same deme at the beginning of a given generation in a neutral population of an infinite number of demes of size $N$ that is in the stationary state.
%randomly selected offspring from the neutral population residing in the same deme after reproduction before potential extinction and subsequent recolonization events. 
Considering the parents of two offspring
$I_1$ and $I_2$ one generation back, three potential scenarios are possible:
\begin{itemize}
\item The offspring have different parents that come from the same deme. This scenario occurs with probability $\left(1-\frac{1}{N}\right)(1-m)$.
\item The offspring have different parents that come from different demes. This situation has probability $\left(1-\frac{1}{N}\right)m$.
\item The offspring have the same parent, a situation that occurs with probability $\frac{1}{N}$.
\end{itemize}

Then, the probabilities $f_{2}$ and $f_{11}$ that $I_1$ and $I_2$ are identical by descent and not identical by descent, respectively, must satisfy
\begin{subequations}
\begin{align}
&f_{2}=\left(1-\frac{1}{N}\right)(1-m)\times f_{2}+\left(1-\frac{1}{N}\right)m\times0+\frac{1}{N}\times1,\\
&f_{11}=\left(1-\frac{1}{N}\right)(1-m)\times f_{11}+\left(1-\frac{1}{N}\right)m\times1+\frac{1}{N}\times0,
\end{align}
\end{subequations}
from which 
\begin{subequations}
\begin{align}
&f_{2}=\frac{1}{mN+1-m},\\
&f_{11}=\frac{m(N-1)}{mN+1-m}.
\end{align}
\end{subequations}

%%%%%%%%%%%%%%%%%%%%%%%%%%%%%%%%%%%%%%%%%%%%%%%%%%%%%%%%%%%%%%%%%%%%%%%%%%%%%%%%%%%%%%%%%%%%%%%%%%%%%%%%%%%%%%%%%%%%%%%%%%%%%%%%%%%%%%%%%%%%%%%%%%%%%%%%%%%%%%%%%%%%%%%%%%%%%%%%%%%%%%%%%%%%%%%%%%%%

\subsubsection*{H2: identity-by-descent measures for three offspring}
%Now, consider three offspring $I_1$, $I_2$ and $I_3$ randomly chosen from the same deme.

%%%%%%%%%%%%%%%%%%%%%%%%%%%%%%%%%%%%%%%%%%%%%%%%%%%%%%%%%%%%%%%%%%%%%%%%%%%%%%%%%%%%%%%%%%%%%%%%%%%%%%%%%%%%%%%%%%%%%%%%%%%%%%%%%%%%%%%%%%%%%%%%%%%%%%%%%%%%%%%%%%%%%%%%%%%%%%%%%%%%%%%%%%%%%%%%%%%%
\paragraph{Calculation of $f_{3}$:}
For three offspring to be identical by descent, three situations are possible for their parents one generation back: 
%Assuming that $I_1\equiv I_2\equiv I_3$ and considering their ancestors one generation back in time, three potential scenarios arise:
\begin{itemize}
\item The offspring have three different parents that come from the same deme. This scenario has probability $\left(1-\frac{1}{N}\right)\left(1-\frac{2}{N}\right)(1-m)$. 
\item The offspring have exactly two different parents that come from the same deme, a situation whose probability is $\frac{\binom{3}{2}}{N}\left(1-\frac{1}{N}\right)(1-m)$. 
\item The offspring have the same parent, which occurs with probability  $\frac{1}{N^2}$.
\end{itemize}
Note that all other scenarios are not compatible with $I_1$, $I_2$ and $I_3$ being identical by descent. By conditioning on the parents one generation back, we find 
\begin{equation}
f_{3}=\left(1-\frac{1}{N}\right)\left(1-\frac{2}{N}\right)(1-m)\times f_{3}
+\frac{3}{N}\left(1-\frac{1}{N}\right)(1-m)\times f_{2}
+\frac{1}{N^2}\times 1,
\end{equation}
from which we deduce
\begin{align}
f_{3}&=\frac{1+3(N-1)(1-m)f_{2}}{1+(N-1)(3-2m+mN)}\nonumber\\
&=\frac{1+(N-1)(3-2m)}{(mN+1-m)(mN^2+(3N-2)(1-m))}.
\end{align}

%%%%%%%%%%%%%%%%%%%%%%%%%%%%%%%%%%%%%%%%%%%%%%%%%%%%%%%%%%%%%%%%%%%%%%%%%%%%%%%%%%%%%%%%%%%%%%%%%%%%%%%%%%%%%%%%%%%%%%%%%%%%%%%%%%%%%%%%%%%%%%%%%%%%%%%%%%%%%%%%%%%%%%%%%%%%%%%%%%%%%%%%%%%%%%%%%%%%
\paragraph{Calculation of $f_{21}$:}
For an offspring $I_1$ to be identical by descent to an offspring $I_2$ but not to an offspring $I_3$, the scenarios to consider for the parents of the offspring one generation back and their probabilities are the following: 
\begin{itemize}
\item The offspring have different parents that come from the same deme, which has probability $\left(1-\frac{1}{N}\right)\left(1-\frac{2}{N}\right)(1-m)$. 
\item The offspring $I_1$ and $I_2$ have the same parent different from the parent of $I_3$ but coming from the same deme,  whose  probability is $\frac{1}{N}\left(1-\frac{1}{N}\right)(1-m)$. 
\item The offspring $I_1$ and $I_2$ have the same parent different from the parent of $I_3$ and coming from a different deme, which occurs with probability $\frac{1}{N}\left(1-\frac{1}{N}\right)m$.
\end{itemize}
By conditioning on these possible scenarios one generation back in time, we get the following relationship
\begin{equation}
f_{21}=\left(1-\frac{1}{N}\right)\left(1-\frac{2}{N}\right)(1-m)\times f_{21}
+\frac{1}{N}\left(1-\frac{1}{N}\right)\Big[(1-m)\times f_{11}
+m\times1\Big],
\end{equation}
from which we obtain
\begin{align}
f_{21}&=\frac{(N-1)[(1-m)\phi_{I_1|I_3}+m]}{mN^2+(3N-2)(1-m)}\nonumber\\
&=\frac{mN(N-1)}{(mN+1-m)(mN^2+(3N-2)(1-m))}.
\end{align}

%%%%%%%%%%%%%%%%%%%%%%%%%%%%%%%%%%%%%%%%%%%%%%%%%%%%%%%%%%%%%%%%%%%%%%%%%%%%%%%%%%%%%%%%%%%%%%%%%%%%%%%%%%%%%%%%%%%%%%%%%%%%%%%%%%%%%%%%%%%%%%%%%%%%%%%%%%%%%%%%%%%%%%%%%%%%%%%%%%%%%%%%%%%%%%%%%%%%
\paragraph{Calculation of $f_{111}$:}
For no offspring to be identical by descent to another among three, two scenarios for the parents one generation back are possible:
\begin{itemize}
\item The offspring have different that come from the same deme, with probability $\left(1-\frac{1}{N}\right)\left(1-\frac{2}{N}\right)(1-m)$. 
\item The offspring have different that come from different deme, with probability $\left(1-\frac{1}{N}\right)\left(1-\frac{2}{N}\right)m$.
\end{itemize}
By conditioning on these scenarios, we get
\begin{equation}
f_{111}=\left(1-\frac{1}{N}\right)\left(1-\frac{2}{N}\right)(1-m)\times f_{111}
+\left(1-\frac{1}{N}\right)\left(1-\frac{2}{N}\right)m\times1,
\end{equation}
from which we obtain
\begin{equation}
f_{111}=\frac{(N-1)(N-2)m}{mN^2+(3N-2)(1-m)}.
\end{equation}

%%%%%%%%%%%%%%%%%%%%%%%%%%%%%%%%%%%%%%%%%%%%%%%%%%%%%%%%%%%%%%%%%%%%%%%%%%%%%%%%%%%%%%%%%%%%%%%%%%%%%%%%%%%%%%%%%%%%%%%%%%%%%%%%%%%%%%%%%%%%%%%%%%%%%%%%%%%%%%%%%%%%%%%%%%%%%%%%%%%%%%%%%%%%%%%%%%%%
\subsubsection*{H3: identity-by-descent measures for four offspring}

Similarly, conditioning on  the number of parents of four offspring and the demes where they come from, 
%In this paragraph, consider four offspring randomly chosen from the same deme after reproduction before possible extinction and  recolonization. Considering their ancestral lineage one generation back in time, seven potential scenarios come into play:
%\begin{itemize}
%\item They are offspring of separate parents and yet reside within the same deme. This scenario occurs in the absence of coalescence and non-extinction.
%\item They come from separate parents residing in four different demes. This scenario occurs in the absence of coalescence and the presence of extinction.
%\item They have three distinct parents yet share the same deme. This situation arises with precisely one coalescent event and in the absence of extinction.
%\item They come from three distinct parents residing in three different demes. This scenario occurs with exactly one coalescent event and the presence of extinction.
%\item They possess two distinct parents, yet they inhabit the same deme. This situation occurs when exactly two coalescent events take place and extinction is absent. 
%\item They have two different parents residing in two different demes. This scenario arises when precisely two coalescent events take place, coupled with the occurrence of extinction. 
%\item They share the same parent, a scenario resulting from multiple coalescence events between them. 
%\end{itemize}
we obtain the system of equations
\begin{subequations}
\begin{align}
f_{4}=&\left(1-\frac{1}{N}\right)\left(1-\frac{2}{N}\right)\left(1-\frac{3}{N}\right)(1-m)f_{4}
+\frac{6}{N}\left(1-\frac{1}{N}\right)\left(1-\frac{2}{N}\right)(1-m)f_{3}
\nonumber\\
&+\frac{7}{N^2}\left(1-\frac{1}{N}\right)(1-m)f_{2}
+\frac{1}{N^3},\\
%%%%
%%%%
f_{31}=&\left(1-\frac{1}{N}\right)\left(1-\frac{2}{N}\right)\left(1-\frac{3}{N}\right)(1-m)f_{31}
+\frac{3}{N}\left(1-\frac{1}{N}\right)\left(1-\frac{2}{N}\right)(1-m)f_{21}
\nonumber\\
&+\frac{1}{N^2}\left(1-\frac{1}{N}\right)\Big[(1-m)f_{11}
+m\Big],\\
%%%%
%%%%
f_{22}=&\left(1-\frac{1}{N}\right)\left(1-\frac{2}{N}\right)\left(1-\frac{3}{N}\right)(1-m)f_{22}
+\frac{2}{N}\left(1-\frac{1}{N}\right)\left(1-\frac{2}{N}\right)(1-m)f_{21}
\nonumber\\
&+\frac{1}{N^2}\left(1-\frac{1}{N}\right)\Big[(1-m)f_{11}+m\Big],\\
%%%%
%%%%
f_{211}=&\left(1-\frac{1}{N}\right)\left(1-\frac{2}{N}\right)\left(1-\frac{3}{N}\right)(1-m)f_{211}\nonumber\\
&+\frac{1}{N}\left(1-\frac{1}{N}\right)\left(1-\frac{2}{N}\right)\Big[(1-m)f_{111}+m\Big],\\
%%%%
%%%%
f_{1111}=&\left(1-\frac{1}{N}\right)\left(1-\frac{2}{N}\right)\left(1-\frac{3}{N}\right)\Big[(1-m)f_{1111}
+m\Big].
\end{align}
\end{subequations}
The solution of this system is given by
\begin{subequations}
\begin{align}
&f_{4}=\frac{1+7(N-1)(1-m)f_{2}+6(N-1)(N-2)(1-m)f_{3}}{mN^3+(6N^2-11N+6)(1-m)},\\
&f_{31}=\frac{(N-1)m+(N-1)(1-m)f_{11}+3(N-1)(N-2)(1-m)f_{21}}{mN^3+(6N^2-11N+6)(1-m)},\\
&f_{22}=\frac{(N-1)m+(N-1)(1-m)f_{11}+2(N-1)(N-2)(1-m)f_{21}}{mN^3+(6N^2-11N+6)(1-m)},\\
&f_{2111}=\frac{(N-1)(N-2)m+(N-1)(N-2)(1-m)f_{111}}{mN^3+(6N^2-11N+6)(1-m)},\\
&f_{1111}=\frac{(N-1)(N-2)(N-3)m}{mN^3+(6N^2-11N+6)(1-m)}. 
\end{align}
\end{subequations}

%%%%%%%%%%%%%%%%%%%%%%%%%%%%%%%%%%%%%%%%%%%%%%%%%%%%%%%%%%%%%%%%%%%%%%%%%%%%%%%%%%%%%%%%%%%%%%%%%%%%%%%%%%%%%%%%%%%%%%
\subsubsection*{H4: identity-by-descent measures for five offspring}
%Similarly to the previous cases when considering the ancestral lineages one generation back in time of five offspring picked at random in the same deme after reproduction but before possible extinction and recolonization, we obtain the following identities
Proceeding as previously, the identity-by-descent measures for five offspring are found to be
\begin{subequations}
\begin{align}
&f_{5}=\frac{1+(N-1)(1-m)[15f_{2}+25(N-2)f_{3}+10(N-2)(N-3)f_{4}]}{mN^4+(1-m) (10 N^3-35N^2+50N-24)},\\
%%%%
%%%%
&f_{41}=\frac{
m(N-1)+(N-1)(1-m)[f_{11}+7(N-2)f_{21}+6(N-2)(N-3)f_{31}]
}{mN^4+(1-m) (10 N^3-35N^2+50N-24)},\\
%%%%
%%%%
&f_{32}=\frac{
m(N-1)+(N-1)(1-m)[f_{11}+4(N-2)f_{21}+(N-2)(N-3)(f_{31}+3f_{22})]
}{mN^4+(1-m) (10 N^3-35N^2+50N-24)},\\
%%%%
%%%%
&f_{311}=\frac{
(N-1)(N-2)m+(N-1)(N-2)(1-m)[f_{111}+3(N-3)f_{211}]
}{mN^4+(1-m) (10 N^3-35N^2+50N-24)},\\
%%%%
%%%%
&f_{221}=\frac{
(N-1)(N-2)m+(N-1)(N-2)(1-m)[f_{111}+2(N-3)f_{211}]
}{mN^4+(1-m) (10 N^3-35N^2+50N-24)},\\
%%%%
%%%%
&f_{2111}=\frac{(N-1)(N-2)(N-3)[m+(1-m)f_{1111}]}{mN^4+(1-m) (10 N^3-35N^2+50N-24)},\\
%%%%
%%%%
&f_{11111}=\frac{(N-1)(N-2)(N-3)(N-4)m}{mN^4+(1-m) (10 N^3-35N^2+50N-24)}.
%%%%
%%%%
\end{align}
\end{subequations}

\subsection*{Appendix I: Proof of Proposition \ref{Proposition7} }

In order to derive explicit expressions for $M(\mathbf{v}(x))$ and $Q(\mathbf{v}(x))$ in Theorem \ref{Theorem}, we will write the covariances in the expressions of $M(\mathbf{z})$ and $Q(\mathbf{z})$ given in Propositions \ref{Proposition1} and \ref{Proposition2} in terms of the identity-by-descent measures defined in Eq. (\ref{identitymeasures}) for offspring chosen at random in the same deme in an infinite neutral population in the stationary state $\mathbf{z}=\mathbf{v}(x)$ defined by Eq. (\ref{sec2-eq39}). Note that an offspring chosen at random and all offspring identical by descent to this offspring  are of type $S_1$ with probability $x$ and of type $S_2$ with probability $1-x$ independently of all offspring that are not identical by descent to this offspring, where $x$ is the constant frequency of $S_1$ in the neutral population at equilibrium. Moreover, offspring in different demes are necessarily non identical by descent since they cannot have a common ancestor in the same deme in the case of an infinite number of demes. 

By conditioning on the identity-by-descent status of two offspring $I_1$ and $I_2$ randomly chosen in the same deme, we have
\begin{flalign}\label{AA-eq1}
Cov_{\mathbf{v}(x)}\left(q_{1,I_1},q_{k,I_1}q_{l,I_2}\right)&=xf_{I_1|I_2}\left(\delta_{1k}x_{(l)}-x_{(k)}x_{(l)}\right)
+xf_{I_1I_2}\left(\delta_{1kl}-\delta_{kl}x_{(k)}\right)\nonumber&\\
&=x\left(\delta_{1k}-x_{(k)}\right)\left(\delta_{kl}f_{2}+f_{11}x_{(l)}\right),
\end{flalign}
where $x_{(k)}=x$ if $k=1$ and $1-x$ if $k=2$, while $\delta_{l_1l_2\ldots l_n}=1$ if $l_1=\cdots=l_n$ and $0$ otherwise.

Similarly, by conditioning on the identity-by-descent relationship between three offspring $I_1$, $I_2$ and $I_3$ randomly selected in the same deme, as well as between two offspring $J_1$ and $J_2$ independently chosen in a same deme, we get
\begin{flalign}\label{AA-eq2.0}
&Cov_{\mathbf{v}(x)}\left(q_{1,I_1},q_{k,I_2}q_{l,I_3}\right)\nonumber&\\
&\quad\quad\quad\quad=xf_{I_1I_2I_3}\left(\delta_{1kl}-\delta_{kl}x_{(k)}\right)+xf_{I_1I_2|I_3}\left(\delta_{1k}x_{(l)}-x_{(k)}x_{(l)}\right)+xf_{I_1I_3|I_2}\left(\delta_{1l}x_{(k)}-x_{(k)}x_{(l)}\right)\nonumber&\\
&\quad\quad\quad\quad=x\left(\delta_{1k}-x_{(k)}\right)\left(f_{3}\delta_{kl}+f_{21}x_{(l)}\right)
+xx_{(k)}f_{21}\left(\delta_{1l}-x_{(l)}\right),
\end{flalign}
\begin{flalign}\label{AA-eq2.1}
%%%%
%%%%
&Cov_{\mathbf{v}(x)}\left(q_{1,I_1},q_{k,I_1}q_{l,I_2}q_{k,J_1}q_{l,J_2}\right)\nonumber&\\
&\quad\quad\quad\quad=Cov_{\mathbf{v}(x)}\left(q_{1,I_1},q_{k,I_1}q_{l,I_2}\right)E_{\mathbf{v}(x)}\left(q_{k,J_1}q_{l,J_2}\right)\nonumber&\\
&\quad\quad\quad\quad=x\left(\delta_{1k}-x_{(k)}\right)\Big[f_{I_1I_2}\delta_{kl}+f_{I_1|I_2}x_{(l)}\Big]
\Big[f_{J_1J_2}\delta_{kl}x_{(k)}+f_{J_1|J_2}x_{(k)}x_{(l)}\Big]\nonumber&\\
&\quad\quad\quad\quad=x\left(\delta_{1k}-x_{(k)}\right)x_{(k)}\Big[f_{2}\delta_{kl}+f_{11}x_{(l)}\Big]^2,
\end{flalign}
\begin{flalign}\label{AA-eq2.3}
%%%%
%%%%
&Cov_{\mathbf{v}(x)}\left(q_{1,I_1},q_{k,I_2}q_{l,I_3}q_{k,J_1}q_{l,J_2}\right)\nonumber&\\
&\quad\quad\quad\quad=Cov_{\mathbf{v}(x)}\left(q_{1,I_1},q_{k,I_2}q_{l,I_3}\right)E_{\mathbf{v}(x)}\left(q_{k,J_1}q_{l,J_2}\right)\nonumber&\\
&\quad\quad\quad\quad=xx_{(k)}\Big[
\left(\delta_{1k}-x_{(k)}\right)\left(f_{3}\delta_{kl}+f_{21}x_{(l)}\right)
+x_{(k)}f_{21}\left(\delta_{1l}-x_{(l)}\right)
\Big] 
\Big[f_{2}\delta_{kl}+f_{11}x_{(l)}\Big],&
\end{flalign}
\begin{flalign}\label{AA-eq2.4}
%%%%
%%%%
&Cov_{\mathbf{v}(x)}\left(q_{1,I_1}q_{1,J_1},q_{k,I_1}q_{l,I_2}q_{k,J_1}q_{l,J_2}\right)\nonumber&\\
&\quad\quad\quad\quad=E_{\mathbf{v}(x)}\left(q_{1,I_1}q_{k,I_1}q_{l,I_2}\right)^2-x^2E_{\mathbf{v}(x)}\left(q_{k,I_1}q_{l,I_2}\right)^2\nonumber&\\
&\quad\quad\quad\quad=x^2\Big[f_{11}\delta_{1k}x_{(l)}+f_{2}\delta_{1kl}\Big]^2-x^2x_{(k)}^2\Big[f_{11}x_{(l)}+f_{2}\delta_{kl}\Big]^2\nonumber&\\
&\quad\quad\quad\quad=x^2\left(\delta_{1k}-x_{(k)}^2\right)\Big[f_{2}\delta_{kl}+f_{11}x_{(l)}\Big]^2,
\end{flalign}
\begin{flalign}\label{AA-eq2.5}
%%%%
%%%%
&Cov_{\mathbf{v}(x)}\left(q_{1,I_1},q_{1,J_1}q_{k,I_1}q_{l,I_2}q_{k,J_1}q_{l,J_2}\right)\nonumber&\\
&\quad\quad\quad\quad=Cov_{\mathbf{v}(x)}\left(q_{1,I_1},q_{k,I_1}q_{l,I_2}\right)E_{\mathbf{v}(x)}\left(q_{1,J_1}q_{k,J_1}q_{l,J_2}\right)\nonumber&\\
&\quad\quad\quad\quad=x\left(\delta_{1k}-x_{(k)}\right)\Big[f_{I_1I_2}\delta_{kl}+f_{I_1|I_2}x_{(l)}\Big] \Big[xx_{(l)}f_{I_1|I_2}\delta_{1k}+xf_{I_1I_2}\delta_{1kl}\Big],\nonumber&\\
&\quad\quad\quad\quad=x^2\delta_{1k}\left(\delta_{1k}-x_{(k)}\right)\Big[f_{2}\delta_{kl}+f_{11}x_{(l)}\Big]^2,
\end{flalign}
\begin{flalign}\label{AA-eq2.6}
%%%
%%%%
&Cov_{\mathbf{v}(x)}\left(q_{1,I_1},q_{1,J_1}q_{k,I_2}q_{l,I_3}q_{k,J_1}q_{l,J_2}\right)\nonumber&\\
&\quad\quad\quad\quad=
Cov_{\mathbf{v}(x)}\left(q_{1,I_1},q_{k,I_2}q_{l,I_3}\right)E_{\mathbf{v}(x)}\left(q_{1,J_1}q_{k,J_1}q_{l,J_2}\right)\nonumber&\\
&\quad\quad\quad\quad=x^2\delta_{1k}\Big[
\left(\delta_{1k}-x_{(k)}\right)\left(f_{3}\delta_{kl}+f_{21}x_{(l)}\right)
+x_{(k)}f_{21}\left(\delta_{1l}-x_{(l)}\right)
\Big]
%\nonumber&\\
%&\quad\quad\quad\quad\quad\times
\Big[f_{2}\delta_{kl}+f_{11}x_{(l)}\Big],&
\end{flalign}
%
%\begin{subequations}\label{AA-eq3}
\begin{flalign}\label{AA-eq2.7}
%%%%%
%%%%
&Cov_{\mathbf{v}(x)}\left(q_{1,I_1},q_{1,J_1}q_{k,I_1}q_{l,I_2}q_{k,J_2}q_{l,J_3}\right)\nonumber&\\
&\quad\quad\quad\quad=Cov_{\mathbf{v}(x)}\left(q_{1,I_1},q_{k,I_1}q_{l,I_2}\right)E_{\mathbf{v}(x)}\left(q_{1,J_1}q_{k,J_2}q_{l,J_3}\right)\nonumber&\\
&\quad\quad\quad\quad=x^2\left(\delta_{1k}-x_{(k)}\right)\left(f_{2}\delta_{kl}+f_{11}x_{(l)}\right)
\nonumber&\\
&\quad\quad\quad\quad\quad
\times\Big[\delta_{1k}\left(f_{3}\delta_{kl}+f_{21}x_{(l)}\right)+\left(f_{21}\delta_{1l}+f_{21}\delta_{kl}+f_{111}x_{(l)}\right)x_{(k)}\Big],
\end{flalign}
\begin{flalign}\label{AA-eq2.8}
%%%%%
%%%%
&Cov_{\mathbf{v}(x)}\left(q_{1,I_1},q_{1,J_1}q_{k,I_2}q_{l,I_3}q_{k,J_2}q_{l,J_3}\right)\nonumber&\\
&\quad\quad\quad\quad=Cov_{\mathbf{v}(x)}\left(q_{1,I_1},q_{k,I_2}q_{l,I_3}\right)E_{\mathbf{v}(x)}\left(q_{1,J_1}q_{k,J_2}q_{l,J_3}\right)\nonumber&\\
&\quad\quad\quad\quad=x^2\Big[\left(\delta_{1k}-x_{(k)}\right)\left(f_{3}\delta_{kl}+f_{21}x_{(l)}\right)
+x_{(k)}f_{21}\left(\delta_{1l}-x_{(l)}\right)\Big]\nonumber&\\
&\quad\quad\quad\quad\quad
\times
\Big[\delta_{1k}\left(f_{3}\delta_{kl}+f_{21}x_{(l)}\right)+\left(f_{21}\delta_{1l}+f_{21}\delta_{kl}+f_{111}x_{(l)}\right)x_{(k)}\Big].
\end{flalign}
Finally, by conditioning on the identity-by-descent relationship between four or five individuals randomly selected from the same deme, we obtain
\begin{flalign}\label{AA-eq4}
&Cov_{\mathbf{v}(x)}\left(q_{1,I_1},q_{k,I_1}q_{l,I_2}q_{k,I_3}q_{l,I_4}\right)\nonumber&\\
&\quad\quad\quad\quad=
x\left(\delta_{1k}-x_{(k)}\right)\Big[\delta_{kl}f_{4}+x_{(l)}f_{22}+x^2_{(l)}f_{211}\Big]\nonumber&\\
&\quad\quad\quad\quad\quad+x\left(\delta_{1k}x-x^2_{(k)}\right)\Big[\delta_{kl}\left(2f_{22}+4f_{31}\right)+x_{(l)}f_{211}+x^2_{(l)}f_{1111}\Big]\nonumber&\\
&\quad\quad\quad\quad\quad+4x\delta_{kl}\left(\delta_{1k}x^2-x^3_{(k)}\right)f_{211}
\end{flalign}
and
\begin{flalign}\label{AA-eq5}
&Cov_{\mathbf{v}(x)}\left(q_{1,I_1},q_{k,I_2}q_{k,I_3}q_{l,I_4}q_{l,I_5}\right)\nonumber&\\
&\quad\quad\quad\quad=x\left(\delta_{1k}-x_{(k)}\right)\left[\delta_{kl}f_{5}+x_{(l)}f_{32}+x_{(l)}^2f_{311}\right] +4x\delta_{kl}\left(\delta_{1k}x^2-x_{(k)}^3\right)\Big[2f_{221}+f_{311}\Big]  \nonumber&\\
&\quad\quad\quad\quad\quad+2x\left(\delta_{1k}x-x_{(k)}^2\right)\Big[\delta_{kl}\left(2f_{41}+4f_{32}\right)+x_{(l)}f_{221}+x_{(l)}^2f_{2111}\Big]\nonumber&\\
&\quad\quad\quad\quad\quad+2xx_{(k)}\left(\delta_{1l}x-x_{(l)}^2\right)\Big[f_{221}+x_{(k)}f_{2111}\Big]\nonumber&\\
&\quad\quad\quad\quad\quad+xx_{(k)}\left(\delta_{1l}-x_{(l)}\right)\Big[f_{32}+x_{(k)}f_{311}\Big].
\end{flalign}
%%%%%%%%%%%%%%%%%%%%%%%%%%%%%%%%%%%%%%%%%%%%%%%%%%%%%%%%%%%%%%%%%%%%%%%%%%%%%%%%%%%%%%%%%%%%%%%
%%%%%%%%%%%%%%%%%%%%%%%%%%%%%%%%%%%%%%%%%%%%%%%%%%%%%%%%%%%%%%%%%%%%%%%%%%%%%%%%%%%%%%%%%%%%%%%%%
%\subsection{Derivation of $M(x,\mathbf{0})$ and $Q(x,\mathbf{0})$ in terms of $f_{n_1,\ldots,n_k}$}
%%%%%%%%%%%%%%%%%%%%%%%%%%%%%%%%%%%%%%%%%%%%%%%%%%%%%%%%%%%%%%%%%%%%%%%%%%%%%%%%%%%%%%%%%%%%%%%
%%%%%%%%%%%%%%%%%%%%%%%%%%%%%%%%%%%%%%%%%%%%%%%%%%%%%%%%%%%%%%%%%%%%%%%%%%%%%%%%%%%%%%%%%%%%%%%%%
%\subsection*{Derivation of $M(\mathbf{v}(x))$}
Using the above expressions for the covariances and the relationships
\begin{subequations}\label{AA-eq6}
\begin{align}
&f_{2}=f_{3}+f_{21},\\
&f_{11}=2f_{21}+f_{111},\\
&f_{4}=f_{5}+f_{41},\\
&f_{31}=f_{41}+f_{32}+f_{311},\\
&f_{22}=2f_{32}+f_{221},\\
&f_{211}=f_{311}+2f_{221}+f_{2111},
\end{align}
\end{subequations}
the infinitesimal mean $M(\mathbf{v}(x))$ in Theorem 1, given by 
\begin{align}\label{}
&\sum_{k,l=1}^{2}\mu_{k,l}\Big[Cov_{\mathbf{v}(x)}\left(q_{1,I_1},q_{k,I_1}q_{l,I_2}\right)-(1-m)\, Cov_{\mathbf{v}(x)}\left(q_{1,I_1},q_{k,I_2}q_{l,I_3}\right)\Big]\nonumber
\\
&+\sum_{k,l=1}^{2}\sigma^2_{k,l}\Big[
(1-m)\Big(Cov_{\mathbf{v}(x)}\left(q_{1,I_1},q_{k,I_1}q_{k,I_2}q_{l,I_3}q_{l,I_4}\right)\nonumber\\
&\quad\quad\quad\quad\quad-Cov_{\mathbf{v}(x)}\left(q_{1,I_1},q_{k,I_2}q_{k,I_3}q_{l,I_4}q_{l,I_5}\right)\Big)-m\,  Cov_{\mathbf{v}(x)}\left(q_{1,I_1},q_{k,I_1}q_{l,I_2}q_{k,J_1}q_{l,J_2}\right)
\Big],
\end{align}
can be expressed as 
\begin{align}\label{AA-eq7}
& \sum_{k,l=1}^{2}\mu_{k,l}\Bigg[
x\left(\delta_{1k}-x_{(k)}\right)\left(\delta_{kl}f_{2}+f_{11}x_{(l)}\right)-
(1-m)\Big(x\left(\delta_{1k}-x_{(k)}\right)\left(f_{3}\delta_{kl}+f_{21}x_{(l)}\right)\nonumber\\
&+xx_{(k)}f_{21}\left(\delta_{1l}-x_{(l)}\right)\Big)\Bigg]+\sum_{k,l=1}^{2}\sigma^2_{k,l}\Bigg[(1-m)\Big(
x\left(\delta_{1k}-x_{(k)}\right)\Big[\delta_{kl}f_{4}+x_{(l)}f_{22}+x^2_{(l)}f_{211}\Big]\nonumber&\\
&+x(\delta_{1k}x-x_{(k)}^2)\Big[\delta_{kl}\left(2f_{22}+4f_{31}\right)+x_{(l)}f_{211}+x^2_{(l)}f_{1111}\Big]+4x\delta_{kl}(\delta_{1k}x^2-x_{(k)}^3)f_{211}\Big)\nonumber\\
&-(1-m)\Big(
x\left(\delta_{1k}-x_{(k)}\right)\Big[\delta_{kl}f_{5}+x_{(l)}f_{32}+x_{(l)}^2f_{311}\Big] +4x\delta_{kl}(\delta_{1k}x^2-x_{(k)}^3)[2f_{221}+f_{311}] \nonumber&\\
&+2x(\delta_{1k}x-x_{(k)}^2)\Big[\delta_{kl}\left(2f_{41}+4f_{32}\right)+x_{(l)}f_{221}
+x_{(l)}^2f_{2111}\Big]+xx_{(k)}(\delta_{1l}-x_{(l)})[f_{32}+x_{(k)}f_{311}]\nonumber\\
&+2xx_{(k)}(\delta_{1l}x-x_{(l)}^2)[f_{221}+x_{(k)}f_{2111}]
\Big) -mx\left(\delta_{1k}-x_{(k)}\right)x_{(k)}\left[f_{2}\delta_{kl}+f_{11}x_{(l)}\right]^2
\Bigg].
%%%
%%%
%&=x(1-x)\Bigg[\Big(f_{21}+xf_{111}\Big)\left(\mu_{1,1}-\mu_{2,1}\right)+\Big(f_{21}+(1-x)f_{111}\Big)\left(\mu_{1,2}-\mu_{2,2}\right)\nonumber\\
%&\quad-\Big(f_{41}+(4f_{311}+3f_{221})x+6f_{2111}x^2
%+f_{11111}x^3\Big)\sigma^2_{1,1}\nonumber\\
%&\quad-\Big(f_{32}+\left(3f_{221}+f_{2111}\right)(1-x)+xf_{311}+2x(1-x)f_{2111}+x(1-x)^2f_{11111}\Big)\sigma^2_{1,2}\nonumber\\
%&\quad+\Big(f_{32}+\left(3f_{221}+f_{2111}\right)x+f_{311}(1-x)+2x(1-x)f_{2111}+x^2(1-x)f_{11111}\Big)\sigma^2_{2,1}\nonumber\\
%&\quad+\Big(f_{41}+(4f_{311}+3f_{221})(1-x)+6f_{2111}(1-x)^2+f_{11111}(1-x)^3\Big)\sigma^2_{2,2}\Bigg]\nonumber\\
%&\quad+O\left(N^{-1}\right).
\end{align}
%%%%%%%%%%%%%%%%%%%%%%%%%%%%%%%%%%%%%%%%%%%%%%%%%%%%%%%%%%%%%%%%%%%%%%%%%%%%%%%%%%%%%%%%%%%%%%%
%%%%%%%%%%%%%%%%%%%%%%%%%%%%%%%%%%%%%%%%%%%%%%%%%%%%%%%%%%%%%%%%%%%%%%%%%%%%%%%%%%%%%%%%%%%%%%%%%
%\subsection*{Derivation of $Q(\mathbf{v}(x))$}
As for the infinitesimal variance $Q(\mathbf{v}(x))$ in Theorem 1, given by 
\begin{align}\label{}
& x(1-x)\left(1+(1-m)(Nm-1)f_2\right)\nonumber\\
&+
\sum_{k,l=1}^{2}\sigma^2_{k,l}\Bigg[
m^2 Cov_{\mathbf{v}(x)}\left(q_{1,I_1}q_{1,J_1},q_{k,I_1}q_{l,I_2}q_{k,J_1}q_{l,J_2}\right)
-2mx\,  Cov_{\mathbf{v}(x)}\left(q_{1,I_1},q_{k,I_1}q_{l,I_2}q_{k,J_1}q_{l,J_2}\right)
\nonumber\\
&+(1-m^2)\Big(
Cov_{\mathbf{v}(x)}\left(q_{1,I_1},q_{1,J_1}q_{k,I_1}q_{l,I_2}q_{k,J_1}q_{l,J_2}\right)
-
Cov_{\mathbf{v}(x)}\left(q_{1,I_1},q_{1,J_1}q_{k,I_2}q_{l,I_3}q_{k,J_1}q_{l,J_2}\right)
\Big)
\nonumber\\
&+
2m(1-m)\, Cov_{\mathbf{v}(x)}\left(q_{1,I_1},q_{k,I_2}q_{l,I_3}q_{k,J_1}q_{l,J_2}\right)
+
(1-m)^2\Big(
Cov_{\mathbf{v}(x)}\left(q_{1,I_1},q_{1,J_1}q_{k,I_2}q_{l,I_3}q_{k,J_2}q_{l,J_3}\right)\nonumber\\
&
-
Cov_{\mathbf{v}(x)}\left(q_{1,I_1},q_{1,J_1}q_{k,I_1}q_{l,I_2}q_{k,J_2}q_{l,J_3}\right)
\Big)\Bigg],
\end{align}
it takes the form
\begin{align}\label{AA-eq9}
& x(1-x)\left(1+(1-m)(Nm-1)f_2\right)\nonumber\\
&+
\sum_{k,l=1}^{2}\sigma^2_{k,l}\Bigg[
m^2x^2(\delta_{1k}-x_{(k)}^2)\left(f_{2}\delta_{kl}+f_{11}x_{(l)}\right)^2
-2mx^2\left(\delta_{1k}-x_{(k)}\right)x_{(k)}\left(f_{2}\delta_{kl}+f_{11}x_{(l)}\right)^2\nonumber\\
&+\delta_{1k}x^2(1-m^2)(f_{2}\delta_{kl}+f_{11}x_{(l)})
\Big(\left(\delta_{1k}-x_{(k)}\right)(f_{21}\delta_{kl}+(f_{21}+f_{111})x_{(l)})
-x_{(k)}f_{21}\left(\delta_{1l}-x_{(l)}\right)
\Big)\nonumber\\
&+ 2m(1-m)xx_{(k)}(f_{2}\delta_{kl}+f_{11}x_{(l)})\Big(
\left(\delta_{1k}-x_{(k)}\right)\left(f_{3}\delta_{kl}+f_{21}x_{(l)}\right)
+x_{(k)}f_{21}\left(\delta_{1l}-x_{(l)}\right)
\Big)
\nonumber\\
&+x^2(1-m)^2\Big(\delta_{1k}\left(f_{3}\delta_{kl}+f_{21}x_{(l)}\right)+\left(f_{21}\delta_{1l}+f_{21}\delta_{kl}+f_{111}x_{(l)}\right)x_{(k)}\Big)\nonumber\\
&\times\Big(
x_{(k)}f_{21}(\delta_{1l}-x_{(l)})-(\delta_{1k}-x_{(k)})(f_{21}\delta_{kl}+(f_{21}+f_{111})x_{(l)})
\Big)
\Bigg].
\end{align}
Here, besides the previous expressions for the covariances, we have used the identities
\begin{subequations}\label{AA-eq8}
\begin{align}
&Var_{\mathbf{v}(x)}(q_{I_1})=x(1-x),\\
&Cov_{\mathbf{v}(x)}(q_{1,I_1},q_{1,I_2})=x(1-x)f_{2}.
\end{align}
\end{subequations}
Finally, letting $N\rightarrow\infty$ and $m\rightarrow 0$ in a way such that $Nm\rightarrow\nu$, we get the approximations
\begin{align}\label{AA-eq10}
M(\mathbf{v}(x))
&\approx x(1-x)\Bigg[\Big(f_{21}+xf_{111}\Big)\left(\mu_{1,1}-\mu_{2,1}\right)+\Big(f_{21}+(1-x)f_{111}\Big)\left(\mu_{1,2}-\mu_{2,2}\right)\nonumber\\
&\quad-\Big(f_{41}+(4f_{311}+3f_{221})x+6f_{2111}x^2
+f_{11111}x^3\Big)\sigma^2_{1,1}\nonumber\\
&\quad-\Big(f_{32}+\left(3f_{221}+f_{2111}\right)(1-x)+xf_{311}+2x(1-x)f_{2111}+x(1-x)^2f_{11111}\Big)\sigma^2_{1,2}\nonumber\\
&\quad+\Big(f_{32}+\left(3f_{221}+f_{2111}\right)x+f_{311}(1-x)+2x(1-x)f_{2111}+x^2(1-x)f_{11111}\Big)\sigma^2_{2,1}\nonumber\\
&\quad+\Big(f_{41}+(4f_{311}+3f_{221})(1-x)+6f_{2111}(1-x)^2+f_{11111}(1-x)^3\Big)\sigma^2_{2,2}\Bigg]%\nonumber\\
%&\quad+O\left(N^{-1}\right)
\end{align}
and
\begin{align}\label{AA-eq11}
Q(\mathbf{v}(x))
&\approx x(1-x)(f_{11}+\nu f_{2})+x^2(1-x)^2\left(f_{21}+xf_{111}\right)^2\left(\sigma^2_{1,1}+\sigma^2_{2,1}\right)\nonumber\\
&\quad+x^2(1-x)^2\left(f_{21}+(1-x)f_{111}\right)^2\left(\sigma^2_{1,2}+\sigma^2_{2,2}\right).%+O\left(N^{-1}\right).
\end{align}

%%%%%%%%%%%%%%%%%%%%%%%%%%%%%%%%%%%%%%%%%%%%%%%%%%%%%%%%%%%%%%%%%%%%%%%%%%%%%%%%%%%%%%%%%%%%%%%\\
%%%%%%%%%%%%%%%%%%%%%%%%%%%%%%%%%%%%%%%%%%%%%%%%%%%%%%%%%%%%%%%%%%%%%%%%%%%%%%%%%%%%%%%%%%%%%%%
\subsection*{Appendix J: Proof of Proposition \ref{Proposition8}}

In a neutral population of an infinite number of demes of size $N$ with $N$ generations as unit of time as $N$ approaches infinity and $m$ approaches zero such that $Nm\rightarrow\nu$, each pair of lineages of offspring within the same deme coalesces  backward in time at the rate $1$ independently of all the other pairs of lineages like in the Kingmann coalescent (Kingmann \cite{K1982}).  On the other hand, each deme goes extinct backward in time at the rate  $\nu$. When this occurs, the lineages within the deme migrate to all different demes and stay in different demes afterward so that they cannot coalesce anymore. Therefore, offspring are identical by descent only if they are in the same deme and their lineages coalesce backward in time before the deme goes extinct (see, e.g., Hudson \cite{H1998} and Rousset \cite{R2002}).

We are interested in the probability for an ordered sample of offspring selected at random from the same deme, represented by $I^{(1)}_1,\ldots,I^{(1)}_{n_1},\ldots,I^{(k)}_1,\ldots, I^{(k)}_{n_k}$, to be such that :
\begin{itemize}
\item $I^{(i)}_1,\ldots,I^{(i)}_{n_i}$ are identical by descent for $i=1,\ldots,k$, and
\item $I^{(i)}_1$ and $I^{(j)}_1$ are not identical by descent if $i\neq j$ for $i,j=1,\ldots,k$.
\end{itemize}
This event is represented by $I^{(1)}_1\ldots \, I^{(1)}_{n_1}|\ldots \,|I^{(k)}_1\ldots \,  I^{(k)}_{n_k}$ and its probability by $f_{I^{(1)}_1\ldots \, I^{(1)}_{n_1}|\ldots \,|I^{(k)}_1\ldots \, I^{(k)}_{n_k}}$.

The offspring $I_1,\ldots, I_{n}$  will be identical by descent if and only if all their lineages coalesce before the deme goes extinct. Note that the probability of a coalescence event involving two lineages among $j$ in the same deme before an extinction event involving the deme occurs is given by the rate of coalescence, $\binom{j}{2}$, divided by the rate of coalescence or extinction, $\binom{j}{2}+\nu$, 
for $j=2, \ldots, n$. Therefore, we have
\begin{equation}\label{B-eq1}
f_{I_1\ldots  \,I_{n}}=\prod_{j=2}^{n}\frac{\binom{j}{2}}{\binom{j}{2}+\nu}.
\end{equation}
On the other hand, none of the offspring 
$I_1,\ldots, I_{n}$ will be identical by descent to another if and only if the first event to occur is the extinction of the deme, which gives
\begin{equation}\label{B-eq2}
f_{I_1|\ldots \, |I_{n}}=\frac{\nu}{\binom{n}{2}+\nu}.
\end{equation}
More generally, the occurrence of the event $I^{(1)}_1\ldots \,I^{(1)}_{n_1}|\ldots \, |I^{(k)}_1\cdots I^{(k)}_{n_k}$ can be explained by a series of $n_1+\cdots+n_k-k$ coalescent events followed by an extinction event. Specifically, $n_i-1$ coalescence events must occur among the ancestral lines of $I^{(i)}_1,\ldots,I^{(i)}_{n_i}$ until they reach their most recent common ancestor, denoted by $J^{(i)}$, for $i=1,\ldots,k$, before the deme goes extinct. 
Moreover, the ancestral lines of $J^{(i_1)}$ and $J^{(i_2)}$ must never experience a coalescence event, for $i_1\not=i_2$, before the deme goes extinct. As a result, we obtain
\begin{equation}\label{B-eq3}
f_{I^{(1)}_1\ldots \, I^{(1)}_{n_1}|\ldots \, |I^{(k)}_1\ldots \, I^{(k)}_{n_k}}=\binom{n_1+\cdots+n_k-k}{n_1-1,\ldots,n_k-1}\frac{\prod_{i=1}^{k}\prod_{j=2}^{n_i}\binom{j}{2}}
{\prod_{j=k+1}^{n_1+\cdots+n_k}\left(\binom{j}{2}+\nu\right)}\times\frac{\nu}{\binom{k}{2}+\nu},
\end{equation}
where
\begin{equation}\label{B-eq4}
\binom{n_1+\cdots+n_k-k}{n_1-1,\ldots,n_k-1}=\dfrac{(n_1+\cdots+n_k-k)!}{(n_1-1)!\cdots( n_k-1)!},
\end{equation}
for $n_1, \ldots, n_k \geq 1$ for $k\geq2$.

%%%%%%%%%%%%%%%%%%%%%%%%%%%%%%%%%%%%%%%%%%%%%%%%%%%%%%%%%%%%%%%%%%%%%%%%%%%%%%%%%%%
%%%%%%%%%%%%%%%%%%%%%%%%%%%%%%%%%%%%%%%%%%%%%%%%%%%%%%%%%%%%%%%%%%%%%%%%%%%%%%%%%%%%%%%%%%%%%%%
%%%%%%%%%%%%%%%%%%%%%%%%%%%%%%%%%%%%%%%%%%%%%%%%%%%%%%%%%%%%%%%%%%%%%%%%%%%%%%%%%%%%%%%%%%%%%%%

\bibliographystyle{unsrt}

\end{document}